\newif\ifappendix
\appendixtrue
\documentclass[
acmsmall
,screen=true
]{acmart}

\makeatletter
\def\@acmplainindent{0pt}
\def\@acmdefinitionindent{0pt}
\def\@proofindent{\noindent}
\makeatother

\usepackage{bbold}
\usepackage{mathtools}
\usepackage{stmaryrd,mathpartir,amsthm}
\usepackage{supertabular}
\usepackage{mathpartir}
\usepackage{graphicx}
\usepackage[T1]{fontenc}
\usepackage{cleveref}
\usepackage{enumitem}
\usepackage{cancel}
\usepackage{wrapfig}
\usepackage{colortbl}
\usepackage{subcaption}
\usepackage{xspace}
\usepackage{csquotes}
\usepackage{thm-restate}
\usepackage{tabularx}
\usepackage{array}
\usepackage{longtable}
\usepackage{trimclip}
\usepackage[disable]{todonotes}

\graphicspath{ {./} }

\AtBeginDocument{
  }

\setcopyright{rightsretained}
\acmDOI{10.1145/3808318}
\acmYear{2026}
\acmJournal{PACMPL}
\acmVolume{10}
\acmNumber{PLDI}
\acmArticle{240}
\acmMonth{6}
\acmSubmissionID{pldi26main-p501-p}
\received{2025-11-14}
\received[accepted]{2026-04-03}



\author{Emmanuel {Su\'arez Acevedo}}
\orcid{0009-0002-5515-6099}
\email{es2278@cornell.edu}
\affiliation{
  \institution{Cornell University}
  \city{Ithaca}
  \country{USA}
}

\author{Tiago Ferreira}
\orcid{0000-0002-6942-0228}
\affiliation{
  \institution{University College London}
  \city{London}
  \country{United Kingdom}
}
\email{t.ferreira@ucl.ac.uk}

\author{Kevin Batz}
\email{ksb239@cornell.edu}
\orcid{0000-0001-8705-2564}
\affiliation{
  \city{Ithaca}
  \institution{Cornell University}
  \country{USA}
}

\author{Oliver B{\o}ving}
\email{oembo@dtu.dk}
\orcid{0009-0009-4702-2876}
\affiliation{
  \institution{Technical University of Denmark}
  \city{Kongens Lyngby}
  \country{Denmark}
}

\author{Nate Foster}
\orcid{0000-0002-6557-684X}
\affiliation{
  \institution{EPFL}
  \city{Lausanne}
  \country{Switzerland}
}
\affiliation{
  \institution{Jane Street}
  \city{New York}
  \country{USA}
}
\email{nate.foster@epfl.ch}

\author{Alexandra Silva}
\orcid{0000-0001-5014-9784}
\affiliation{
  \institution{Cornell University}
  \city{Ithaca}
  \country{USA}
}
\email{alexandra.silva@cornell.edu}

\definecolor{lightgray}{RGB}{245,245,245}
\definecolor{headergray}{RGB}{220,220,220}

\definecolor{commentgreen}{HTML}{008000}

\definecolor{latencygreen}{RGB}{27,158,119}
\definecolor{failureorange}{RGB}{217,95,2}
\definecolor{bandwidthpurple}{RGB}{117,112,179}
\definecolor{cyclepink}{RGB}{231,41,138}
\definecolor{tunnelgreen}{HTML}{66a61e}
\definecolor{tunnelred}{HTML}{e7298a}
\definecolor{tunnelgold}{HTML}{e6ab02}
\definecolor{tunnelbrown}{HTML}{a6761d}
\definecolor{tunnelblack}{HTML}{666666}

\newtheorem{theorem}{Theorem}

\newtheorem{definition}{Definition}
\newtheorem{lemma}{Lemma}
\newtheorem{corollary}{Corollary}
\newtheorem{example}{Example}

\newtheorem{remark}{Remark}
\newtheorem{notation}{Notation}

\newcommand{\repeattheorem}[1]{
  \begingroup
  \renewcommand*{\thetheorem}{\ref{#1}}
  \expandafter\expandafter\expandafter\theorem
  \csname rtheorem@#1\endcsname
  \endtheorem
  \endgroup
}
\NewEnviron{rtheorem}[1]{
  \global\expandafter\xdef\csname rtheorem@#1\endcsname{
    \unexpanded\expandafter{\BODY}
  }
  \expandafter\theorem\BODY\unskip\label{#1}\endtheorem
}

\newcommand{\repeatlemma}[1]{
  \begingroup
  \renewcommand*{\thelemma}{\ref{#1}}
  \expandafter\expandafter\expandafter\lemma
  \csname rlemma@#1\endcsname
  \endlemma
  \endgroup
}
\NewEnviron{rlemma}[1]{
  \global\expandafter\xdef\csname rlemma@#1\endcsname{
    \unexpanded\expandafter{\BODY}
  }
  \expandafter\lemma\BODY\unskip\label{#1}\endlemma
}

\newcommand{\repeatproposition}[1]{
  \begingroup
  \renewcommand*{\theproposition}{\ref{#1}}
  \expandafter\expandafter\expandafter\proposition
  \csname rproposition@#1\endcsname
  \endproposition
  \endgroup
}
\NewEnviron{rproposition}[1]{
  \global\expandafter\xdef\csname rproposition@#1\endcsname{
    \unexpanded\expandafter{\BODY}
  }
  \expandafter\proposition\BODY\unskip\label{#1}\endproposition
}

\newcommand{\repeatcorollary}[1]{
  \begingroup
  \renewcommand*{\thecorollary}{\ref{#1}}
  \expandafter\expandafter\expandafter\corollary
  \csname rcorollary@#1\endcsname
  \endproposition
  \endgroup
}
\NewEnviron{rcorollary}[1]{
  \global\expandafter\xdef\csname rcorollary@#1\endcsname{
    \unexpanded\expandafter{\BODY}
  }
  \expandafter\corollary\BODY\unskip\label{#1}\endcorollary
}

\renewenvironment{definition}{\begin{olddefn}}{\unskip\nobreak\hfill$\triangle$\par\medskip \end{olddefn}}

\DeclareMathAlphabet\mathbfcal{OMS}{cmsy}{b}{n}

\newcommand{\appendixref}[1][]{
  \ifappendix
      \ifx\relax#1\relax
        the appendix
      \else
        \Cref{#1}
      \fi
  \else
    the appendix~\cite{extended}
  \fi
}

\newcommand{\kat}{\textsf{KAT}\xspace}
\newcommand{\netkat}{\textsf{NetKAT}\xspace}
\newcommand{\probnetkat}{\textsf{ProbNetKAT}\xspace}
\newcommand{\wnetkat}{\textsf{wNetKAT}\xspace}
\newcommand{\iwnetkat}[1]{{#1}-\wnetkat}
\newcommand{\swnetkat}{\iwnetkat{$\semi$}}
\newcommand{\wgkat}{\textsf{wGKAT}\xspace}
\newcommand{\gkat}{\textsf{GKAT}\xspace}
\newcommand{\kawt}{\textsf{KAWT}\xspace}

\usepackage{tikz}
\usetikzlibrary{automata, positioning, arrows.meta, calc, decorations.pathmorphing, matrix, fit}
\definecolor{darkgreen}{rgb}{0.0, 0.5, 0.0}
\tikzstyle{smallstate}=[inner sep=1.2pt,draw,circle,fill]
\tikzstyle{input-edge}=[-{Stealth[length=2mm,width=1.5mm]}, shorten >=0.8mm]
\tikzstyle{output-edge}=[-Implies, double distance=1pt, shorten <=0.8mm]
\tikzstyle{transition-edge}=[-{Stealth[length=2mm,width=1.5mm]},shorten >=0.8mm, shorten <=1.2mm]
\tikzset{every loop/.style={transition-edge, looseness=20}}
\tikzstyle{output}=[draw=none, node distance=.6cm, inner sep=0pt,yshift=1pt]

\newcommand{\cpo}{\mathcal{D}}
\newcommand{\cpodom}{\ensuremath{D}}
\newcommand{\cpoa}{\ensuremath{d}}
\newcommand{\cpoord}[1][]{\ensuremath{\sqsubseteq_{#1}}}
\newcommand{\cposup}[1][]{\ensuremath{\bigsqcup\limits_{\substack{#1}}}}

\newcommand{\cpobot}{\bot}
\newcommand{\cpochain}{\mathcal{C}}
\newcommand{\cposubset}{\ensuremath{S}}
\newcommand{\cpoendo}{\ensuremath{f}}

\newcommand{\lfp}{\ensuremath{\mathsf{lfp}~}}

\newcommand{\mon}{\mathcal{M}}
\newcommand{\mona}{\ensuremath{s}}
\newcommand{\mondom}{\ensuremath{M}}
\newcommand{\semimul}[1][]{\ensuremath{\cdot_{#1}}}
\newcommand{\mulid}[1][]{\ensuremath{\mathbb{1}_{#1}}}

\newcommand{\semi}{\ensuremath{\mathcal{S}}}
\newcommand{\semia}{\ensuremath{s}}
\newcommand{\semib}{\ensuremath{r}}
\newcommand{\semidom}{\ensuremath{S}}
\newcommand{\semiadd}[1][]{\ensuremath{+_{#1}}}

\newcommand{\addid}[1][]{\ensuremath{\mathbb{0}_{#1}}}
\newcommand{\semisum}[1]{\ensuremath{\sum{{#1}}}}
\newcommand{\semisumtop}[2]{\ensuremath{\sum{{#1}}^{#2}}}

\newcommand{\semistar}[1]{\ensuremath{{#1}^{*}}}

\newcommand{\laddid}[1][]{\addid{#1}}
\newcommand{\lsemimul}[1][]{\semimul{#1}}
\newcommand{\lsemiadd}[1][]{\semiadd{#1}}
\newcommand{\lsemiprod}[1][]{\ensuremath{\mathbin{\times_{#1}}}}
\newcommand{\lsemisum}[1]{\semisum{#1}}
\newcommand{\lsemisumtop}[2]{\semisumtop{#1}{#2}}

\newcommand{\semiord}[1][]{\ensuremath{\mathrel{\preceq}_{#1}}}
\newcommand{\semiordrev}[1][]{\ensuremath{\mathrel{\succeq}_{#1}}}
\newcommand{\semiordrevstrict}[1][]{\ensuremath{\mathrel{\succ}_{#1}}}

\newcommand{\semisup}[1][]{\ensuremath{\bigsqcup\limits_{\substack{#1}}}}
\newcommand{\semichain}{\mathcal{C}}

\newcommand{\indexset}{\ensuremath{I}}
\newcommand{\indexsetb}{\ensuremath{J}}
\newcommand{\indexx}{\ensuremath{i}}
\newcommand{\enum}{\mathsf{enum}}

\DeclareMathOperator*{\BigAdd}{\bigoplus}
\newcommand{\llb}{\llbracket}
\newcommand{\rrb}{\rrbracket}
\newcommand{\llp}{\llparenthesis}
\newcommand{\rrp}{\rrparenthesis}
\newcommand{\bind}{\gg\!=}
\newcommand{\unit}{\eta}
\newcommand*{\longeq}{:\Longleftrightarrow}
\newcommand{\eeq}{~{}={}~}
\newcommand{\ttriangleq}{~{}\triangleq{}~}
\newcommand{\nats}{\mathbb{N}}
\newcommand{\reals}{\ensuremath{\mathbb{R}}}
\newcommand{\posreals}{\ensuremath{\reals_{\geq 0}}}
\newcommand{\posrealsinf}{\ensuremath{\posreals^{\infty}}}
\newcommand{\mylam}[1]{\ensuremath{\lambda #1.\,}}
\newcommand{\setsof}[1]{\ensuremath{2^{#1}}}
\let\oldsum\sum
\renewcommand{\sum}[1]{\ensuremath{\oldsum\limits_{\substack{#1}}}}

\newcommand{\probsum}[2]{\ensuremath{#1 \uplus #2}}

\newcommand{\boolSemiring}{\ensuremath{2}}
\newcommand{\probSemiring}{\ensuremath{\mathsf{Prob}}}
\newcommand{\nonnegRealsSemiring}{\realsSemiring}
\newcommand{\realsSemiring}{\ensuremath{\mathsf{Real}}}

\newcommand{\PublicSecurity}{\ensuremath{0}}
\newcommand{\LowSecurity}{\ensuremath{\text{L}}}
\newcommand{\MedSecurity}{\ensuremath{\text{M}}}
\newcommand{\HighSecurity}{\ensuremath{\text{H}}}

\newcommand{\wta}{\ensuremath{\semib}}
\newcommand{\wtb}{\ensuremath{\semia}}

\newcommand{\histories}{\ensuremath{\mathsf{H}}}
\newcommand{\hista}{\ensuremath{h}}
\newcommand{\histb}{\ensuremath{h'}}
\newcommand{\histc}{\ensuremath{h''}}
\newcommand{\cons}{{::}}
\newcommand{\hcons}[2]{#1\cons#2}
\newcommand{\hempty}{\langle\rangle}

\newcommand{\packets}{\mathsf{Pk}}
\newcommand{\pkt}{\ensuremath{\pi}}
\newcommand{\pkta}{\ensuremath{\alpha}}
\newcommand{\pktb}{\ensuremath{\beta}}
\newcommand{\pktc}{\ensuremath{\gamma}}
\newcommand{\pktd}{\ensuremath{\xi}}

\newcommand{\updatepkt}[3]{\ensuremath{#1 [ #2 \coloneq #3 ]}}

\newcommand{\ctest}[1]{\ensuremath{#1 ?}}
\newcommand{\cass}[1]{\ensuremath{#1 !}}
\newcommand{\Test}{\ensuremath{\packets^?}}
\newcommand{\Assign}{\ensuremath{\packets^!}}

\newcommand{\oldctestA}{\ensuremath{\alpha}}
\newcommand{\oldctestB}{\ensuremath{\beta}}
\newcommand{\oldctestC}{\ensuremath{\gamma}}

\newcommand{\oldcassA}{\ensuremath{\pi}}

\newcommand{\oldcass}[1]{\ensuremath{\oldcassA_{#1}}}
\newcommand{\oldTest}{\ensuremath{\mathsf{At}}}

\newcommand{\GS}{\ensuremath{\mathsf{GS}}}
\newcommand{\gsA}{\ensuremath{x}}
\newcommand{\gsB}{\ensuremath{y}}
\newcommand{\gsC}{\ensuremath{z}}

\newcommand{\wtinga}{\ensuremath{m}}
\newcommand{\wtingb}{\ensuremath{m^\prime}}
\newcommand{\supp}{\ensuremath{\mathsf{supp}}}
\newcommand{\mass}[1]{\ensuremath{|#1|}}

\newcommand{\pols}{\ensuremath{\mathsf{Pol}}}
\newcommand{\pola}{\ensuremath{p}}
\newcommand{\polb}{\ensuremath{q}}
\newcommand{\pol}[1]{\ensuremath{\pola_{#1}}}

\newcommand{\linka}{\ensuremath{l}}

\newcommand{\topoa}{\linka}

\newcommand{\tests}{\ensuremath{\mathsf{Pred}}}
\newcommand{\testa}{\ensuremath{t}}
\newcommand{\testb}{\ensuremath{u}}

\newcommand{\histsa}{\ensuremath{a}}
\newcommand{\histsb}{\ensuremath{b}}
\newcommand{\proba}{\ensuremath{r}}
\newcommand{\eventa}{\ensuremath{A}}

\newcommand{\measurea}{\ensuremath{\mu}}

\newcommand{\markova}{\ensuremath{P}}
\newcommand{\borelsets}{\mathcal{B}}
\newcommand{\discreteprobmonad}{\mathcal{D}}
\newcommand{\probmonad}{\mathcal{M}}
\newcommand{\measures}[1]{\discreteprobmonad(#1)}

\newcommand{\fields}{\mathsf{F}}
\newcommand{\fielda}{\ensuremath{f}}

\newcommand{\values}{\mathsf{Val}}
\newcommand{\vala}{\ensuremath{n}}

\newcommand{\concreteval}{\ensuremath{v}}

\newcommand{\fieldsw}{\ensuremath{\mathsf{sw}}}
\newcommand{\fielddst}{\ensuremath{\mathsf{dst}}}
\newcommand{\fieldpt}{\ensuremath{\mathsf{pt}}}
\newcommand{\fieldvid}{\ensuremath{\mathsf{vid}}}
\newcommand{\switcha}{\ensuremath{S}}
\newcommand{\hosta}{\ensuremath{H}}
\newcommand{\millis}{\ensuremath{\mathsf{ms}}}
\newcommand{\mbps}{\ensuremath{\mathsf{Mbps}}}

\newcommand{\consttrue}{\ensuremath{\mathsf{TRUE}}}

\newcommand{\red}{\ensuremath{{\downarrow}}}
\newcommand{\rpols}{\ensuremath{\mathsf{Pol}^\red}}

\newcommand{\DROP}{\ensuremath{\mathsf{drop}}}
\newcommand{\SKIP}{\ensuremath{\mathsf{skip}}}
\newcommand{\TRUE}{\ensuremath{\mathsf{true}}}
\newcommand{\FALSE}{\ensuremath{\mathsf{false}}}
\newcommand{\AND}[2]{\ensuremath{#1 \, \wedge \, #2}}
\newcommand{\OR}[2]{\ensuremath{#1 \, \vee \, #2}}
\newcommand{\EQ}[2]{\ensuremath{#1 = #2}}
\newcommand{\NOT}[1]{\ensuremath{\neg #1}}
\newcommand{\ASSN}[2]{\ensuremath{#1 \leftarrow #2}}
\newcommand{\SEQ}[2]{\ensuremath{#1 \SEQN #2}}
\newcommand{\SEQN}{\ensuremath{\, ;}}
\newcommand{\ITER}[1]{\ensuremath{#1 ^{*}}}
\newcommand{\DUP}{\ensuremath{\mathsf{dup}}}
\newcommand{\ADD}[2]{\ensuremath{#1 \ADDN #2}}
\newcommand{\ADDN}{\ensuremath{\oplus}}
\newcommand{\WEIGH}[2]{\ensuremath{#1 \odot #2}}

\newcommand{\IFN}{\mathsf{if}}
\newcommand{\THENN}{\mathsf{then}}
\newcommand{\ELSEN}{\mathsf{else}}

\newcommand{\SUM}[2]{\ensuremath{\BigAdd\limits_{\substack{#1}} #2}}
\newcommand{\SUMTOP}[3]{\ensuremath{\BigAdd\limits_{\substack{#1}}^{#2} #3}}
\newcommand{\NOTEQ}[2]{\ensuremath{#1 \neq #2}}
\newcommand{\IF}[3][\testA]{\ensuremath{\IFN \, #1 \, \THENN \, #2 \, \ELSEN \, #3}}
\newcommand{\WHILE}[2][\testA]{\ensuremath{\mathsf{while} \, #1 \, \mathsf{do} \, #2 }}
\newcommand{\NFOLD}[2][n]{\ensuremath{{#2}^{(#1)}}}
\newcommand{\ALTNFOLD}[2][n]{\ensuremath{{#2}^{\langle #1 \rangle}}}

\newcommand{\PROBCHOICE}[3][\probA]{\ensuremath{#2 \oplus_{#1} #3}}
\newcommand{\PAR}[2]{\ensuremath{#1 \,\&\, #2}}

\newcommand{\transpol}[1]{#1}

\newcommand{\approxpol}[3][]{\approxn[#3]{\left[ #2 \right]}}
\newcommand{\altapproxpol}[3][]{\approxn[#3]{\left( #2 \right)}}
\newcommand{\altaltapproxpol}[3][]{\approxn[#3]{\langle #2 \rangle}}
\newcommand{\approxn}[2][n]{#2_{#1}}

\newcommand{\WNKA}{\textsf{WNKA}\xspace}
\newcommand{\wnkapath}{\ensuremath{\mathsf{Runs}}}
\newcommand{\wnka}{\ensuremath{\mathcal{A}}}

\newcommand{\wnkastates}{\ensuremath{Q}}
\newcommand{\wnkastatea}{\ensuremath{q}}
\newcommand{\wnkastateb}{\ensuremath{q'}}
\newcommand{\wnkastatec}{\ensuremath{q''}}
\newcommand{\wnkainit}{\ensuremath{\iota}}
\newcommand{\wnkatrans}[2]{\ensuremath{\delta_{#1 #2}}}
\newcommand{\wnkaout}[2]{\ensuremath{\lambda_{#1 #2}}}
\newcommand{\wnkatransext}[1]{\ensuremath{\delta_{#1}^*}}
\newcommand{\wnkaoutext}[1]{\ensuremath{\lambda_{#1}^*}}

\newcommand{\wfastartstate}{\ensuremath{q_\iota}}
\newcommand{\wfaendstate}{\ensuremath{q_\lambda}}
\newcommand{\wfastates}{\wnkastates'}
\newcommand{\wfastatea}{q}
\newcommand{\wfastateb}{q'}
\newcommand{\wfainit}{\ensuremath{I}}
\newcommand{\wfatrans}[1]{\ensuremath{\Delta_{#1}}}
\newcommand{\wfaout}{\ensuremath{\Lambda}}
\newcommand{\wfatransext}[1]{\ensuremath{\Delta_{#1}^*}}

\newcommand{\weightingsfunctor}[1][\semi]{\ensuremath{\mathbfcal{W}_{#1}}}
\newcommand{\weightings}[2][\semi]{\weightingsfunctor[#1](#2)}
\newcommand{\sem}[1]{\llb #1 \rrb}
\newcommand{\altsem}[1]{\llp #1 \rrp}

\newcommand{\eword}{\ensuremath{\varepsilon}}

\newcommand{\wlang}[1]{\llbracket #1 \rrbracket}

\newcommand{\concatword}[2]{{#1}\cdot{#2}}
\newcommand{\wrun}{\rho}
\newcommand{\wruns}[2]{\mathsf{Runs}_{#1 \to #2}}
\newcommand{\wrunscfree}[2]{\mathsf{Runs}_{#1 \to #2}^{\cancel{\circlearrowleft}}}
\newcommand{\idmatrix}{\mathsf{id}}
\newcommand{\runweight}{\textnormal{weight}}

\newcommand{\semithresh}{\semib}
\newcommand{\autsafe}[1]{\textnormal{{$#1$}-SAFE}}
\newcommand{\autreach}[1]{\textnormal{{$#1$}-REACHABLE}}

\newcommand{\edge}{\ensuremath{\shortrightarrow}}
\newcommand{\fieldnode}{\ensuremath{\mathsf{node}}}

\newcommand{\failure}[1]{\textcolor{failureorange}{\ensuremath{#1\%}}}
\newcommand{\latency}[1]{\textcolor{latencygreen}{\ensuremath{#1 \millis}}}
\newcommand{\bandwidth}[1]{\textcolor{bandwidthpurple}{\ensuremath{#1 \ \mbps}}}
\newcommand{\abilene}{\ensuremath{\mathit{abilene}}}

\newcommand{\seattle}{\ensuremath{\mathsf{SEA}}}
\newcommand{\bay}{\ensuremath{\mathsf{BAY}}}
\newcommand{\la}{\ensuremath{\mathsf{LA}}}
\newcommand{\denver}{\ensuremath{\mathsf{DEN}}}
\newcommand{\kansas}{\ensuremath{\mathsf{KAN}}}
\newcommand{\houston}{\ensuremath{\mathsf{HOU}}}
\newcommand{\atlanta}{\ensuremath{\mathsf{ATL}}}
\newcommand{\chicago}{\ensuremath{\mathsf{CHI}}}
\newcommand{\nyc}{\ensuremath{\mathsf{NYC}}}
\newcommand{\dc}{\ensuremath{\mathsf{DC}}}
\newcommand{\indiana}{\ensuremath{\mathsf{IND}}}
\newcommand{\relstr}{\ensuremath{\mathsf{rel}}}
\newcommand{\bandstr}{\ensuremath{\mathsf{band}}}
\newcommand{\vidstr}{\ensuremath{\mathsf{vid}}}
\newcommand{\defaultstr}{\ensuremath{\mathsf{default}}}
\newcommand{\fieldtid}{\ensuremath{\mathsf{tid}}}

\title{Weighted NetKAT}
\subtitle{A Programming Language For Quantitative Network Verification}

\begin{document}
\begin{abstract}
We introduce weighted \netkat, a domain-specific language for modeling and verifying {\em quantitative} network properties. The language is parametric on a {\em semiring}, enabling the treatment of a wide range of quantities in a uniform way. We provide a denotational semantics and an equivalent operational semantics, the latter based on a novel model of \emph{weighted \netkat automata} ($\mathsf{WNKA}$) capturing the stateful behavior of our language. With $\mathsf{WNKA}$, we obtain a class of generic decision procedures for reasoning about \emph{quantitative safety and reachability} in a fully automatic way, even in the presence of possibly unbounded iteration. We demonstrate the applicability of our framework in a case study using Internet2's Abilene network as the underlying topology.

\end{abstract}
\maketitle
\section{Introduction}
\label{sec:introduction}

The field of network verification has emerged as a significant success story for the programming languages community in recent years. The idea is to see the {\em network as a program}, and model the topology of a network and the configurations of its devices as programs in a domain-specific language, which can then be analyzed to verify properties of interest. This basic approach has been applied successfully at scale in industry, where it has shown to improve the correctness and reliability of networks by catching errors at design time~\cite{switchv,secguru}.

Among the numerous network verification frameworks that have been proposed, \netkat~\cite{netkat}
stands out for its strong theoretical foundations based on Kleene Algebra with Tests (\kat)~\cite{Kozen97}.
Indeed, the deep connection between \netkat and finite automata has been instrumental in
facilitating production-grade~\cite{mark,netkat-google} scalable verification based on automata-theoretic methods~\cite{coalg-netkat,katch,netkat-learning}.

However, \netkat has a critical limitation: its semantics only captures the packet-forwarding behavior of the network. Hence, it can be used to capture basic properties involving the paths that packets take (e.g., reachability, isolation, forwarding loops, etc.). But in many situations, network operators need to reason about richer {\em quantitative} properties, such as bandwidth, latency, reliability, or security, that cannot be gleaned from topologies and device configurations, but are important for applications such as traffic engineering, fault tolerance, and security.

This paper presents weighted \netkat (\wnetkat), a new framework for modeling and reasoning about such quantitative network properties. \wnetkat enriches \netkat with new syntactic constructs for assigning and manipulating weights and a semantics that assigns a weight to each execution. At a technical level, we model weights as elements of a semiring---intuitively, semirings arise in networking as their operations model both alternative (e.g., sum) and joint (e.g., product) use of information, across all possible paths in the topology.

Although extending \netkat with weights may appear straightforward initially, there are numerous challenges that arise in the design of the language and in formulating the semantics correctly. For the latter, one has to carefully restrict the semiring to ensure that iteration can be computed and, more interestingly, the operational semantics of the language requires a new automaton model that captures both the presence of weights (very much in line with classical weighted automata) but also correctly accounts for the idiosyncrasies of \netkat semantics. At the level of the expressiveness of the language one has to take into account that the weights needed to compute the relevant quantities might be associated with different parts of the network (e.g. a switch, a link, a port) and therefore the new syntactic constructs need to offer that flexibility. 

We provide a thorough formalization of the language and its metatheory including a denotational semantics, language model, operational model using \emph{\wnetkat automata}, and theorems that equate these different models. The soundness of our automata construction then enables the verification of \wnetkat policies at the level of  automata. At time of verification the semiring parametricity shines: in different contexts, the nature of the relevant weights varies. For latency we might want to use integers whereas for security we might want to use an ordered set of permission levels. Moreover, the way these quantities need to be combined to propagate through the network to yield the answer to a verification question also varies (e.g., worst-case latency or best-case reliability). 

We focus on two classes of quantitative properties in our verification quest: $\wta$-safety (``Do all paths in the network have weight \emph{at most} $r$?'') and $\wta$-reachability (``Does there exist a path in the network with weight \emph{at least} $r$?''). We provide algorithms to decide these verification questions and then illustrate their applicability in a case study. In particular, although \wnetkat cannot precisely model the kind of quality of service (QoS) properties that depend on flow-level interactions such as congestion, these verification questions encompass a broad range of network performance characteristics. For example, many quantitative aspects of networks---such as reliability (from historical packet loss) or security (whether a link is trusted)---do not depend on modeling dynamic packet-processing behavior at all. Increased bandwidths and packet-processing rates in networks have also recently enabled network performance to be modeled at a coarser granularity that does not need to directly consider queueing or packet-level congestion~\cite{b4}.

In summary, this paper makes the following contributions:
\begin{itemize}[leftmargin=*]
\item We develop \wnetkat, a semiring-based framework 
    to facilitate reasoning about quantitative network behaviors  such as bandwidth constraints, latency measurements, and reliability metrics.

\item We provide a comprehensive formal treatment of the semantics of \wnetkat,
    including denotational (\Cref{sec:syntax-semantics}),
    language-theoretic (\Cref{sec:lang-model}), and
    operational (\Cref{sec:wnetkat-automata}) models. The latter is based on a new automaton model---\wnetkat automata. 
    We develop a sound translation, akin to the classical Thompson construction, from \wnetkat expressions to \wnetkat automata.

\item We establish the exact computation of \wnetkat expressions---in particular, providing the first computable semantics for probabilistic extensions of \netkat. This enables algorithms for verifying quantitative network properties: $\wta$-safety and $\wta$-reachability (\Cref{sec:decidability}).

\item We showcase the applicability of \wnetkat in the setting of Internet2's
    Abilene backbone network, providing worst-/best-case network guarantees over a range of practical network phenomena with automatically generated concrete witnesses and/or counterexamples
        (\Cref{sec:case-studies}).
\end{itemize}
We also show that \wnetkat subsumes the original semantics of \netkat as well as the guarded fragment of \probnetkat~\cite{probnetkat,mcnetkat}. We include proofs of all formal claims in \appendixref.

\section{Quantitative Network Verification with \wnetkat}
\label{sec:model-quant-prop}

In this section, we give an overview of the \emph{quantitative network
verification} enabled by \wnetkat. First, we briefly recap modeling networks with \netkat, after which we discuss the
challenges of quantitative network behavior.
Finally, we describe the verification of quantitative properties, namely $\wta$-safety and $\wta$-reachability, through a computable semantics based on \wnetkat automata.

\subsection{Background: Encoding Networks in \wnetkat}
\label{sec:encode-network}

\begin{figure}[t]
    \begin{tabular}{lcr}
    \begin{minipage}{.5\textwidth}
    \includegraphics[scale=0.25]{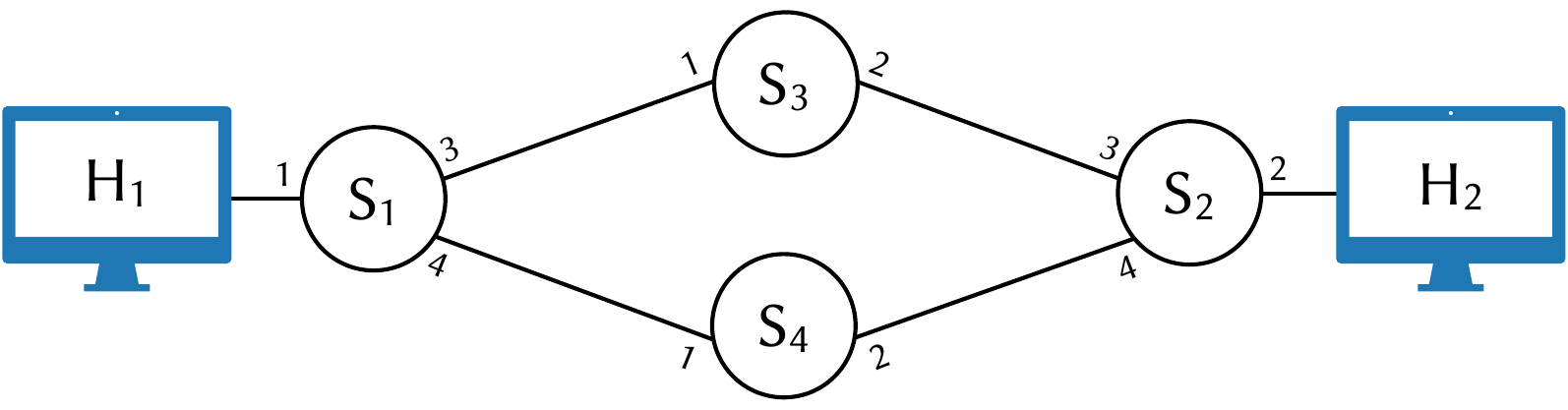}
    \end{minipage}
    &
    \vrule
    &
    \(
    \begin{array}{rcl}
    \pola &\triangleq&
        \IF[\EQ{\fieldsw}{\switcha_1}]
           {\pola_1}
           {\dots}
    \\
    \linka &\triangleq&
        \IF[\EQ{\fieldsw}{\switcha_1}]
           {\linka_1}
           {\dots}
    \\
    \mathit{in,out} &\triangleq&
        \OR{(\AND{\EQ{\fieldsw}{\switcha_1}}{\EQ{\fieldpt}{1}})}{} \\
        && \quad (\AND{\EQ{\fieldsw}{\switcha_2}}{\EQ{\fieldpt}{2}})
    \\
    \mathit{net} &\triangleq& \mathit{in} \SEQN
                              \ITER{(\SEQ{\pola}{\SEQ{\topoa}{\DUP}})} \SEQN
                              \mathit{out}
    \end{array}
    \)
    \end{tabular}
    \caption{Sample network and its encoding in \wnetkat.}
    \label{fig:topo}
\end{figure}

\wnetkat is a conservative extension of \netkat~\cite{netkat}, a domain-specific language for modeling a networks' forwarding policies. When disregarding quantitative aspects, modeling in \wnetkat is thus analogous to modeling networks in \netkat. Let us illustrate this by means of an example.

Consider the network in \Cref{fig:topo} (left), consisting of two \emph{hosts} and four \emph{switches}. These hosts and switches are connected via links at designated \emph{ports}, giving rise to the network's topology. Every switch operates according to a forwarding table. For instance, if $\switcha_1$ receives a packet destined for $\hosta_2$, then $\switcha_1$ will send the packet either via port $3$ or port $4$. 

In \wnetkat, we model networks as intuitive programs, which are called \emph{policies}. More specifically, a policy models how the \emph{fields} of a packet that is being sent through the network are modified over time. In our example, a packet consists of the fields $\fieldsw$ (holding the switch the packet is currently at), $\fieldpt$ (holding the port the packet is currently at), and $\fielddst$ (holding the packet's destination). The high-level structure of the policy modeling our example network is depicted in  \Cref{fig:topo} (right). Consider the top-level policy $\mathit{net}$ and let us go over each of its components separately: 
\[
\mathit{in} \SEQN
\ITER{(\SEQ{\pola}{\SEQ{\topoa}{\DUP}})} \SEQN
\mathit{out}~.
\]
$\mathit{in}$ and $\mathit{out}$ are predicates specifying the network's ingress/egress points: a packet can enter/leave the network at port $1$ of $\switcha_1$ or at port $2$ of $\switcha_2$. The expression $\ITER{(\SEQ{\pola}{\SEQ{\topoa}{\DUP}})}$ then models the iterative behavior of the network and is intuitively to be read as follows: \enquote{$\pola$} look up in the current switch's forwarding table at which port the packet is to be placed next, then \enquote{$\topoa$} send the packet via the corresponding link, then 
\enquote{$\DUP$} log the current packet's state in a history, and \enquote{$\ITER{(\ldots)}$} repeat.

Both $\pola$ and $\topoa$ are basically case distinctions on the current packet's switch. For instance, if the packet is currently at switch $\switcha_1$, then the corresponding sub-policies are given by
\begin{equation}\label{eq:enc1}
	\begin{array}{rcl}
		\pola_1 &\triangleq& \IF[\EQ{\fielddst}{\hosta_2}]
		{(\ADD{\ASSN{\fieldpt}{3}}
			{\ASSN{\fieldpt}{4}})}
		{\IF[\EQ{\fielddst}{\hosta_1}]
			{\ASSN{\fieldpt}{1}}
			{\DROP}}
		\medskip \\
		\linka_1 &\triangleq& \IF[\EQ{\fieldpt}{3}]
		{(\SEQ{\ASSN{\fieldsw}{\switcha_3}}
			{\ASSN{\fieldpt}{1}})}
		{\IF[\EQ{\fieldpt}{4}]
			{(\SEQ{\ASSN{\fieldsw}{\switcha_4}}
				{\ASSN{\fieldpt}{1}})}
			{\EQ{\fieldpt}{1}}}
	\end{array}
\end{equation}
$\pola_1$ branches on the packet's destination: If the destination is $\hosta_2$, then the packet is forwarded via port $3$ or $4$, which is modeled via \wnetkat's \emph{choice operator} $\ADDN$. Similarly, $\linka_1$ branches on the current port, and modifies the fields $\fieldsw$ and $\fieldpt$ according to the network's topology.

\subsection{Modeling Quantitative Network Behavior in \wnetkat}
\label{sec:quant-prop}
\begin{figure}[t]
\vspace{-.2cm}
    \includegraphics[scale=0.3]{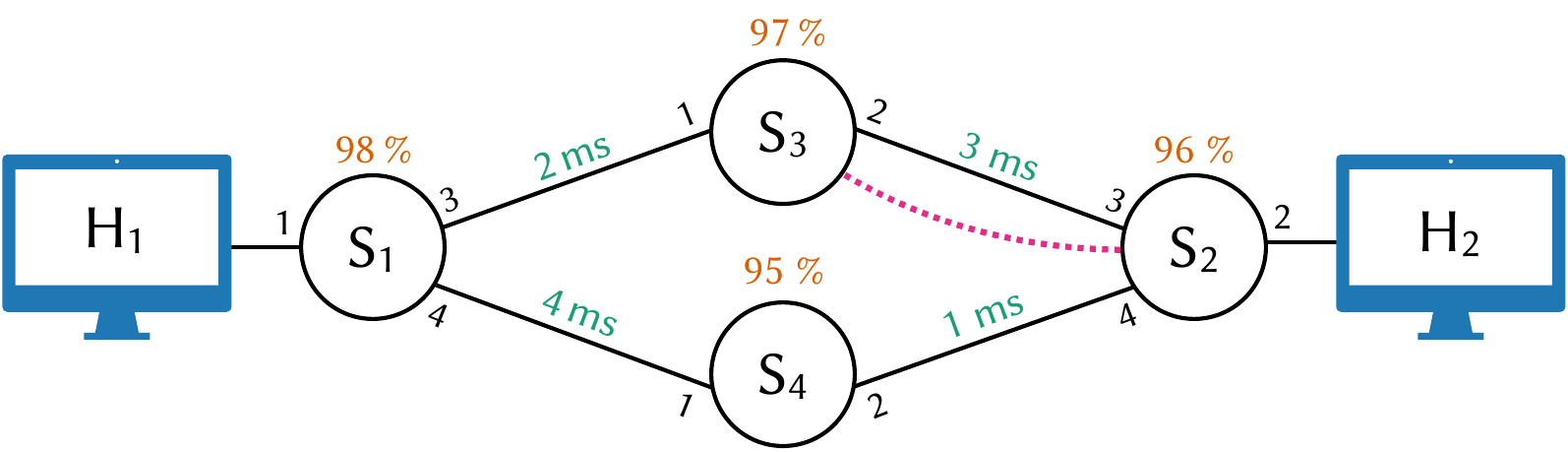}
    \caption{
        Sample network with links now weighted by \textcolor{latencygreen}{latency} and
        switches weighted by \textcolor{failureorange}{reliability}, and an optional
        retry \textcolor{cyclepink}{link} (dotted line) in the network.
    }
    \label{fig:weighted-topo}
    \vspace{-0.5cm}
\end{figure}

We have just exemplified how to model a network's packet forwarding behavior in \wnetkat when disregarding quantitative aspects. Let us now consider \wnetkat's novel and generic capability of \emph{modeling various quantitative aspects of networks}. Again, we proceed example-driven.

Consider the network in \Cref{fig:weighted-topo}. The topology and packet-forwarding behavior of this network coincides with the one from \Cref{fig:topo}. Both the switches and the links are now annotated with \emph{quantitative information}: the switches are annotated by \textcolor{failureorange}{success/failure rates}, i.e., the probability of succeeding in forwarding a packet. Links are annotated by \textcolor{latencygreen}{latencies}, i.e., the time it takes for packet to be sent via a particular link. \wnetkat enables modeling these aspects in a natural manner.

\begin{figure}[t]
    \centering
    \raggedright
    \textbf{Arctic semiring} to model
    \textbf{Latency} or \textbf{Information Leakage}
    \hfill
    ${(\nats \cup \{\infty,-\infty\},\, {\max},\, {+},\, {-\infty},\, 0)}$
    \\[.6ex]
	{
		\scriptsize
		\begin{tabular}{@{}lll@{}}
			\textbf{Weighting $\WEIGH{\wta}{\pola}$:} &
                \textbf{Choice $\ADD{\pola}{\polb}$:} &
                \textbf{Interpretation of $\sem{\pola}(\hista)$:} \\
			\parbox{.31\linewidth}{Policy $\pola$ has latency $\wta$ \millis\ (or reveals $\wta$ bits of information)} &
            \parbox{.32\linewidth}{Choose policy with worse latency (or leakage)} &
            \parbox{.32\linewidth}{\emph{Worst-case} latency (or information leakage) of network paths}
			\\[.6ex]
		\end{tabular}
	}

	\vspace{-.6ex}

    \textbf{Probabilistic-union semiring}
    to model \textbf{Failure Rates} \hfill
    ${([0,1] \cup \{-\infty\},\, {\max},\, \probsum{}{},\, -\infty,\, 0)}$
    \\
    \hfill
    {\footnotesize
    where $\probsum{\proba_1}{\proba_2}$ is the \emph{probabilistic union}
    $\proba_1 + \proba_2 - \proba_1 \cdot \proba_2$
    } \\[.6ex]
	{
		\scriptsize
		\begin{tabular}{@{}lll@{}}
			\textbf{Weighting $\WEIGH{\wta}{\pola}$:} &
                \textbf{Choice $\ADD{\pola}{\polb}$:} &
                \textbf{Interpretation of $\sem{\pola}(\hista)$:} \\
			\parbox{.31\linewidth}{Policy $\pola$ has a failure rate of $\wta$} &
            \parbox{.32\linewidth}{Choose policy with higher failure rate} &
            \parbox{.32\linewidth}{\emph{Worst-case} failure rate of network paths}
		\end{tabular}
	}

	\vspace{-.6ex}

    \textbf{Tropical semiring}
    to model
	\textbf{Confidentiality} or \textbf{Cost}
    \hfill
    ${(\nats \cup \{\infty\},\, {\min},\, {+},\, {\infty},\, 0)}$
    \\[.6ex]
	{
		\scriptsize
		\begin{tabular}{@{}lll@{}}
			\textbf{Weighting $\WEIGH{\wta}{\pola}$:} &
                \textbf{Choice $\ADD{\pola}{\polb}$:} &
                \textbf{Interpretation of $\sem{\pola}(\hista)$:} \\
            \parbox{.31\linewidth}{Policy $\pola$ reveals $\wta$ bits of information
            (or has cost $\wta$)
            } &
            \parbox{.32\linewidth}{Choose whichever policy reveals less bits of information
            (or has cheaper cost)} &
            \parbox{.32\linewidth}{\emph{Best-case} confidentiality (or cost) of network paths}
		\end{tabular}
	}

	\vspace{-.6ex}

	\textbf{Bottleneck semiring}
    to model \textbf{Network Bandwidth}
    \hfill
    $(\nats \cup \{\infty,-\infty\},\, \max,\, \min,\, -\infty,\, \infty)$
    \\[.6ex]
	{
		\scriptsize
		\begin{tabular}{@{}lll@{}}
			\textbf{Weighting $\WEIGH{\wta}{\pola}$:} &
                \textbf{Choice $\ADD{\pola}{\polb}$:} &
                \textbf{Interpretation of $\sem{\pola}(\hista)$:} \\
			\parbox{.31\linewidth}{Restrict bandwidth of policy $\pola$ to $\wta \ \mbps$} &
            \parbox{.32\linewidth}{Choose policy with higher bandwidth} &
            \parbox{.32\linewidth}{\emph{Best-case} bandwidth of network paths}
		\end{tabular}
	}

	\vspace{-.6ex}

    \textbf{Viterbi semiring} to model \textbf{Reliability}
    \hfill
    ${([0,1],\, {\max},\, \,{\cdot}\,,\, {0},\, {1})}$ \\[.6ex]
	{
		\scriptsize
		\begin{tabular}{@{}lll@{}}
			\textbf{Weighting $\WEIGH{\wta}{\pola}$:} &
                \textbf{Choice $\ADD{\pola}{\polb}$:} &
                \textbf{Interpretation of $\sem{\pola}(\hista)$:} \\
			\parbox{.31\linewidth}{Policy $\pola$ has a success rate of $\wta$} &
            \parbox{.32\linewidth}{Choose policy with higher success rate} &
            \parbox{.32\linewidth}{\emph{Best-case} reliability of network paths}
		\end{tabular}
	}

	\vspace{-.6ex}

    \textbf{Security semiring} to model \textbf{Security Levels}
    \hfill
    $(\PublicSecurity < \LowSecurity < \MedSecurity < \HighSecurity,\, \max,\,
        \min,\, \PublicSecurity,\, \HighSecurity)$ \\[.6ex]
	{
		\scriptsize
		\begin{tabular}{@{}lll@{}}
			\textbf{Weighting $\WEIGH{\wta}{\pola}$:} &
                \textbf{Choice $\ADD{\pola}{\polb}$:} &
                \textbf{Interpretation of $\sem{\pola}(\hista)$:} \\
			\parbox{.31\linewidth}{Policy $\pola$ has security level $\wta$} &
            \parbox{.32\linewidth}{Choose policy with higher security level} &
            \parbox{.32\linewidth}{\emph{Best-case} security level of network paths}
		\end{tabular}
	}

	\vspace{-.6ex}

    \textbf{Why semiring}
    to model \textbf{Resource Tracking}
    \hfill
    ${(\text{propositional \emph{positive} DNF}, {\lor}, {\land}, 0, 1)}$
    \\[.6ex]
	{
		\scriptsize
		\begin{tabular}{@{}lll@{}}
			\textbf{Weighting $\WEIGH{\wta}{\pola}$:} &
                \textbf{Choice $\ADD{\pola}{\polb}$:} &
                \textbf{Interpretation of $\sem{\pola}(\hista)$:} \\
			\parbox{.31\linewidth}{Policy $\pola$ uses resource $\wta$} &
            \parbox{.32\linewidth}{Use resources from $\pola$ or resources from $\polb$} &
            \parbox{.32\linewidth}{Resources used by network paths}
		\end{tabular}
	}

	\vspace{-.6ex}

    \textbf{Boolean semiring} to model
    \textbf{NetKAT}~\cite{netkat}
    \hfill
    ${(\{0,1\},\, {\lor},\, {\land},\, 0,\, 1)}$ \\[.6ex]
	{
		\scriptsize
		\begin{tabular}{@{}lll@{}}
			\textbf{Weighting $\WEIGH{\wta}{\pola}$:} &
                \textbf{Choice $\ADD{\pola}{\polb}$:} &
                \textbf{Interpretation of $\sem{\pola}(\hista)$:} \\
            \parbox{.31\linewidth}{$\WEIGH{1}{\pola} = \pola$ and
            $\WEIGH{0}{\pola} = \DROP$} &
            \parbox{.32\linewidth}{Nondeterministic choice between $\pola$ and $\polb$} &
            \parbox{.32\linewidth}{All possible network paths}
		\end{tabular}
	}

	\vspace{-.6ex}

    \textbf{Real Numbers semiring}
    to model
    \textbf{ProbNetKAT}~\cite{probnetkat}
    \hfill
    ${(\reals^{\geq 0} \cup \{\infty\},\, +,\, \cdot,\, 0,\, 1)}$
    \\[.6ex]
	{
		\scriptsize
		\begin{tabular}{@{}lll@{}}
			\textbf{Weighting $\WEIGH{\wta}{\pola}$:} &
                \textbf{Choice $\ADD{\pola}{\polb}$:} &
                \textbf{Interpretation of $\sem{\pola}(\hista)$:} \\
			\parbox{.31\linewidth}{$\pola$ has probability $\wta$} &
            \parbox{.32\linewidth}{Probabilistic choice between $\pola$ and $\polb$} &
            \parbox{.32\linewidth}{Probability of each network path}
		\end{tabular}
	}
    \caption{Instances of \wnetkat: examples of semirings ${(\semidom,\, \semiadd,\, \semimul,\, \addid,\, \mulid)}$ and their use in networking.}
	\label{tab:weighted-netkat-instances}\vspace{-5pt}
\end{figure}

The key idea in modeling this behavior is to introduce a \emph{weighting operation}
\(
	 \WEIGH{\wta}{\pola}~,
\)
where $\pola$ is a policy and $\wta$ is an element from a fixed \emph{semiring}. Intuitively, this operation says \enquote{execute policy $\pola$ with weight $\wta$.} \wnetkat is parametric in that fixed semiring, rendering it a generic language for modeling all kinds of quantitative aspects. \Cref{tab:weighted-netkat-instances} provides an overview of different semirings and what aspects they are capable of modeling. Details on semirings and their operations are provided in \Cref{sec:syntax-semantics}. Let us, for now, gain some intuition for our example network.

With the appropriate semiring (Tropical or Arctic, depending on whether we are interested in best- or worst-case behavior), we can weight policies by \textcolor{latencygreen}{latencies}. To model the latencies attached to the {\em links} in \Cref{fig:weighted-topo}, we extend, e.g.,  the policy $\linka_1$ as follows:
\[
\begin{array}{rcl}
	\linka_1 &\triangleq& \IFN \ \EQ{\fieldpt}{3} \
	\THENN \ \WEIGH{\latency{2}}
	{(\SEQ{\ASSN{\fieldsw}{\switcha_3}}
		{\ASSN{\fieldpt}{1}})} \
	\ELSEN \
	 \\ &&
	\IFN \ \EQ{\fieldpt}{4} \
	\THENN \
	\WEIGH{\latency{4}}
	{(\SEQ{\ASSN{\fieldsw}{\switcha_4}}
		{\ASSN{\fieldpt}{1}})} \
	\ELSEN \
	{\EQ{\fieldpt}{1}}
\end{array}
\]
Alternatively, by choosing the Viterbi semiring, we model a {\em switch's} forwarding success rate (i.e., its reliability) by weighting the  policy encoding its forwarding table. For, e.g., $\switcha_1$ and $\switcha_2$, we have:
\[
\begin{array}{rcl}
	\pola_1 &\triangleq& \WEIGH{\failure{98}}
	{\IF[\EQ{\fielddst}{\hosta_2}]
		{(\ADD{\ASSN{\fieldpt}{3}}
			{\ASSN{\fieldpt}{4}})}
		{\dots}} \\
	\pola_2 &\triangleq& \WEIGH{\failure{96}}
	{\IF[\EQ{\fielddst}{\hosta_2}]
		{\ASSN{\fieldpt}{2}}
		{\dots}}
\end{array}
\]

Note how the construct $ \WEIGH{\wta}{\pola}$ can be placed in different parts of the network policy to model quantities associated with different components (in the above links and switches). We also emphasize that the choice of the semiring determines the interplay of weighting $\odot$ and choice $\ADDN$, which influences whether we are modeling best- or worst-case behavior. Recall that switch $\switcha_1$ may \emph{choose} between forwarding the packet via port $3$ or $4$ if the destination is $\hosta_2$ (modeled by $\ADDN$ in $\pola_1$ from \eqref{eq:enc1} on page \pageref{eq:enc1}). Regarding, e.g., latencies it is thus natural to distinguish between the best- and the worst-case latency of a packet. Which of these cases we actually model depends on the semiring: The Tropical semiring resolves choices in a latency-minimizing manner (since $\ADDN$ is interpreted as a minimum) and the Arctic semiring resolves them in a maximizing manner (since $\ADDN$ maximizes).

\subsection{From Modeling to Verification}
Our first goal was to design a language that inherits the modeling aspects from classic \netkat while enabling to model quantitative aspects in a generic and natural way. Our ultimate goal, however, is to obtain \emph{effective procedures to fully automate quantitative reasoning about networks}.

Developing these effective procedures is challenging. In \Cref{sec:syntax-semantics}, we will present a semiring-valued denotational semantics 
$
	\sem{\pola} \colon \packets \to (\histories \to \semi)
$ 
where $\packets$ denotes the finite set of packets, $\histories$ denotes the countably infinite set of histories (think: \enquote{traces} of packets) and $\semi$ denotes the chosen semiring. Intuitively, $\sem{\pola}(\pkt)(\hista)$ is the weight (e.g., worst-case latency) $\pola$ associates with the history $\hista$ on input packet $\pkt$. While the denotational semantics provides us with a ground truth for assigning meanings to (weighted) policies, it is not immediately amenable to automation---\wnetkat features unbounded iteration, thus denotationally it does not provide a finitary executable description.

Therefore, in the second part of our paper, we develop the novel notion of (finite-state) \wnetkat Automata (\WNKA, for short). The idea is to \emph{compile}---in an algorithmic manner---every \wnetkat policy $\pola$ to a \WNKA $\wnka_\pola$, which accepts a \emph{weighted language over guarded strings}, i.e.,
\[
	\sem{\mathcal{A}_\pola} \colon \GS \to \semi,\quad \text{where the set of guarded strings is $\GS \cong \packets \cdot (\packets \cdot \DUP)^* \cdot \packets$}~.
\]
Guarded strings represent complete packet traces;
intuitively, we can think of a guarded string $\gsA$ as the concatenation of an input packet $\pkt$ and an output history $\hista$, i.e., morally $\gsA = \pkt \cdot \hista$, so that 
\[
	\underbrace{\sem{\wnka_\pola}(\pkt \cdot \hista)}_{\text{operational semantics enabling effective reasoning}} 
	\eeq \qquad 
	\underbrace{\vphantom{\wnka_\pola}\sem{\pola}(\pkt)(\hista)}_{\text{denotational semantics}}~.
\]
The $\mathsf{WNKA}$ $\mathcal{A}_\pola$ \emph{does} provide us with the finitary description required for the algorithmic verification of quantitative network aspects. To produce these automata models we have devised a Thompson-like construction specific to \wnetkat (\Cref{tab:thompson}). Like the classic Thompson construction for regular languages and NFAs, it operates recursively on the structure of a policy. However, as we will see in \Cref{sec:wnetkat-automata}, this construction is far from being a trivial generalization as we cannot rely in  $\varepsilon$-transitions when composing sub-automata (obtained recursively) and we must therefore employ a rather complex on-the-fly epsilon-elimination procedure in a weighted setting. The construction is further complicated by \wnetkat's so-called \emph{carry-on packet semantics}: unlike traditional regular expressions, \netkat is \emph{stateful}, and the output packet of a transition is carried on to the next state. In \wnetkat this carry-on packet \emph{additionally} incurs a weight that needs to be accounted for in the Thompson construction. This carry-on packet is also the reason why we cannot directly use classical weighted automata for the operational semantics and need to introduce a new automaton model. Let us now consider the algorithmic problems $\mathsf{WNKA}$s enable us to tackle effectively.

\subsection{Automatic Reasoning in \wnetkat}
\begin{figure}[t]
\centering
\begin{tikzpicture}[
    >=stealth,
    node distance=5cm,
    every node/.style={font=\small, align=center},
    box/.style={
        draw, rounded corners, minimum width=2cm, minimum height=1cm, inner sep=2pt
    }
]

\coordinate (left) at (0,0.65);
\coordinate (midleft) at (3.5,0.65);
\coordinate (midright) at (8,0.65);
\coordinate (right) at (11.5,0.65);

\node (net) at (left) {Network};
\node[box] (pol) at (midleft) {$\wnetkat$ \\[2pt]  policy };
\node[box] (wnka) at (midright) {$\wnetkat$  \\[2pt]  automaton
};
\node[box] (verif) at (right) {Decision \\[2pt] Procedures};

\draw[->, thick]
    (net) -- node[above]{encoded}
    node[below]{as}
    (pol);

\draw[->, thick]
    (pol) --
    node[above]{compile}
    node[below]{(\Cref{thm:sem-comp})}
    (wnka);

\draw[->, thick]
    (wnka) --
    (verif);

\end{tikzpicture}
\caption{Quantitative verification pipeline of \wnetkat instantiated with a semiring $\semi$.
    }
\label{fig:pipeline}
\end{figure}

Our verification pipeline is depicted in \Cref{fig:pipeline}. With $\mathsf{WNKA}$s, we tackle the following problems:
\begin{enumerate}
	\item \textbf{${\wta}$-Safety}: Do \emph{{all}}---out of possibly infinitely many---traces through the network satisfy a given \emph{{u}pp{er bound}} $\wta$ on, e.g., latency or the overall probability of failure?
	\item \textbf{$\wta$-Reachability}: Does \emph{there exist} a trace through the network satisfying a given  \emph{lower bound} $\wta$ on, e.g., a confidentiality measure or the probability of successfully transmitting a packet?
    \item \textbf{Computing Weights}: Given an input packet $\pkt$ and a history $\hista$, what is the weight $\sem{\pola}(\pkt)(\hista)$ policy $\pola$ assigns to $\hista$ on $\pkt$?
\end{enumerate}
Our corresponding decision procedures are as generic as $\wnetkat$'s modeling capabilities:
In \Cref{sec:decidability}, we provide sufficient conditions on the semiring in order for the above problems to be decidable, and we provide corresponding generic decision procedures based on our $\mathsf{WNKA}$s. The rest of this section is devoted to examples for the above problems.

\subsubsection{Deciding Safety}
The verification question of $\wta$-safety naturally arises when upper-bounding worst-case quantities associated with a network:
\begin{example}
	\label{ex:r-safe}
	We can decide that all traffic in the network from \Cref{fig:weighted-topo}
	has a latency of at most $5\millis$.  For that, we encode the network in \wnetkat (Arctic semiring) and check whether the resulting policy $\pola$ satisfies $5$-safety: For all input packets $\pkt$ and all histories $\hista$, we have $\sem{\pola}(\pkt)(\hista) \leq 5$. 
\end{example}
Moreover, if the safety property is violated, our decision procedure provides
a \emph{witness}: a trace with a weight greater than the safety threshold $\wta$. For
example, suppose a network provider was considering adding the retry \textcolor{cyclepink}{link} (dotted line) to
\Cref{fig:weighted-topo}. In this case, the network would no longer be $5$-safe
as a packet might be repeatedly forwarded back and forth between switches
$\switcha_3$ and $\switcha_2$ (accumulating an unbounded latency). Our decision
procedure identifies this updated network as violating the safety property and
provides a witness including the new link.

\subsubsection{Deciding Reachability}
The question of $\wta$-reachability is more natural to ask in settings that model
\emph{best-case} analyses, such as modeling reliability with the Viterbi
semiring. In this setting, $\wta$-reachability corresponds to finding reliable paths
between nodes.

\begin{example}
	\label{ex:r-reachable}
	We verify, for the network from \Cref{fig:weighted-topo}, that host $\hosta_2$ is reachable from $\hosta_1$ with a
	reliability of at least $90\%$ by deciding if the policy
	\[
	\EQ{\fieldsw}{\hosta_1} \SEQN \mathit{net} \SEQN \EQ{\fieldsw}{\hosta_2}
	\]
	is $0.9$-reachable (in the setting of \wnetkat instantiated with the
	Viterbi semiring).
	If the property is satisfied, we provide a sample trace, i.e., an input
	packet and output history such that the associated weight is
	$\geq 0.9$. In this particular instance, the property is satisfied and the
	sample trace corresponds to the network path
	$\switcha_1 \edge \switcha_3 \edge \switcha_2$ with
	a $91.26\%$ reliability.
\end{example}

While  the network in \Cref{fig:weighted-topo} is simple, we can reason about much more complicated
networks in \wnetkat. In \Cref{sec:case-studies}, we model Internet2's Abilene
backbone network in \wnetkat. Abilene features several \emph{nodes} across
major cities in the United States; we use the network to showcase the
verification of quantitative network behavior in a real-world setting.

In the rest of the paper, we make precise the notions covered throughout this
section and then tie everything back to these examples of quantitative
network verification with a case study over the Abilene network. We first provide a formal
definition of the syntax and semantics of \wnetkat (\Cref{sec:syntax-semantics}),
followed by a language model (\Cref{sec:lang-model}) and an
operational semantics based on \wnetkat automata
(\Cref{sec:wnetkat-automata}). Finally, we describe the verification of
$\wta$-safety and $\wta$-reachability through decision procedures at the level
of \wnetkat automata (\Cref{sec:decidability}), and apply them in the real-world setting of
Internet2's Abilene backbone network (\Cref{sec:case-studies}).

\begin{figure}[t]
\renewcommand{\arraystretch}{1.5}
\begin{tabularx}{\linewidth}{>{\raggedright\arraybackslash}X >{\raggedright\arraybackslash}X}
\rowcolor{lightgray}
A \emph{semiring} is a structure $\semi \eeq (\semidom,\, \semiadd,\, \semimul, \, \addid,\, \mulid)$,
	where $\semidom$ is a set equipped with two binary operations $\semiadd,\semimul\colon\semidom\times\semidom \to \semidom$, and constants $\addid,\mulid\in \semidom$ satisfying:
	\begin{enumerate}[leftmargin=*]
		\item $(\semidom,\, \semiadd,\, \addid)$ is a commutative monoid,
		\item $(\semidom,\, \semimul,\, \mulid)$ is a monoid,
		\item \emph{multiplication} distributes over \emph{addition}.
		\item multiplying with $\addid$ is \emph{annihilating}.
		\end{enumerate}
&
An \emph{$\omega$-continuous semiring} is a structure $(\semi, \semiord)$:
\begin{enumerate}[leftmargin=*]
	\item $\semi$ is a semiring,
	\item $(\semidom,\, \semiord)$ is an $\omega$-complete partial order,
	\item $\semiord$ is positive, i.e., $\addid$ is least element of $\semiord$,
	\item both $\semiadd$ and $\semimul$ are $\omega$-continuous,
	\item $\semi$ admits countable sums, defined as:
\end{enumerate}
\begin{center}
$\displaystyle \semisum {i\in\nats} \semia_i \eeq \underbrace{\cposup[n\in\nats] \semisumtop{i=0}{n} \semia_i}_{\mathclap{\text{supremum of partial sums}}}$
\end{center}

\medskip
\end{tabularx} 
\caption{Definition of ($\omega$-continuous) semirings. Semirings offer an
algebraic basis to formalize weights and relevant operations, whereas
$\omega$-continuity captures the existence of adequate countable sums,
essential for the semantics of iteration.
\ifappendix
All details are provided in \appendixref[appendix:monoids-semirings].
\fi
}\label{fig:semiring}
\end{figure}

\section{\wnetkat: Syntax and Semantics}
\label{sec:syntax-semantics}
\wnetkat is parameterized by an \emph{$\omega$-continuous semiring}, a class of semirings admitting \emph{countably infinite} sums (see \Cref{fig:semiring}). We restrict to $\omega$-continuous semirings as unbounded iteration $\ITER{(-)}$ is naturally captured by an infinite choice---the semantics of which then needs to make use of these countable sums to be well-defined. We discuss this further in \Cref{sec:semantics}.

We fix an $\omega$-continuous semiring $\semi$ throughout this section and the rest of the paper. To emphasize the instantiation with a given semiring $\semi$, we sometimes refer to \swnetkat (and write just \wnetkat when the semiring is clear from context).

\subsection{Syntax}
\label{sec:syntax}
The syntax of \wnetkat is shown in \Cref{fig:syntax-semantics} (left).
 We fix a \emph{finite} set of
\emph{(packet header) values} $\values$ that each of the finitely many
\emph{fields} $\fields$ can take. A \emph{packet} $\pkt$ in the set
$\packets$ is thus a finite function of type $\fields \to \values$,
assigning a value to each field. We usually write $\pkt.\fielda$
instead of $\pkt(\fielda)$ and denote packets in record notation
$\{\fielda_1=\vala_1, \dots , \fielda_k = \vala_k\}$. A \emph{(packet) history}
$\hista = \hcons{\pkt}{\hbar}$ is a \emph{non-empty} list of packets, where
$\pkt$ is the \emph{head packet} and $\hbar$ is the (possibly empty) \emph{tail}.
We sometimes abuse notation, and write $\hcons{\pkt}{\hista}$
instead of $\hcons{\pkt}{\hbar}$.
Finally, elements $\wta$ of the semiring $\semi$ are called \emph{weights}.

\noindent\emph{Predicates} $\testa$ are Boolean combinations of $\FALSE$, $\TRUE$, and \emph{tests} of the form $\EQ{\fielda}{\vala}$. They act as filters: If  the current packet does not satisfy $\testa$, the packet is dropped.
 The \emph{modification} $\ASSN{\fielda}{\vala}$ assigns the value $\vala$ to the field $\fielda$ in the current packet. \emph{Sequential composition} $\SEQ{\pola}{\polb}$ first executes $\pola$ and then executes $\polb$. \emph{Weighting} $\WEIGH{\wta}{\pola}$ weights the execution of $\pola$ by the semiring element $\wta$. The \emph{choice} $\ADD{\pola}{\polb}$ executes either $\pola$ or $\polb$. The primitive $\DUP$ can be understood as a logging command for keeping track of a packet's trajectory through a network.  \emph{Iteration} $\ITER{\pola}$ is, intuitively, a countably infinite choice between terminating or keeping iterating $\pola$, i.e., schematically, $\ITER{\pola}$ is equivalent to
\[
	\ADD{\SKIP\quad}{\quad\ADD{\pola\quad}{\quad\ADD{\SEQ{\pola}{\pola}\quad}{\quad\ADD{\SEQ{\SEQ{\pola}{\pola}}{\pola}\quad}{\quad\ldots} }}}~,
\]
where we write $\SKIP$ ($\DROP$) instead of $\TRUE$ ($\FALSE$) to emphasize the behavior of the predicate when used as a policy. Finally, we can also encode usual control-flow structures in \wnetkat:
\[
    \begin{array}{rcl}
        \SKIP &\triangleq& \TRUE \\
        \DROP &\triangleq& \FALSE \\
    \end{array}
    \qquad\qquad\qquad
    \begin{array}{rcl}
		\IF[\testa]{\pola}{\polb} &\triangleq& 
		\ADD{\SEQ{\testa}{\pola}}{\SEQ{\NOT{\testa}}{\polb}} \\
		\WHILE[\testa]{\pola} &\triangleq&
		\SEQ{\ITER{(\SEQ{\testa}{\pola})}}{\NOT{\testa}}
	\end{array}
\]
These encodings are justified by the semantics of \wnetkat, which we go over next.

\subsection{Semantics}
\label{sec:semantics}
\begin{figure*}[t]
	\begin{minipage}[t]{.535\textwidth}
		\textbf{Syntax}
		\[
		\begin{array}{r@{\ \ }r@{~}c@{~}l@{\ \ }l}
			\textrm{Values} & \values \ni \vala & ::=
			\mathrlap{\concreteval_1 \mid \ldots \mid \concreteval_n} & \\
			\textrm{Fields} & \fields\ni\fielda & ::=  & \mathrlap{\fielda_1 \mid \ldots \mid \fielda_k} \\
			\textrm{Packets} & \packets \ni \pkt,\pkta,\pktb,\pktc & ::= & \mathrlap{\{\fielda_1=\vala_1, \dots , \fielda_k = \vala_k\}} \\
            \textrm{Histories} & \histories \ni \hista& ::= & \mathrlap{\hcons{\pkt}{\hbar}}\\
                               & \hbar & ::= & \mathrlap{\hempty \mid \hcons{\pkt}{\hbar}} \\
			\textrm{Weights} & \semi \ni \wta,\wtb \\
			\textrm{Predicates} & \tests \ni \testa, \testb &
			::= & \FALSE                       & \textit{False/Drop} \\
			& & \mid & \TRUE                  & \textit{True/Skip} \\
			& & \mid & \EQ{\fielda}{\vala}       & \textit{Test} \\
			& & \mid & \OR{\testa}{\testb} & \textit{Disjunction} \\
			& & \mid & \AND{\testa}{\testb}   & \textit{Conjunction} \\
			& & \mid & \NOT{\testa}           & \textit{Negation} \\

			\textrm{Policies} & \pols \ni \pola,\polb &\
			::= & \testa                         & \textit{Filter} \\
			& & \mid & \ASSN{\fielda}{\vala}      & \textit{Modification} \\
			& & \mid & \DUP                   & \textit{Duplication} \\
			& & \mid & \SEQ{\pola}{\polb}     & \textit{Seq.\ Comp.} \\
			& & \mid & \WEIGH{\wta}{\pola}        & \textit{Weighting} \\
			& & \mid & \ADD{\pola}{\polb}   & \textit{Choice} \\
			& & \mid & \ITER{\pola}           & \textit{Iteration} \\
		\end{array}\]
	\end{minipage}\hfill\vrule\hfill\begin{minipage}[t]{.46\textwidth}
		\ \ \textbf{Semantics}
		\[
		\def\arraystretch{1.2}
		\begin{array}{r@{~~}c@{~~}l}
            \sem{\testa}_{\tests} &\colon& \packets \to \boolSemiring \\
            \sem{ \FALSE }_{\tests} (\pkt) & {}={} & 0 \\
            \sem{ \TRUE }_{\tests} (\pkt) & {}={}  & 1 \\
            \sem{ \EQ{\fielda}{\vala} }_{\tests} (\pkt) & {}={} &
                \left[ \pkt.\fielda = \vala \right]  \\
            \sem{\OR{\testa}{\testb}}_{\tests}(\pkt) & {}={} &
                \sem{\testa}_{\tests}(\pkt) \vee \sem{\testb}_{\tests}(\pkt) \\
            \sem{ \AND{\testa}{\testb} }_{\tests} (\pkt) &=&
                \sem{\testa}_{\tests}(\pkt) \wedge \sem{\testb}_{\tests}(\pkt) \\
            \sem{ \NOT{\testa} }_{\tests} (\pkt) & {}={} &
                \left[ \sem{\testa}_{\tests}(\pkt) = 0 \right ]
            \medskip \\

			\sem{\pola} &\colon& {\histories} \to \weightings{\histories} \\
            \sem{ \testa  } (\hcons{\pkt}{\hbar}) &{}={}&
                \left[ \sem{\testa}_{\tests}(\pkt) = 1 \right] \semimul \eta(\hcons{\pkt}{\hbar}) \\
			\sem{ \ASSN{\fielda}{\vala}  } (\hcons{\pkt}{\hbar}) &{}={}&
			\unit\left(\hcons{\updatepkt{\pkt}{\fielda}{\vala}}{\hbar\right}) \\
            \sem{ \DUP } (\hcons{\pkt}{\hbar}) & {}={} &
			\unit(\hcons{\hcons{\pkt}{\pkt}}{\hbar}) \\
			\sem{ \SEQ{\pola}{\polb} } (\hista) &=&
			\sem{ \pola } (\hista) \bind \sem{ \polb } \\
			\sem{ \WEIGH{\wta}{\pola} } (\hista) & {}={} &
			\wta \semimul \sem{ \pola } (\hista) \\
			\sem{ \ADD{\pola}{\polb} } (\hista) & {}={}&
			\sem{ \pola }(\hista) \semiadd \sem{ \polb } (\hista) \\
			\sem{ \ITER{\pola} } (\hista) & {}= {} &
			\semisum{n \in \nats} \sem{ \NFOLD{\pola} } (\hista) \\
			\multicolumn{3}{l}
			{\text{where }
				~ $\NFOLD[0]{\pola} = \SKIP$ ~
				\text{ and } ~
				$\NFOLD[n + 1]{\pola} = \SEQ{\pola}{\NFOLD{\pola}}$}
		\end{array}
		\]
	\end{minipage}
	\hrule
	\caption{Syntax and Semantics of \swnetkat, where $\semi = (\semidom,\, \semiadd,\, \semimul,\, \addid,\, \mulid)$ is an $\omega$-continuous semiring. We assume that the operators bind stronger in the order: $\NOT{}{}$, $\AND{}{}$, 	$\OR{}{}$, $\SEQ{}{}$, $\WEIGH{}{}$, $\ADD{}{}$. }
	\label{fig:syntax-semantics}
\end{figure*}

The semantics of \wnetkat is shown in \Cref{fig:syntax-semantics} (right). In what follows, we first introduce so-called \emph{weightings}---the central semantic objects. We then detail the semantics of the individual constructs and state various desirable properties such as a fixed point characterization of iteration.

\subsubsection{Weightings}
 Intuitively, a policy $\pola$ takes as input a history $\hista$ and produces a set of output histories $\hista'$, where \emph{each output history is weighted by some element from $\semi$}. To formalize an appropriate semantic domain, we introduce the following central objects:
\begin{definition}[Weightings]
	We define \emph{weightings over  a set $X$} as the tuple
			$(\weightings{X}, \, \unit, \, \bind)$,
	where:
	\begin{enumerate}
		\item $\weightings{X} \triangleq  \{ \wtinga \colon X \rightarrow \semi ~|~ \supp(\wtinga) ~	\text{is countable}   \}$ with $\supp(\wtinga) \triangleq \{ x \in X ~|~ \wtinga(x) \neq \addid\}$,
        \item $\unit \colon X \to \weightings{X}$ is the \emph{unit}, defined as $\unit(x) \triangleq \mylam{y} [x=y]$.\footnote{Here, and throughout the paper, we use Iverson bracket notation: $[P] = \mulid$ if the proposition $P$ holds and $[P] = \addid$ otherwise.}
		\item $\bind \colon \weightings{X} \to (X \to \weightings{X}) \to \weightings{X}$ is the \emph{bind}, defined as (using infix notation)
		\[
			\wtinga \bind f \quad{}\triangleq{}\quad
			\mylam{y} {\semisum{x \in\supp(\wtinga)} \wtinga(x) \semimul f(x)(y)}
		\]
	\end{enumerate}
	This sum is well-defined as $\supp(\wtinga)$ is countable and $\omega$-continuous semirings admit countable sums.
\end{definition}
We often call {weightings} the elements of $\weightings{X}$ and denote them by $\wtinga,\wtingb$, and variations thereof. $\supp(\wtinga)$ is called the \emph{support} of $\wtinga$.
It is easy to check that the monad axioms
hold for the operations as defined above, and therefore $(\weightings{X}, \, \unit, \, \bind)$ forms a monad.

We lift the operations and the order of the semiring $\semi$ pointwise to
$\weightings{X}$, i.e., for $\wta\in\semi$, $\wtinga_1,\wtinga_2 \in \weightings{X}$, and
$\{\wtinga_\indexx\}_{\indexx\in\indexset} \colon \indexset \to\weightings{X}$ a family in
$\weightings{X}$ indexed by $\indexset$, we define:
\[
\begin{array}{rcl}
	\wta \semimul \wtinga &\triangleq& \mylam{x} \wta \semimul \wtinga(x) \\
	\wtinga \semimul \wta &\triangleq& \mylam{x} \wtinga(x) \semimul \wta \\
	\wtinga_1 \semiord \wtinga_2 &\text{iff}&
	\forall x \colon
	\wtinga_1(x) \semiord \wtinga_2(x) \\
\end{array}
\qquad
\begin{array}{rcl}
	\addid &\triangleq& \mylam{x} \addid \\
	\wtinga_1 \semiadd \wtinga_2 &\triangleq&
	\mylam{x} \wtinga_1(x) \semiadd \wtinga_2(x) \\
	\semisum{i\in I} \wtinga_i &\triangleq& \mylam{x} \semisum{i\in I} \wtinga_i(x) \\
\end{array}
\]
Finally, as weightings and their associated operations are lifted from
$\omega$-continuous semirings, they satisfy many expected commutativity, associativity, and distributivity properties
\ifappendix
(see \Cref{sec:weightings-props})
\fi.

\subsubsection{The Denotational Semantics of Policies}
Intuitively, a policy $\pola$ takes as input a history $\hista$ and produces a set of output histories $\hista'$, each output being weighted by some element from $\semi$. The notion of weightings formalizes this:  the semantics $\sem{\pola}$ of a policy $\pola$ is of type
 $\histories \to \weightings{\histories}$, i.e., each input history $\hista$ is mapped to a weighting $\sem{\pola}(\hista)$ of (output) histories, and $\sem{\pola}(\hista)(\hista')$ is the weight $\pola$ assigns to the output history $\hista'$ when executed on the input history $\hista$. The set of all histories produced by $\pola$ on input $\hista$ is $\supp(\sem{\pola}(\hista))$, i.e., all histories to which $\sem{\pola}(\hista)$ assigns a non-$\addid$ weight. It is useful to note that, operationally, only the \emph{head} of the input $\hista$ is relevant for the execution of $\pola$ in the sense that for all packets $\pkt$ and all histories $\hista,\hista'$, we have
 \[
 	\sem{\pola}(\hcons{\pkt}{\hempty})(\hista')
 	~{}={}~
 	\sem{\pola}(\hcons{\pkt}{\hista})
 	(\hcons{\hista'}{\hista})~.
 \]

Predicates, modification, and duplication produce at most one output history (behaving analogously
to \netkat~\cite{netkat}). We embed their semantics into \wnetkat via the unit $\unit$ of
$\weightings{\histories}$. Sequential composition, weighting, choice, and iteration yield \wnetkat's
generic capabilities for modeling \emph{quantitative} aspects of networks and require a more
involved treatment.

\smallskip \noindent\textbf{Predicates.} Recall that predicates $\testa$ are Boolean combinations of
$\FALSE$, $\TRUE$, and tests $\EQ{\fielda}{\vala}$. Intuitively, if the head packet $\pkt$ satisfies
$\testa$, then executing $\testa$ does not alter the input history---it is effectless in the sense
that it simply outputs the input history. Otherwise, i.e., if $\pkt$ does not satisfy $\testa$,
then $\pkt$ is dropped. Semantically, this is captured as follows:
$\sem{\testa}(\hcons{\pkt}{\hbar})(\hista') = [\text{$\pkt$ satisfies $\testa$ and
$\hcons{\pkt}{\hbar} = \hista'$}]$.

\smallskip \noindent\textbf{Modification.} The policy $\ASSN{\fielda}{\vala}$ sets the field $\fielda$ of the input history's head packet to $\vala$. To capture this, we use $\updatepkt{\pkt}{\fielda}{\vala}$ to denote the packet obtained from $\pkt$ by updating the value of the field $\fielda$ to the value $\vala$
and define $\sem{ \ASSN{\fielda}{\vala}  } (\hcons{\pkt}{\hbar})
	\eeq
	\unit\left(\hcons{\updatepkt{\pkt}{\fielda}{\vala}}{\hbar\right})$, which is $\mulid$ only when $\hista' =  \hcons{\updatepkt{\pkt}{\fielda}{\vala}}{\hbar}$.

\smallskip \noindent\textbf{Duplication. } $\DUP$ is intended to be a logging statement for keeping track of a packet's trajectory through a network. Its semantics makes this explicit: $\sem{ \DUP } (\hcons{\pkt}{\hbar}) = \unit(\hcons{\hcons{\pkt}{\pkt}}{\hbar})$.

\smallskip \noindent\textbf{Weighting and Choice. } All constructs considered so far produce only $\{\addid,\mulid\}$-valued weightings. Weighting and choice bring the capability of modeling \emph{quantitative} aspects of networks to \wnetkat. Let us consider an example to illustrate how these constructs act in concert.
\begin{example}
    Let $\semi = {(\nats \cup \{\infty,-\infty\},\, {\max_{\nats}},\, {+_{\nats}},\, {-\infty},\, 0)}$ be the Arctic semiring and consider the simple policies
	\[
		\pola_1 \eeq \WEIGH{3}{\ASSN{\fielda}{1}}
		\qquad\text{and}\qquad
		\pola_2 \eeq \WEIGH{5}{\ASSN{\fielda}{2}}~.
	\]
	We have
	$
	\sem{\pola_1}(\hcons{\pkt}{\hbar})(\hista') = 3 +_\nats \sem{{\ASSN{\fielda}{1}}}(\hcons{\pkt}{\hbar})(\hista')
	$, where $\sem{{\ASSN{\fielda}{1}}}(\hcons{\pkt}{\hbar})(\hista') = 0$ if $\hista' =  \hcons{\updatepkt{\pkt}{\fielda}{1}}{\hbar}$ and $ {-\infty}$ otherwise. Hence, $
	\sem{\pola_1}(\hcons{\pkt}{\hbar})(\hista') = 3$  when $\hista' =  \hcons{\updatepkt{\pkt}{\fielda}{1}}{\hbar}$ enabling us to model that
	\begin{center}\emph{
		\enquote{Modifying the input packet by setting $\fielda$ to $1$ incurs a cost (latency) of $3$}
	}
	\end{center}
	and similarly for $\pola_2$. 	Let us now combine $\pola_1$ and $\pola_2$ via a \emph{choice}, i.e., let $\pola = \ADD{\pola_1}{\pola_2}$. We have
	\[
		\sem{\pola}(\hcons{\pkt}{\hbar})
		\eeq
		\textstyle \max_\nats(\sem{\pola_1}(\hcons{\pkt}{\hbar}), \sem{\pola_2}(\hcons{\pkt}{\hbar}))
		\eeq \mylam{\hcons{\pkt'}{\hbar'}}
		\begin{cases}
			3 & \text{if $\pkt' = \updatepkt{\pkt}{\fielda}{1}$ and $\hbar = \hbar'$} \\
			5 & \text{if $\pkt' = \updatepkt{\pkt}{\fielda}{2}$ and $\hbar = \hbar'$} \\
			-\infty & \text{otherwise}~.
		\end{cases}
	\]
	Here, because the assignments in each summand of $\pola$ are different the result of the \emph{semiring} addition $\max_\nats (\sem{\pola_1}(\hcons{\pkt}{\hbar}), \sem{\pola_2}(\hcons{\pkt}{\hbar}))$ will be of the shape $\max_\nats (n, -\infty)$ for $n \in \{3, 5\}$, producing two different output histories in the support.
	Now consider a small change in $\pola_2$ using \emph{the same output history}:
	\[
	\pola_1 \eeq \WEIGH{3}{\ASSN{\fielda}{1}}
	\qquad\text{and}\qquad
	\pola_2 \eeq \WEIGH{5}{\colorbox{headergray}{\ASSN{\fielda}{1}}}~.
	\]
	How does $\pola = \ADD{\pola_1}{\pola_2}$ behave now? We have $\sem{\pola}(\hcons{\pkt}{\hbar})(\hista)
	=
	\max_\nats ( \sem{\pola_1}(\hcons{\pkt}{\hbar})(\hista) , \sem{\pola_2}(\hcons{\pkt}{\hbar})(\hista) ) = \max_\nats (3, 5)$ (or vice-versa),
    which is $5$ if $\hista = \hcons{\updatepkt{\pkt}{\fielda}{1}}{\hbar}$.
    i.e., the choice is resolved in a \emph{cost maximizing manner}. This emphasizes how the semiring operations determine the interplay of weighting and choice. If, e.g., instead of the Arctic semiring, we were to choose the \emph{Tropical} semiring $\semi ={(\nats\cup\{\infty\},\, {\min_\nats},\, {+_\nats},\, {\infty},\, 0)}$, the choice is resolved in a \emph{cost minimizing manner} and $\sem{\pola}(\hcons{\pkt}{\hbar})(\hista)=3$.

	\end{example}
\smallskip \noindent\textbf{Sequential Composition.}
The weightings produced by sequentially composing policies is naturally captured by the bind $\bind$ operation of the monad $\weightings{\histories}$ of weightings. We have
\[
	\sem{\SEQ{\pola}{\polb}}(\hista)(\hista')
	\eeq
	(\sem{ \pola } (\hista) \bind \sem{ \polb })(\hista')
	\eeq
	\semisum{\hista'' \in\supp(\sem{\pola}(\hista))} \sem{\pola}(\hista)(\hista'') \semimul \sem{\polb}(\hista'')(\hista')~.
\]
This is intuitive: The weight $\SEQ{\pola}{\polb}$ assigns to $\hista'$ on input $\hista$ is obtained by summing over all \emph{intermediate outputs} $\hista''$ that $\pola$ produces on $\hista$. For each such $\hista''$, we multiply the weight $\pola$ assigns to $\hista''$ on input $\hista$ by the weight weight $\polb$ assigns to $\hista'$ on input $\hista''$, which captures the sequential behavior of $\SEQ{\pola}{\polb}$.
Note that the semantics of conjunction of two predicates coincide with their
sequencing, i.e., $\sem{\AND{\testa}{\testb}}_{\tests} = \sem{\SEQ{\testa}{\testb}}$; we sometimes
use the two combinators interchangeably for predicates.

\smallskip \noindent\textbf{Iteration.}  Recall that $\ITER{\pola}$ is, intuitively, a countably infinite choice between terminating or keeping iterating $\pola$. Semantically, this behavior is captured by the countable sum
\[
	\sem{\ITER{\pola}}(\hista)
	\eeq
	\semisum{n \in \nats} \sem{ \NFOLD{\pola} } (\hista)
	\eeq
	\ADD{\sem{\SKIP}(\hista)}
	{
		\ADD{\sem{\pola}(\hista)}
		{\ADD{
			\sem{\SEQ{\pola}{\pola}}(\hista)
		}{
		\ldots
		}
		}
	}~.
\]

As a sanity check for this definition,  we establish the usual least fixed characterization
from \kat and \netkat\cite{netkat}, i.e., $\sem{\ITER{\pola}} $ is the least solution of
$
	\sem{\ITER{\pola}} \eeq \sem{\ADD{\SKIP}{\SEQ{\pola}{\ITER{\pola}}}}
$
\ifappendix
(see \Cref{sec:iter-lfp})
\fi.

\section{Language Model}
\label{sec:lang-model}

In this section, we define the language model of \wnetkat, in which each policy $p$ is assigned a {\em weighting of guarded strings}, generalizing \netkat's languages of guarded strings.

\subsection{Reduced Syntax}

\begin{figure}[t]
\[
\begin{array}{rcl}
\multicolumn{3}{l}{\textbf{Complete Tests}}\\
\Test \ni   \ctest{\pkt} & \triangleq  & \AND{\AND{\EQ{\fielda_1}{\pkt.\fielda_1}}{\dots}}{\EQ{\fielda_k}{\pkt.\fielda_k}} \\
\multicolumn{3}{l}{\textbf{Complete Assignments}}\\
\Assign \ni  \cass{\pkt} \ \ & \triangleq  & \SEQ{\SEQ{\ASSN{\fielda_1}{\pkt.\fielda_1}}{\;\,\dots\,}}{\,\ASSN{\fielda_k}{\pkt.\fielda_k}}
\end{array}
\quad \vrule \quad
\begin{array}{rcl}
\multicolumn{3}{l}{  \textbf{Reduced Policies}}\\
\rpols \ni \pola, \polb	&    ::= & \ctest{\pkt}  \ |\ \cass{\pkt} \ |\  \DUP \ |\ \SEQ{\pola}{\polb} \\
      &|&  \WEIGH{\wta}{\pola} \ |\  \ADD{\pola}{\polb} \ |\  \ITER{\pola}
    \end{array}
\]
\caption{Reduced \wnetkat Syntax.}
\label{fig:reduced-syntax}
\end{figure}
First, we restrict all \wnetkat policies to a reduced syntax, see \Cref{fig:reduced-syntax}, without loss of expressivity. At the core of our reduced syntax are \emph{complete tests} and \emph{complete assignments}.
A \emph{complete test} is a conjunction of tests $\AND{\fielda_1 = \vala_1}{\AND{\dots}{\fielda_k = \vala_k}}$, covering all  $\fielda_i \in \fields$. In particular, note that this conjunction over all fields means that complete tests precisely match one and only one packet: $\pkt \triangleq \{ \fielda_1 = \vala_1, \dots, \fielda_k = \vala_k\}$. As such the complete test matching packet $\pkt$ is labelled $\ctest{\pkt}$, and the set of all complete tests denoted $\Test$. A complete test is often called an {\em atom} as complete tests are precisely the minimal nonzero elements of the Boolean algebra generated by basic tests $\fielda_i = \vala_i$.
Dually, a \emph{complete assignment} is an expression $\cass{\pkt} \triangleq \SEQ{\fielda_1 \gets \vala_1}{\SEQ{\dots}{\fielda_k \gets \vala_k}}$. We call $\Assign$ the set of all complete assignments. It is easy to see that there are isomorphisms between $\Assign$, $\packets$, and $\Test$. Hence, we often use simply $\pkt \in \packets$ to represent the respective complete test or assignment.

Note that every reduced policy $p\in\rpols$ is itself a standard \wnetkat policy (i.e., $\rpols \subset \pols$). Most importantly, every policy can be converted to a semantically equivalent \emph{reduced} policy.
Therefore, from now on, we \emph{assume w.l.o.g. that all policies are reduced}.

\subsection{Guarded Strings: Basic Notation and Operations}
Guarded strings appeared originally in the work of Kaplan~\cite{Kaplan69} and later played a prominent role in the work of Kozen~\cite{Kozen97} to reason about program (trace) equivalence.  Formally, guarded strings are elements of the set $\GS \subset \rpols$, defined inductively as:
\[
    \GS \triangleq \bigcup_{i \in \nats} \GS^i \qquad
    \GS^0 \triangleq \{ \ctest{\pkta} \SEQN \cass{\pktb} \mid \pkta,\pktb \in \packets \} \qquad
    \GS^{i+1} \triangleq \{ \gsA \SEQN \DUP \SEQN \cass{\pktc} \mid \gsA \in \GS^i, \pktc \in \packets \}
\]
Guarded strings encompass the minimal nonzero elements of the standard model of \netkat and represent \emph{complete packet traces}. This is analogous to the language models of KA(T) in which expressions are interpreted as regular sets of minimal nonzero (join-irreducible) terms.
For convenience, we will exploit the above mentioned isomorphisms of complete tests and assignments, and represent guarded strings as regular strings over packets ($\GS \cong \packets \cdot (\packets \cdot \DUP)^* \cdot \packets$).
A crucial difference, however, between guarded and regular strings is that the concatenation operation captures the consistency of state between two sequentially composed programs. Given two guarded strings $\pkta \gsA \pktb, \pktc \gsB \pktd \in \GS$ their \emph{guarded concatenation} (see \Cref{fig:language} (right)) $\pkta \gsA \pktb \diamond \pktc \gsB \pktd$ is a \emph{partial} operation. The final state $\pktb$ of the first string has to be compatible with the initial state $\pktc$ of the second string, that is $\pktb = \pktc$ for the concatenation to be defined: $\pkta \gsA \pktb \diamond \pktc \gsB \pktd  = \pkta \gsA \gsB \pktd$. Note that guarded concatenation can be lifted to a {\em total} operation over weightings of guarded strings (see \Cref{fig:language}).

\subsection{Language Model for \wnetkat}
We now have all the ingredients to define the language model of \wnetkat as a class of functions $G \colon
\pols \rightarrow \weightings{\GS}$. We remark that \netkat models were given by regular sets of guarded strings, or equivalently,
functions $G \colon \pols \rightarrow 2^\GS$ so we are generalizing the Boolean semiring underlying sets to an arbitrary $\omega$-continuous semiring (note that $2^\GS \cong \weightings[\boolSemiring]{\GS}$).
\begin{figure}[t]
\begin{tabular}{l || l}
		\begin{tabular}{l  l}
		\toprule
		\toprule
            $\boldsymbol{\pola \in \pols}$ & $\boldsymbol{G(\pola)(\gsA) \in \semi}$ \\
		\midrule
		$\ctest{\pkt}$ \quad \quad& $[\gsA = \pkt \, \pkt]$ \\[1ex]
		$\cass{\pkt}$ \quad \quad& $[\exists \pkta \in \packets.\; \gsA = \pkta \, \pkt]$ \\[1ex]
		$\DUP$ \quad \quad& $[\exists \pkta \in \packets.\; \gsA = \pkta \, \pkta \, \DUP \, \pkta]$ \\[1ex]
		$\ADD{\pola_1}{\pola_2}$ \quad \quad& $G(\pola_1)(\gsA) \semiadd G(\pola_2)(\gsA)$ \\[1ex]
		$\SEQ{\pola_1}{\pola_2}$ \quad \quad& $(G(\pola_1) \diamond G(\pola_2))(x)$ \\[1ex]
		$\WEIGH{\wta}{\pola_1}$ \quad \quad& $\wta \semimul G(\pola_1)(\gsA)$ \\[1ex]
                $\ITER{\pola_1}$ \quad \quad& $\semisum{\vala \in \nats} G(\NFOLD{\pola_1})(\gsA)$ \\[1ex]
		\bottomrule
		\bottomrule
	\end{tabular}
	&
	\begin{tabular}{l}
		\textbf{Guarded Strings}\\[.5ex]
	$\GS \cong \packets \cdot \ITER{(\packets \cdot \DUP)}\cdot {\packets}$\\
		\bottomrule
	\textbf{Guarded Concatenation: }
	$
\diamond \colon {\GS}^2 \rightharpoonup {\GS}$\\[.5ex]
$
\pkta \gsA \pktb \diamond \pktc \gsB \pktd = \begin{cases} \pkta \gsA \gsB \pktd  & \pktb = \pktc \\
\text{undefined} &  \pktb \neq \pktc \\
\end{cases}
$
\\
	\bottomrule
\textbf{Lifted Guarded Concatenation}\\[1ex]
$
\diamond \colon \weightings{\GS}^2 \rightarrow \weightings{\GS}$\\\\
$(m_1 \diamond m_2)(x) \triangleq \semisum{\scriptsize\begin{array}{c} x_i \in \supp(m_i)\\ x = x_1 \diamond x_2\end{array}} m_1(x_1) \semimul m_2(x_2)
$
	\end{tabular}
	\end{tabular}
	\caption{Language Model.}
    \label{fig:language}
\end{figure}
\begin{definition}[Language Model]\label{def:lang}
	Let $\pola$ be a \wnetkat policy. We define the \emph{weighted language} $G(\pola)$ of $\pola$ as
    the weighting $G(\pola) \colon \weightings{\GS}$ given by the table in \Cref{fig:language} (left), where we write $G(\pola)(\gsA)$ for the weight attributed to $\gsA \in \GS$ by $\pola$.
\end{definition}

\begin{rlemma}{thm:deno-lang}[Denotational--Language Correspondence]
    For all $\hista \in \histories$ and $\pola \in \pols$:
    \[\llb \pola \rrb (\hista) \eeq \semisum{\gsA \in \supp(G(\pola))} G(\pola)(\gsA) \semimul \llb \gsA \rrb (\hista)\]
\end{rlemma}
The reader familiar with formal power series~\cite{BerstelReutenauer2010} might notice the similarity between the inductive definition $G(p)$ and {\em rational functions}. There is a crucial difference with the presence of $\diamond$, but we will show that, for $\omega$-continuous semirings, such functions can be recognized by a special finite  automaton (\Cref{sec:wnetkat-automata}). Furthermore, the language model enables reasoning about the behavior of \wnetkat policies on \emph{complete traces} of the network, rather than separate input/output histories.

\begin{example}
    \label{ex:language1}
Consider the \wnetkat policy $\ITER{(\WEIGH{3}{\DUP})}$ over the semiring of the extended naturals $\mathbb{N}^\infty = (\mathbb{N} \cup \{\infty\}, +, \cdot, 0, 1)$.
    Then over a single $\pkta \in \packets$, $G(\pola)(\pkta \, \NFOLD{(\pkta \, \DUP)} \, \pkta) = 3^n$.
\end{example}

Note how in \Cref{ex:language1} the star allows for the unbounded sequencing of zero or more $\DUP$s, each accruing a weight multiplier of $3$. This is because, if we intuitively unfold our example policy, we would see that the non-zero output weight of the guarded string $ \gsA = \pkta \, \NFOLD{(\pkta \, \DUP)} \, \pkta$ is produced entirely by the subexpression $\NFOLD{(\WEIGH{3}{\DUP})}$:
\[
    \underbrace{\ITER{(\WEIGH{3}{\DUP})}}_{\mathclap{\gsA \;\mapsto\; 3^n}} \equiv \SUM{n \in \nats} \NFOLD{(\WEIGH{3}{\DUP})} \equiv \underbrace{\SKIP}_{\mathclap{\gsA \;\mapsto\; 0}} \ADDN \underbrace{(\WEIGH{3}{\DUP})}_{\mathclap{\gsA \;\mapsto\; 0}} \ADDN \underbrace{\NFOLD[2]{(\WEIGH{3}{\DUP})}}_{\mathclap{\gsA \;\mapsto\; 0}} \ADDN \cdots \ADDN \underbrace{\NFOLD{(\WEIGH{3}{\DUP})}}_{\mathclap{\gsA \;\mapsto\; 3^n}} \ADDN \cdots
\]
This is by no accident: unlike common string concatenation, one cannot construct longer guarded strings through the concatenation of $\DUP$-free strings. For instance: $\pkta \, \pkta \diamond \pkta \, \pkta = \pkta \, \pkta$. It is the $\DUP$ construct that entirely determines the length of guarded strings, preventing their concatenation from collapsing the intermediate state. After all, the purpose of $\DUP$ from a denotational perspective is precisely to \emph{duplicate} the current packet, freezing it and producing a longer history as its output. By attributing a weight to $\DUP$, we are associating the weight (be it cost, latency, reliability, etc.) of taking another \emph{hop} in our network, thus allowing us to reason about its paths. We can, however, also reason about the overall input-output behavior of a network, without restricting to specific network paths. In that case, we consider only $\DUP$-free policies.

\begin{example}
    \label{ex:language2}
Consider the policy $\ITER{(\ADD{(\WEIGH{\{a\}}{\cass{\pkta}})}{(\WEIGH{\{b\}}{\cass{\pktb}})})}$ over complete assignments $\cass{\pkta}, \cass{\pktb}$, and the formal language semiring $(\mathcal{P}(\Sigma^*), \cup, \cdot, \emptyset, \{\varepsilon\})$, for $\Sigma = \{a, b\}$. We compute $G(\pola)(\gsA)$ for subexpressions $\pola$ of our policy. We consider only $\DUP$-free strings, as $\DUP$-free policies always assign $\addid$ to longer strings.
    \begin{table}[H]
        \begin{tabular}{r || c c c c}
		\toprule
		\toprule
                $\boldsymbol{x \in \GS}$ & $\boldsymbol{\WEIGH{\{a\}}{\cass{\pkta}}}$ & $\boldsymbol{\WEIGH{\{b\}}{\cass{\pktb}}}$ & $\ADD{(\boldsymbol{\WEIGH{\{a\}}{\cass{\pkta}}})}{(\boldsymbol{\WEIGH{\{b\}}{\cass{\pktb}}})}$ & $\ITER{(\ADD{(\boldsymbol{\WEIGH{\{a\}}{\cass{\pkta}}})}{(\boldsymbol{\WEIGH{\{b\}}{\cass{\pktb}}})})}$\\
		\midrule
		\multicolumn{1}{r||}{$\pkta \, \pkta$} & $\{a\}$ & $\varnothing$ & $\{a\}$ & $\{\varepsilon, wa \mid w \in \Sigma^* \}$ \\
                \multicolumn{1}{r||}{$\pkta \, \pktb$} & $\varnothing$ & $\{b\}$ & $\{b\}$ & $\{wb \mid w \in \Sigma^* \}$\\
                \multicolumn{1}{r||}{$\pktb \, \pkta$} & $\{a\}$ & $\varnothing$ & $\{a\}$ & $\{ wa \mid w \in \Sigma^* \}$\\
                \multicolumn{1}{r||}{$\pktb \, \pktb$} & $\varnothing$ & $\{b\}$ & $\{b\}$ & $\{\varepsilon, wb \mid w \in \Sigma^* \}$\\
		\bottomrule
		\bottomrule
	\end{tabular}
    \end{table}
\end{example}

Note that unlike our $\DUP$ example above, a $\DUP$-free expression does not constrain at all the summands $\NFOLD{\pola}$ that produce non-zero weights as part of the star computation. For $\pola \triangleq \ITER{(\ADD{(\WEIGH{\{a\}}{\cass{\pkta}})}{(\WEIGH{\{b\}}{\cass{\pktb}})})}$, and an example input $\DUP$-free guarded string $\gsA \triangleq \pkta \, \pkta$ we have that:
\[
    \underbrace{\ITER{\pola}}_{\mathclap{\gsA \;\mapsto\; \{\varepsilon, wa \mid w \in \Sigma^* \}}} \quad \equiv \quad \SUM{n \in \nats} \NFOLD{(\pola)} \quad \equiv \quad \underbrace{\SKIP}_{\mathclap{\gsA \;\mapsto\; \{ \varepsilon \}}} \enspace \ADDN \enspace \underbrace{\pola}_{\mathclap{\gsA \;\mapsto\; \{ a \}}} \enspace \ADDN \enspace \underbrace{\NFOLD[2]{\pola}}_{\mathclap{\gsA \;\mapsto\; \{ aa, ba \}}} \enspace \ADDN \enspace \cdots \enspace \ADDN \enspace \underbrace{\NFOLD{\pola}}_{\mathclap{\gsA \;\mapsto\; \{ wa \mid w \in \Sigma^n \}}} \enspace \ADDN \enspace \cdots
\]
One may find it unintuitive that the same $\DUP$-less string $\gsA \triangleq \pkta \, \pkta$ is assigned non-zero weight not only by $\SKIP$, but also by every positive n-th exponentiation of $\pola$. After all, the semantics of n-th iteration requires that the input guarded string be a representative trace of the ``n-times sequencing'' of $\pola$. However, note that the nuance of guarded concatenation (\Cref{fig:language}) allows for precisely this:
\begin{align*}
    G(\NFOLD[2]{\pola})(\pkta \, \pkta) \eeq G(\SEQ{\pola}{\pola})(\pkta \, \pkta) \eeq \big( G(\pola) \diamond G(\pola) \big) (\pkta \, \pkta) &\eeq G(\pola)(\pkta \, \pkta) \semimul G(\pola)(\pkta \, \pkta) \\
																							  &\;\semiadd\; G(\pola)(\pkta \, \pktb) \semimul G(\pola)(\pktb \, \pkta)
\end{align*}
Intuitively, this is because guarded concatenation of $\DUP$-free strings will always collapse any intermediate state back into a minimal, $\DUP$-free string:
\(
    \pkta \, \pkta \diamond \pkta \, \pkta \eeq \pkta \, \pkta \eeq \pkta \, \pktb \diamond \pktb \, \pkta.
\)
We can extend this behavior to as many concatenations of $\DUP$-free strings as needed to ``match'' a given $\NFOLD{\pola}$. This is crucial to model the input-output behavior of networks with cycles, where in the absence of trace information, one has no control over how many times a cycle is traversed from a given input-output packet pair.
As such, computing the exact semantics of star requires that we compute the semantics of an infinite amount of policies: namely, the n-th iterates of the policy, and then sum them. Any finite approximation will yield an incorrect result. As seen above, $G(\ITER{\pola})(\pkta \, \pkta)$ is an infinite language, but all its n-th iterates produce only finite languages. We achieve the exact computation of \wnetkat policies through the use of an automata-based operational semantics.

\section{\wnetkat Automata}
\label{sec:wnetkat-automata}
In this section, we will present an operational semantics for \wnetkat. We will define a special weighted automaton model---\wnetkat automata---and show how to construct a finite automaton from any \wnetkat expression. \wnetkat Automata (\WNKA) resemble classic weighted automata, albeit adjusted to the specifics of \wnetkat and its guarded string-based language model. While classic weighted automata would only consume a symbol at a time, \WNKA consume packets in \emph{linked pairs}, as processing guarded strings requires keeping a state of the ``previous packet''.

\begin{definition}[\wnetkat Automaton]
    A \emph{\wnetkat automaton} (\WNKA) is a 4-tuple $\wnka = (\wnkastates, \wnkainit, \wnkatrans{}{}, \wnkaout{}{})$ where $\wnkastates$ is a finite set of states, $\wnkainit \colon \weightings{\wnkastates}$ is the initial weighting, $\wnkatrans{}{}$ is a family of transition functions $\wnkatrans{\pkta}{\pktb} \colon \wnkastates \rightarrow \weightings{\wnkastates}$ indexed by packet pairs, and $\wnkaout{}{}$ is a family of output weightings $\wnkaout{\pkta}{\pktb} \colon \weightings{\wnkastates}$.
 \end{definition}

The above definition is similar to that of weighted automata: weightings $\weightings{\wnkastates}$ are simply $\semi$-valued ``vectors'' over $\wnkastates$ (i.e., elements of the $\wnkastates$-semimodule over $\semi$), and the transition functions $\wnkatrans{\pkta}{\pktb} \colon \wnkastates \rightarrow \weightings{\wnkastates}$ are matrices $\wnkastates \times \wnkastates \to \semi$. We push this analogy further and note that functions $\wnkatrans{\pkta}{\pktb} \colon \wnkastates \rightarrow \weightings{\wnkastates}$ are isomorphic to $\weightings{\wnkastates \times \wnkastates}$, which we shall call {\em weighting matrices}. Given weighting matrices $m \colon \weightings{X \times Y}$ and $m' \colon \weightings{Y \times Z}$ we define their product $m \lsemiprod m' \colon  \weightings{X \times Z}$:
  \[
   \wtinga \lsemiprod \wtingb \triangleq \lambda (x, z).\; \semisum{(x, y) \in \supp(\wtinga)} \wtinga(x, y) \semimul \wtingb(y, z) \]
   The initial and output weightings can be equivalently represented as matrices $\wnkainit \colon \weightings{1 \times \wnkastates}$ and $\wnkaout{\pkta}{\pktb} \colon \weightings{\wnkastates \times 1}$. With these notational conventions in hand, we can now more easily define the weighted language of guarded strings $\llb \wnka \rrb \colon \weightings{\GS}$ recognized by an \WNKA $\wnka$:
	\begin{align*}
	    &\sem{\wnka} (\pkt_0 \, \pkt_1 \, \DUP \, \pkt_2 \, \DUP \, \dots \, \DUP \, \pkt_n) \triangleq \wnkainit \times \wnkatrans{\pkt_0}{\pkt_1} \times \wnkatrans{\pkt_1}{\pkt_2} \times \cdots \times \wnkatrans{\pkt_{n-2}}{\pkt_{n-1}} \times \wnkaout{\pkt_{n-1}}{\pkt_n}
	\end{align*}
    Note that when $n = 1$ we obtain $\sem{\wnka} (\pkt_0 \, \pkt_1) = \wnkainit \times \wnkaout{\pkt_0}{\pkt_1}$.

\begin{example}
	\label{ex:automata1}
    Consider the policy $\pola \triangleq \ITER{(\WEIGH{3}{\DUP})}$ from \Cref{ex:language1}. We define a minimal automaton $\wnka_\pola$ such that $\sem{\wnka_\pola} = G(\pola)$, as below. We use single-line arrows for transitions between states and for initial weights, and double-line arrows for the output weight function on each state.
	\vspace{-0.5cm}
	\begin{figure}[H]
	\centering
	\begin{minipage}[c]{0.49\linewidth}
	    \centering
	\begin{tikzpicture}
	    \node[smallstate,label=above:{$s_0$}] (x) {};

	    \draw ($(x) + (-0.7, 0)$) edge[input-edge] node[pos=0.4, above=-0.2mm] {\small$1$} (x);
	    \draw[->] (x) edge[output-edge] node[below=1mm] {\small$[\pkta = \pktb]$} ($(x) + (0, -0.4)$);
	    \draw[->] (x) edge[transition-edge,loop right,looseness=50] node[right] {\small$3 \semimul {[\pkta = \pktb]}$} (x);
	\end{tikzpicture}
	\end{minipage}
	\hfill
	\begin{minipage}[c]{0.5\linewidth}
	\begin{align*}
		\wnkainit(s_0) &\triangleq 1 \quad
		\wnkatrans{\pkta}{\pktb}(s_0) \triangleq 3 \semimul {[\pkta = \pktb]} \\
		\wnkaout{\pkta}{\pktb}(s_0) &\triangleq [\pkta = \pktb]
	\end{align*}
	\end{minipage}
\label{fig:automaton1}
\end{figure}
\vspace{-0.5cm}
Note how the automaton's transitions are labeled by conditions on both the current packet and the previous one---hence the packet-pair semantics of our \WNKA. We can then perform the computation of $G(\pola)(\pkta \, \NFOLD{(\pkta \, \DUP)} \, \pkta)$ as done through the automaton:
\[
\sem{\wnka_\pola}(\pkta \, \NFOLD{(\pkta \, \DUP)} \, \pkta) = \wnkainit \times \wnkatrans{\pkta}{\pkta}^n \times \wnkaout{\pkta}{\pkta} = (1) \times (3 \semimul {[\pkta = \pkta]})^n \times ([\pkta = \pkta]) = 3^n
\]
\end{example}

Although simple, the automaton of \Cref{ex:automata1} captures a good intuition of the behavior of $\DUP$ when iterated by star. As captured by the semantics of \WNKA, a transition can only be taken by consuming a $\DUP$ from the input guarded string. If the $\DUP$ is furthermore part of an expression being iterated with star, then we must allow for the \emph{unbounded} consumption of such $\DUP$s. As our automata are finite, this is achieved by \emph{looping} the transition. Equally, a $\DUP$-free policy must not consume any $\DUP$s from input guarded strings, and as such the corresponding automaton will not have any effectively traversable transitions. In practice, this means that although our automaton may still be comprised of many states with no transitions, these can be collapsed into a single state automaton, where the state output entirely captures the language for $\DUP$-free guarded strings.

\subsection{From Expressions To Automata}

Classically, \netkat automata are constructed  on-the-fly using \emph{Brzozowski derivatives}. However, in the presence of weights, Brzozowski derivatives are known to generally not yield finite automata~\cite{bonchi-algebra}. To avoid this, we instead describe a generalized Thompson's construction (\Cref{tab:thompson}), that is guaranteed to terminate by traversing the syntax of the given expression.

\begin{table}[t]
	\caption{
        \wnetkat Thompson construction. For expression $\pola$, we inductively build
        $\wnka_\pola \triangleq (S_\pola, \wnkainit_\pola, \wnkatrans{}{}^\pola, \wnkaout{}{}^\pola)$.
        We use $+$ to denote coproducts of sets; given a coproduct $X+Y$, $f\colon X \to Z$, and
        $g\colon Y \to Z$ the function $[f,g] \colon X+Y \to Z$ is the copairing of $f$ and $g$.
        We denote by $\wnkainit_{\pola_1} \boxtimes \wnkaout{}{}^{\pola_1}$ the square matrix
        $\pkta\pktb \mapsto \wnkainit_{\pola_1}  \times \wnkaout{\pkta}{\pktb}^{\pola_1}$.}
	\label{tab:thompson}
    \small
	\centering
	\begin{tabular}{l  l  l l l}
		\toprule
		\toprule
		$\boldsymbol{\pola}$ &
        $\boldsymbol{S_\pola}$ &
        $\boldsymbol{\wnkainit_\pola \colon\weightings{S_\pola}}$ &
        $\boldsymbol{\wnkatrans{\pkta}{\pktb}^{\pola} \colon S_\pola \to \weightings{S_\pola}}$&
        $\boldsymbol{\wnkaout{\pkta}{\pktb}^{\pola}\colon \weightings{S_\pola}}$\\
		\midrule

        $\ctest{\pkt}$ &
        $\{ \heartsuit \}$ &
        $\eta(\heartsuit)$ &
        $\addid$ &
        $\heartsuit \mapsto[\pkta = \pktb = \pkt]$

        \\[1ex]

        $\cass{\pkt}$ &
        $\{ \heartsuit \}$ &
        $\eta(\heartsuit)$ &
        $\addid$ &
        $\heartsuit \mapsto[\pktb = \pkt]$

        \\[1ex]

        $\DUP$ &
        $\{ \heartsuit, \clubsuit \}$ &
        $\eta(\heartsuit)$ &
        $s \mapsto
            \begin{cases}
                \eta(\clubsuit) & s=\heartsuit \wedge \pkta = \pktb \\
                \addid & \text{otherwise}
            \end{cases}$ &
        $s \mapsto
            \begin{cases}
                [\pkta = \pktb] & s=\clubsuit \\
                \addid & s = \heartsuit
            \end{cases}$

        \\[3ex]

		$\WEIGH{\wta}{\pola_1}$ &
        $S_{\pola_1}$ &
        $\wta \semimul \wnkainit_{\pola_1}$ &
        $ \wnkatrans{\pkta}{\pktb}^{\pola_1}$ \quad \quad& $\wnkaout{\pkta}{\pktb}^{\pola_1}$

        \\[1.5ex]

        $\ADD{\pola_1}{\pola_2}$ &
        $S_{\pola_1} + S_{\pola_2}$ &
        $\left[\wnkainit_{\pola_1}, \wnkainit_{\pola_2}\right]$ &
        $\left[\left[\wnkatrans{\pkta}{\pktb}^{\pola_1} , \addid \right], \left[ \addid , \wnkatrans{\pkta}{\pktb}^{\pola_2}\right]\right]$ &
        $\left[\wnkaout{\pkta}{\pktb}^{\pola_1} , \wnkaout{\pkta}{\pktb}^{\pola_2}\right]$

        \\[1.5ex]

		$\SEQ{\pola_1}{\pola_2}$ &
        $S_{\pola_1} + S_{\pola_2}$ &
        $[\wnkainit_{\pola_1}, \addid ]$ &
        $\left[ \left[ \wnkatrans{\pkta}{\pktb}^{\pola_1}
                     , \semisum{\pktc}
             \wnkaout{\pkta}{\pktc}^{\pola_1} \times \wnkainit_{\pola_2} \times \wnkatrans{\pktc}{\pktb}^{\pola_2}
                     \right]
              , \left[ \addid
                     , \wnkatrans{\pkta}{\pktb}^{\pola_2}
                     \right]
              \right]$ &
        $\left[ \semisum{\pktc}
        \wnkaout{\pkta}{\pktc}^{\pola_1} \times \wnkainit_{\pola_2} \times \wnkaout{\pktc}{\pktb}^{\pola_2}
              , \wnkaout{\pkta}{\pktb}^{\pola_2}
              \right]$ \\[2ex]

		$\ITER{\pola_1}$ &
        $S_{\pola_1} + \{\heartsuit\}$ &
        $[\addid, \mulid]$ &
        $\left[ \left[ \wnkatrans{\pkta}{\pktb}'
                     , \addid
                     \right]
                     , \left[ (\wnkaout{}{}^\heartsuit \times \wnkainit_{\pola_1} \times \wnkaout{}{}^{\pola_1})_{\pkta\pktb}
                            , \addid
                            \right]
                     \right]$ &
        $\left[ (\wnkaout{}{}^{\pola_1} \times \wnkaout{}{}^\heartsuit)_{\pkta\pktb}
              , \wnkaout{\pkta}{\pktb}^\heartsuit
              \right]$

        \\[1ex]

        &&&

        \footnotesize  \textbf{where}
        $  \wnkatrans{}{}'= \wnkatrans{}{}^{\pola_1} + \wnkaout{}{}^{\pola_1} \times \wnkaout{}{}^\heartsuit \times \wnkainit_{\pola_1} \times \wnkatrans{}{}^{\pola_1}$
        \ \textbf{and} &
        \footnotesize
        $\wnkaout{}{}^\heartsuit=\semistar{(\wnkainit_{\pola_1}  \boxtimes \wnkaout{}{}^{\pola_1})}$

        \\[1ex]
		\bottomrule
		\bottomrule
	\end{tabular}
\end{table}

Crucially, this construction is guaranteed to produce an automaton whose language matches the language of its expression. By transitivity, our construction computes precisely the denotational semantics of the expression, by converting the input histories into a guarded string.

\begin{rlemma}{thm:thompson-soundness}[Soundness of Thompson]
Given a policy $\pola \in \pols$: $\sem{ \wnka_\pola } = G( \pola )$.
\end{rlemma}

\begin{corollary}[Equivalence of \wnetkat policies and \wnetkat automata]
    \label{cor:deno-aut}
    Given a policy $\pola \in \pols$:\footnote{
        The soundness of our syntax and semantics (i.e., \Cref{thm:deno-lang,thm:thompson-soundness} and \Cref{cor:deno-aut}) is additionally mechanized in Lean; the mechanization can be found at \url{https://github.com/cornell-pl/wnetkat-lean/blob/pldi2026/WeightedNetKAT/Papers/PLDI2026.lean}.
}
    \[
    \sem{\pola}(\hcons{\pkt_0}{\hempty})(\hcons{\pkt_n}{\hcons{\dots}{\hcons{\pkt_1}{\hempty}}}) \eeq \sem{\mathcal{A}_{\pola}}(\pkt_0 \, \pkt_1 \, \DUP \, \dots \, \DUP \, \pkt_n)~.
    \]
\end{corollary}

Like the classic Thompson construction, our construction works by combining simpler base automata to achieve more complex ones. This is achieved by composing the vector and matrix components in adequate ways. For example, for $\pola_1 \ADDN \pola_2$, the two automata are simply juxtaposed into one, with no interaction between the two, as depicted below.
\begin{equation*}
\setlength{\arraycolsep}{2pt}
\wnkainit^\ADDN =
\begin{tikzpicture}[>=stealth,thick,baseline={(A.center)},every node/.style={inner sep=1pt,outer sep=0pt,font=\small}]
  \matrix [matrix of math nodes,
           row sep=1pt, column sep=2pt,
	   inner ysep=6pt,
           left delimiter=(, right delimiter=)] (A) {
    \, & \wnkainit^{\pola_1} & \, & \, & \wnkainit^{\pola_2} & \, \\
  };
  \draw ([xshift=0pt,yshift=11pt]A-1-1.north west) rectangle ([xshift=-0.5pt,yshift=-4pt]A-1-3.south east);
  \draw ([xshift=0.5pt,yshift=11pt]A-1-4.north west) rectangle ([xshift=0pt,yshift=-4pt]A-1-6.south east);
\end{tikzpicture}
\qquad
\wnkatrans{\pkta}{\pktb}^\ADDN =
\begin{tikzpicture}[>=stealth,thick,baseline={(A.center)},every node/.style={inner sep=1pt,outer sep=0pt,font=\small}]
  \matrix [matrix of math nodes,
           row sep=1pt, column sep=2pt,
           left delimiter=(, right delimiter=)] (A) {
    \, & \, & \, & \, & \, & \, \\
    \, & \wnkatrans{\pkta}{\pktb}^{\pola_1} & \, & \, & \addid & \, \\
    \, & \, & \, & \, & \, & \, \\
    \, & \, & \, & \, & \, & \, \\
    \, & \, & \, & \, & \, & \, \\
    \, & \addid & \, & \, & \wnkatrans{\pkta}{\pktb}^{\pola_2} & \, \\
    \, & \, & \, & \, & \, & \, \\
  };
  \draw ([xshift=0pt]A-1-1.north west) rectangle ([xshift=0pt]A-3-3.south east);
  \draw ([xshift=0pt]A-5-4.north west) rectangle ([xshift=0pt]A-7-6.south east);
\end{tikzpicture}
\qquad
\wnkaout{\pkta}{\pktb}^\ADDN =
\begin{tikzpicture}[>=stealth,thick,baseline={(A.center)},every node/.style={inner sep=1pt,outer sep=0pt,font=\small}]
  \matrix [matrix of math nodes,
           row sep=1pt, column sep=2pt,
           left delimiter=(, right delimiter=)] (A) {
     \, & \\
     \,\,\wnkaout{\pkta}{\pktb}^{\pola_1} & \\
     \, & \\
     \, & \\
     \, & \\
     \,\,\wnkaout{\pkta}{\pktb}^{\pola_2} & \\
     \, & \\
  };
  \draw ([xshift=-7pt]A-1-1.north west) rectangle ([xshift=9pt]A-3-1.south east);
  \draw ([xshift=-7pt]A-5-1.north west) rectangle ([xshift=9pt]A-7-1.south east);
\end{tikzpicture}
\end{equation*}

The ``quadrants'' of the transition matrices created by copairing allow us to specify the behavior of two classes of transitions: the top left, and bottom right quadrants capture transitioning \emph{inside} the two component automata, while the top right and bottom left qudrants capture transitioning \emph{between} the two component automata. As expected, the construction for choice does not allow transitioning between automata, and the internal automata transitions remain the same.

If we look at sequencing, however, we see a more interesting case. When sequencing two automata, we want to preserve their independent transitions, while additionally being able to ``jump'' from the first automaton to the second. Traditionally, this is done by introducing $\varepsilon$-transitions, non-deterministically bridging each state of the first automaton into the start state of the second. \WNKA{s} do not have $\varepsilon$-transitions and we produce these ``jumps'' by using the packets consumed on exiting a state, to instead transition into the second automaton. Formally we do this by: exiting the first automaton ($\wnkaout{\pkta}{\pktc}^{\pola_1}$), entering the second automaton ($\wnkainit^{\pola_2}$), and finally taking a first transition ($\wnkatrans{\pktc}{\pktb}^{\pola_2}$). This is equivalent to transitioning directly from the first automaton to the second, while accumulating the three weights, i.e. $\pkta\pktb \mapsto \semisum{\pktc} \wnkaout{\pkta}{\pktc}^{\pola_1} \times \wnkainit^{\pola_2} \times \wnkatrans{\pktc}{\pktb}^{\pola_2}$ (note that $\pkta\pktb = \pkta\pktc \diamond \pktc\pktb$).

The most interesting case for the Thompson construction, however, is the one of iteration. In the classic Thompson construction, iteration is achieved by introducing a new (accepting) start state ($\heartsuit$), and whenever a state of the automaton would be accepting, instead $\varepsilon$-transitioning to the new start state, where the string is either accepted or free to begin another iteration. Generalizing this principle requires care, due to the specific semantics of \WNKA{s}. This is best observed in two classes of cases: when iterating over expressions with $\DUP$ vs. $\DUP$-free expressions.

Firstly, as demonstrated in \Cref{ex:language2}, a $\DUP$-free policy, even if iterated, will only attribute non-zero weight to $\DUP$-free guarded string inputs. However, these atomic guarded strings may still be iterated unboundedly, and unobservedly, due to the lack of a $\DUP$. For example, $\pkta \pktd = \pkta \pktb \diamond \pktb \pktc \diamond \pktc \pktd$, so the weight of $\pkta \pktd$ must take into account its ``longer'', unobserved equivalents. Processing such a string is akin to entering the start state ($\wnkainit^{\pola_1}$) and exiting it directly ($\wnkaout{}{}^{\pola_1}$)---in this case three times---for each atomic component being sequenced. In general, this must be done for an unbounded number of intermediate atomic guarded strings, meaning the output of the new start state ($\wnkaout{}{}^\heartsuit$) must be the \emph{star} of the combined $\wnkainit^{\pola_1} \times \wnkaout{}{}^{\pola_1}$ matrix, which naively would require computing an infinite sum.

\begin{definition}[Matrix Star]
	\label{def:matrixstar}
    For a weighting square matrix $M \colon \weightings{X \times X}$, we define its star as $\ITER{M} \triangleq \semisum{n \in \nats} M^n$, where $M^0_{xy} \triangleq [x = y]$ and $M^{n+1} = M^n \times M$.
\end{definition}

The star of a matrix is, however, a common operation in automata theory, also known as the matrix closure. In the past, \citet{bloom1993} developed an algorithm for computing the matrix star in terms of only the underlying semiring operations, together with the \emph{semiring star}, a shorthand defined as $\ITER{\semia} \triangleq \semisum{n \in \nats} \semia^n$. Although computing the star of a semiring element requires computing a countable sum, for all our documented semirings and their combinations, this operation is easily computed, usually in constant time~\cite{Mohri2009}. The Thompson construction depicted in \Cref{ex:automata4} below provides a good view of the matrix star in action, as used to compute the output weights of state $\heartsuit$. For instance, as expected the policy maps the guarded string $\pkta \, \pkta$ to the \emph{infinite} regular language $\{wa \mid w \in \Sigma^*\}$, which is computed entirely by $\wnkaout{}{}^\heartsuit$, the starred output matrix.

\begin{example}
	\label{ex:automata4}
	For the policy $\pola \triangleq \ITER{(\ADD{(\WEIGH{\{a\}}{\cass{\pkta}})}{(\WEIGH{\{b\}}{\cass{\pktb}})})}$ from \Cref{ex:language2}, we obtain the following automaton through the Thompson construction:
	\vspace{-0.5cm}
    \begin{figure}[H]
    \centering
    \begin{subfigure}[b]{0.35\textwidth}
        \centering
	\begin{tikzpicture}
	    \node[smallstate,label=right:{$s_0$}] (x) {};
	    \node[smallstate,label=right:{$s_1$}] (x') at (2.0, 0) {};
	    	    \draw ($(x) + (0, 0.7)$) edge[input-edge] node[pos=0.4, right=-0.2mm] {\small$\{a\}$} (x);
	    \draw ($(x') + (0, 0.7)$) edge[input-edge] node[pos=0.4, right=-0.2mm] {\small$\{b\}$} (x');

	    \draw[->] (x) edge[output-edge] node[below=1mm] {\small$\pkta\pkta, \pktb\pkta \mapsto \{\eword\}$} ($(x) + (0, -0.4)$);
	    \draw[->] (x') edge[output-edge] node[below=1mm] {\small$\pkta\pktb, \pktb\pktb \mapsto \{\eword\}$} ($(x') + (0, -0.4)$);
	\end{tikzpicture}
        \caption{$\ADD{(\WEIGH{\{a\}}{\cass{\pkta}})}{(\WEIGH{\{b\}}{\cass{\pktb}})}$}
	\label{ex:automata4-a}
    \end{subfigure}
    \hfill
    \begin{subfigure}[b]{0.6\textwidth}
        \centering
		\begin{tikzpicture}
	    \node[smallstate,label=right:{$s_0$}] (x) at (0, -0.5){};
	    \node[smallstate,label=right:{$s_1$}] (x') at (5.0, -0.5) {};
	    \node[smallstate,label=right:{$\heartsuit$}] (x'') at (2.5, 0) {};

	    \draw ($(x'') + (-0.7, 0)$) edge[input-edge] node[pos=0.4, above=-0.2mm] {\small$\{\eword\}$} (x'');

	    \draw[->] (x) edge[output-edge] node[below=1mm] {\small$\pkta\pkta \mapsto \{\eword, wa\}$} ($(x) + (0, -0.4)$);
	    \draw[->] (x) edge[output-edge] node[below=5mm, xshift=-1.28mm] {\small$\pktb\pkta \mapsto \{wa\}$} ($(x) + (0, -0.4)$);

	    \draw[->] (x') edge[output-edge] node[below=1mm, xshift=-1.28mm] {\small$\pkta\pktb \mapsto \{wb\}$} ($(x') + (0, -0.4)$);
	    \draw[->] (x') edge[output-edge] node[below=5mm] {\small$\pktb\pktb \mapsto \{\eword, wb\}$} ($(x') + (0, -0.4)$);

	    \draw[->] (x'') edge[output-edge] node[below=0mm] {\small$\pkta\pkta \mapsto \{\eword, wa\}$} ($(x'') + (0, -0.4)$);
	    \draw[->] (x'') edge[output-edge] node[below=3mm] {\small$\pktb\pktb \mapsto \{\eword, wb\}$} ($(x'') + (0, -0.4)$);
	    \draw[->] (x'') edge[output-edge] node[below=6mm, xshift=-1.28mm] {\small$\pktb\pkta \mapsto \{wa\}$} ($(x'') + (0, -0.4)$);
	    \draw[->] (x'') edge[output-edge] node[below=9mm, xshift=-1.28mm] {\small$\pkta\pktb \mapsto \{wb\}$} ($(x'') + (0, -0.4)$);

	\end{tikzpicture}
	\caption{$\ITER{(\ADD{(\WEIGH{\{a\}}{\cass{\pkta}})}{(\WEIGH{\{b\}}{\cass{\pktb}})})}$}
	\label{ex:automata4-b}
    \end{subfigure}
\end{figure}
\vspace{-0.4cm}
\end{example}

Our second key consideration for the case of iteration is best seen when iterating expressions with $\DUP$.
This is the ``classical'' setting, where it is possible to non-deterministically return to the
start state whenever an end state is reached, ready to process another iteration (as is the case
in \Cref{ex:language1}). However, \wnetkat automata do not have a single start and end state. In fact,
\emph{every} state can be both initial and final, as
determined by the initial and output weights. This makes looping the automaton more delicate,
and requires an entirely matrix-based treatment, as demonstrated by component $\wnkatrans{}{}^\pola$
of the star construction: To transition inside the automaton is to either take an internal transition
as normal ($\wnkatrans{}{}^{\pola_1}$), or non-deterministically ($\semiadd$) exit the state
($\wnkaout{}{}^{\pola_1}$) and, via the new start state ($\wnkaout{}{}^\heartsuit$),
loop back into the automaton ($\wnkainit^{\pola_1}$), and transition through ($\wnkatrans{}{}^{\pola_1}$).

\subsection{Computable Semantics of \wnetkat}

Finally, the effective computation of the Thompson construction allows us to compute the automaton for any given \wnetkat policy. Using the correctness of the semantics of the automaton  (\Cref{cor:deno-aut}), and the fact that the automaton semantics is computed through matrix multiplication, which is itself computed in terms of semiring addition and multiplication, we obtain the following:

\begin{theorem}[Computable Semantics]
    \label{thm:sem-comp}
    Given a computable
         semiring $(\semi, \semiord)$, the semantics
    of every \swnetkat policy is computable by compiling
    it to its corresponding \wnetkat automaton.
\end{theorem}

Furthermore, given that the semantics of every \wnetkat policy is computable
by compiling to a \wnetkat automaton, we can verify the questions of $\wta$-safety
and $\wta$-reachability through decision procedures that we develop at the level of
\wnetkat automata in the
following section.

\section{Decidability Results for \wnetkat}
\label{sec:decidability}

In this section, we tie the technical results of the previous sections back to the
verification questions in \Cref{sec:model-quant-prop}. Consider again 
verifying whether a network encoded in \wnetkat as $\pola$ is $\wta$-safe or
$\wta$-reachable. Semantically, these questions correspond to the following two
properties:
\begin{align*}
	\label{def:r-safety}
    \forall \pkt \in \packets, \hista \in \histories \colon \quad &
        \sem{\pola}(\hcons{\pkt}{\hempty})(\hista) \semiord \wta \ ,
    \tag{$\pola$ is $\wta$-safe} \\
	\label{def:r-reachability}
    \exists \pkt \in \packets, \hista \in \histories \colon \quad &
        \sem{\pola}(\hcons{\pkt}{\hempty})(\hista) \semiordrev \wta \ .
    \tag{$\pola$ is $\wta$-reachable}
\end{align*}
In words, $\wta$-safety says that \emph{all} (out of possibly \emph{infinitely many}) traces produced by the policy $\pola$ have weight at most
$\wta$. Dually, $\wta$-reachability says that there exists \emph{some} trace with weight at least~$\wta$.

We generalize techniques by \citet{decidable-weighted} to obtain \emph{generic} decision procedures for $\wta$-safety and $\wta$-reachability for a broad class of semirings.
These procedures can produce \emph{witnesses}, providing operators with diagnostic information (when $\wta$-safety is violated), and synthesize traces satisfying some desired lower bound on a quantity of interest. We illustrate this in \Cref{sec:case-studies} via a case study.

\subsection{Decidability of $\wta$-safety in \wnetkat}
\label{sec:verif-safe}
We begin by establishing the decidability of $\wta$-safety for
semirings that model a \emph{worst-case} analysis.
\begin{rtheorem}{thm:verif-safety}
    [Decidability of $\wta$-safety]
    Let $(\semi, \semiord)$ be a computable semiring such
    that
    \[
    \semia_1 \semiadd \semia_2 \semiord \semia_3
    \quad \text{iff} \quad
    \semia_1 \semiord \semia_3 ~\text{and}~\semia_2\semiord\semia_3~,
    \]
and let
    $\pola$ be an \swnetkat policy. Then
    \(
        \text{``$\pola$ is $\wta$-safe'' is decidable}~.
    \)
    Moreover, if $\semiord$ is total and $\pola$ is not $\wta$-safe, then we can compute a witness, i.e., $\pkt \in \packets$ and $\hista\in\histories$ such that $\sem{\pola}(\hcons{\pkt}{\hempty})(\hista) \not\semiord \wta$.
\end{rtheorem}

Before we describe the decision procedure, let us gain some intuition on the conditions imposed on $\semi$. Computability is necessary to effectively compute the semantics of the \wnetkat policy. The condition on $\semiadd$ expresses that $\semi$ models worst-case behavior: the semiring addition \enquote{chooses} a worst-case scenario so that upper-bounding $\semia_1\semiadd\semia_2$ is equivalent to upper-bounding both $\semia_1$ and $\semia_2$. This condition is satisfied by the Arctic, Probabilistic-union, and Why semirings (cf.\ \Cref{tab:weighted-netkat-instances}).

The decision procedure works as follows. First, we invoke \Cref{cor:deno-aut}, which gives us
\begin{align*}
	\text{$\pola$ is $\wta$-safe}
	\quad\text{iff}\quad
	\forall \gsA\in\GS\colon \sem{\wnka_{\pola}}(\gsA) \semiord \wta~.
\end{align*}
It follows from $\omega$-continuity of $\semi$ and the side condition on $\semiadd$ that the latter is equivalent to
\[
\semisum{\gsA \in \GS} \wlang{\wnka_\pola}(\gsA) \semiord \semithresh~.
\]
We then proceed by showing that this infinite sum over all guarded strings \emph{is computable}, which implies the claim. The key idea is to reduce the computation of this infinite sum to the computation of a matrix star (cf.\ \Cref{def:matrixstar}), which can be done via well-established algorithms~\cite{bloom1993}.
\ifappendix
All details are provided in \appendixref[pf:verif-safety].
\fi
In case $\wta$-safety is violated, we are---by the totality of $\semiord$---guaranteed to find a witness by enumerating guarded strings $\gsA \in \GS$ in a breadth-first search manner until we find one with $\sem{\mathcal{A}_{\pola}}(\gsA) \not\semiord \wta$, where we use
\Cref{thm:sem-comp} to compute the weight $\mathcal{A}_{\pola}$ assigns to $\gsA$. By \Cref{cor:deno-aut}, $\gsA$ can then be turned into an appropriate witness.

\subsection{Decidability of $\wta$-reachability in \wnetkat}
\label{sec:verif-reach}

Next, we establish the decidability of $\wta$-reachability for semirings that model a \emph{best-case} analysis.
\begin{rtheorem}{thm:verif-reach}
[Decidability of $\wta$-reachability for \wnetkat policies]
    Let $(\semi, \semiord)$ be a computable semiring such
    that $\semia_1 \semiordrev \semia_1 \semimul \semia_2$ and
    \[
	\semia_1 \semiadd \semia_2 \semiordrev \semia_3
    \quad \text{iff} \quad
    \semia_1 \semiordrev \semia_3
    ~\text{or}~
    \semia_2 \semiordrev \semia_3~,
    \]
	and let $\pola$ be a \swnetkat policy. Then ``$\pola$ is $\wta$-reachable'' is decidable. Moreover, if $\pola$ is $\wta$-reachable, we can compute a witness, i.e., $\pkt \in \packets$ and $\hista\in\histories$ such that $\sem{\pola}(\hcons{\pkt}{\hempty})(\hista) \semiordrev \wta$.
\end{rtheorem}

Let us again gain some intuition on the imposed conditions. Dually to $\wta$-safety, the condition on the semiring addition expresses that $\semi$ models best-case behavior: lower-bounding $\semia_1\semiadd\semia_2$ is equivalent to lower-bounding one of $\semia_1,\semia_2$. The condition on the semiring multiplication expresses that making traces longer can only make things worse since $\semia_1 \semimul\semia_2$ will always be smaller than $\semia_1$. These conditions are satisfied by the Tropical, Viterbi, and Bottleneck semirings (cf.\ \Cref{tab:weighted-netkat-instances}).

Our decision procedure works as follows. First, we invoke \Cref{cor:deno-aut} to get
\begin{align*}
	\text{$\pola$ is $\wta$-reachable}
	\quad\text{iff}\quad
	\exists \gsA\in\GS\colon \sem{\mathcal{A}_{\pola}}(\gsA) \semiordrev \wta~.
\end{align*}
We then exploit the conditions on $\semiadd$ and $\semimul$ to conclude that the above is equivalent to the existence of a guarded string $\gsA\in\GS$ corresponding to a  \emph{cycle-free run} of $\wnka_\pola$---for a notion of runs we define over the underlying graph structure of $\wnka_\pola$. There are only finitely many cycle-free runs, so it suffices to check $\sem{\mathcal{A}_{\pola}}(\gsA) \semiordrev \wta$ for only finitely many $\gsA$ to decide $\wta$-reachability. If we find such a guarded string $\gsA$, we use \Cref{cor:deno-aut} to turn it into a witness.
\ifappendix
All details can be found in \appendixref[pf:verif-reach].
\fi

\section{Case Studies}
\label{sec:case-studies}
In this section, we demonstrate how our marriage of classic \netkat's modeling capabilities and weighted reasoning enables the automatic quantitative analysis of intricate network configurations.
For that, we use a topology based on Internet2's Abilene backbone network (see \Cref{fig:abilene})\footnote{The TikZ code used in \Cref{fig:abilene} was produced with the help of a generative AI software tool (Claude, Sonnet 4.6).}, which features \emph{nodes} across several cities in the United States. Traffic can enter or exit the network from any node (e.g., a network packet entering at $\bay$ destined for $\nyc$), and every node is able to forward packets to nodes it is linked to (e.g., $\kansas$ can forward packets to $\denver$, $\houston$, and $\indiana$).
We have additionally annotated the network topology in \Cref{fig:abilene} with several quantities, such as the associated \textcolor{failureorange}{failure rates} of each node or the \textcolor{bandwidthpurple}{bandwidth} of each link between nodes (e.g., based on forwarding failure metrics or historical average bandwidths).

\begin{figure}[t]
  \centering
  \begin{tikzpicture}
    \node[anchor=south west, inner sep=0] (image) at (0,0)
      {\includegraphics[scale=0.22]{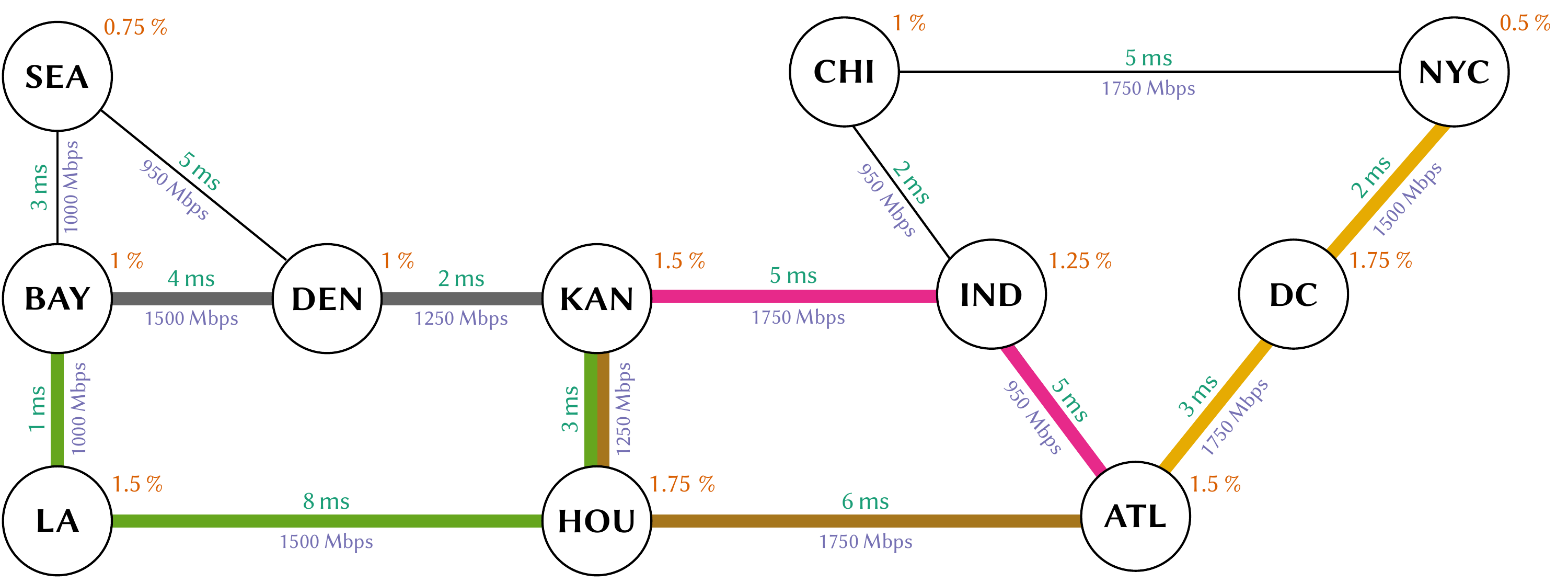}};

    \begin{scope}[x={(image.south east)}, y={(image.north west)}]
      \node[anchor=south east, inner sep=0] at (1.245, 0) {
          \small
          \begin{tabular}{lc}
            \fieldtid & \textbf{Path} \\
            \midrule
              \textcolor{tunnelblack}{1} & \textcolor{tunnelblack}{$\bay,\denver,\kansas$} \\
              \textcolor{tunnelgreen}{2} & \textcolor{tunnelgreen}{$\bay,\la,\houston,\kansas$} \\
              \textcolor{tunnelbrown}{3} & \textcolor{tunnelbrown}{$\kansas,\houston,\atlanta$} \\
              \textcolor{tunnelred}{4} & \textcolor{tunnelred}{$\kansas,\indiana,\atlanta$} \\
              \textcolor{tunnelgold}{5} & \textcolor{tunnelgold}{$\atlanta,\dc,\nyc$} \\
          \end{tabular}
        }
      ;
    \end{scope}
  \end{tikzpicture}
\caption{Topology from Abilene network with links between nodes 
    weighted by \textcolor{latencygreen}{latency}/\textcolor{bandwidthpurple}{bandwidth}
    and nodes by forwarding \textcolor{failureorange}{failure rate}. Network tunnels are
    mapped to their corresponding tunnel ID ($\fieldtid$) in table.}
\label{fig:abilene}
\end{figure}

Suppose that for certain source-destination pairs it is preferable to forward traffic through \emph{tunnels}
instead of with the usual forwarding behavior (e.g., based on \emph{shortest-paths} or some other routing scheme).
\Cref{fig:abilene} highlights example tunnels available in the network to be used by specific nodes for $\nyc$-bound traffic. For example,
 $\kansas$ is configured to use either of the tunnels
$\kansas \edge \houston \edge \atlanta$ ($\fieldtid$ 3) or
$\kansas \edge \indiana \edge \atlanta$ ($\fieldtid$ 4).
Both tunnels exit into $\atlanta$, which in turn is configured to forward $\nyc$-bound traffic through tunnel~$5$. In what follows, we first model this tunneling behavior in classic \netkat and then demonstrate how \wnetkat enables quantitative reasoning. We model the network's forwarding behavior using choice and \emph{nested iteration}:
\[
\begin{array}{rc}
    \pol{\kansas} &\triangleq
\end{array}
\begin{array}{l}
        \IFN \ \EQ{\fieldtid}{0} \ \THENN \\ \quad
            \IFN \ \EQ{\fielddst}{\nyc} \ \THENN \
                (\ADD{\ASSN{\fieldtid}{3}}{\ASSN{\fieldtid}{4}}) \\ \quad
            \ELSEN \ \dots \quad \text{\textcolor{commentgreen}{\texttt{(* Default behavior *)}}} \\
        \ELSEN \ \IFN \ \OR{\EQ{\fieldtid}{1}}{\EQ{\fieldtid}{2}} \ \THENN \ \ASSN{\fieldtid}{0} \\
        \ELSEN \ \IFN \ \EQ{\fieldtid}{3} \ \THENN \ \ASSN{\fieldnode}{\houston} \\
        \ELSEN \ \IFN \ \EQ{\fieldtid}{4} \ \THENN \ \ASSN{\fieldnode}{\indiana} \\
        \ELSEN \ \DROP
\end{array}
\begin{array}{rcl}
    \pola &\triangleq&
        \IFN \ \EQ{\fieldnode}{\atlanta} \ \THENN \\ && \quad
            \ITER{(\pol{\atlanta})} \SEQN \NOTEQ{\fieldnode}{\atlanta} \\ &&
        \ELSEN \ \IFN \ \EQ{\fieldnode}{\bay} \ \THENN \\ && \quad
            \ITER{(\pol{\bay})} \SEQN \NOTEQ{\fieldnode}{\bay} \\ &&
        \ELSEN \ \dots \medskip \\
    \abilene &\triangleq& \ITER{(\SEQ{\pola}{\DUP})}
\end{array}
\]

Similarly to \Cref{sec:encode-network}, $\abilene$ models the process of (i) forwarding a packet according to its current node (modeled by $p$), (ii) recording the packet's state in the history (using $\DUP$), and (iii) repeating this process (using iteration). Policy $\pola$ branches on the packet's current node and invokes the corresponding routing policy. Notice that these routing policies are \emph{also iterated}---an idea dating back to \citet{10.1145/2619239.2626300}. The reason for that becomes apparent when considering $\pol{\kansas}$: we use the field $\fieldtid$ to keep track of the tunnel the packet is currently in ($\fieldtid=0$ meaning no tunnel). $\kansas$ thus marks the end of tunnels~$1$ and~$2$, and  the start of tunnels $3$ and $4$. In particular, if $\kansas$ receives a packet with $\fieldtid=0$ and destined for $\nyc$, this packet may be forwarded \emph{either} via tunnel~$3$ \emph{or}~$4$ (using choice). Now, if $\kansas$ receives a packet destined for $\nyc$ on tunnel~$1$, say, then we have to execute $\pol{\kansas}$ \emph{twice} to ensure that, subsequently, it is correctly forwarded via tunnel~$3$ or~$4$. Iterating $\pol{\kansas}$ naturally captures the necessity of possibly having to execute the policy multiple times.

Consider $\bay$, which is configured analogously to $\kansas$ (i.e., $\nyc$-bound traffic is forwarded via tunnels~$1$ or~$2$). The full encoding of $\bay$ (and for all other examples throughout this section) is included in \appendixref[appendix:case-studies]. By instantiating \wnetkat with the Boolean semiring we can, as with \netkat,
already verify that the tunneled paths between $\bay$ and $\nyc$ are configured so that the two nodes are connected. This amounts to checking that the following policy is $1$-reachable:
\[
    \SEQ{(\AND{\EQ{\fieldnode}{\bay}}{\EQ{\fielddst}{\nyc}})}{
    \SEQ{\abilene}
        {(\AND{\EQ{\fieldnode}{\nyc}}{\NOTEQ{\fieldtid}{0}})}}
\]

\subsection{Verifying Reliability of Tunneled Paths with \Cref{thm:verif-safety}}

Let us now illustrate how \wnetkat can serve as a \emph{verification tool} for bounds on the \emph{reliability} of a network's configuration. Consider again network traffic entering at $\bay$ destined for $\nyc$. Beyond checking that $\bay$ and $\nyc$ are indeed connected, a network provider may additionally wish to ensure that all tunneled paths between $\bay$ and $\nyc$ are sufficiently reliable. This corresponds to upper-bounding the \emph{worst-case failures across all combinations of tunnels a packet might take}. We take the encoding from the previous section and \emph{weight} each node's (iterated) routing policy by its forwarding \textcolor{failureorange}{failure rate}, instantiating \wnetkat with the Probabilistic-union semiring (cf. \Cref{tab:weighted-netkat-instances}):
\[
\begin{array}{rc}
    \pol{\relstr} & \triangleq
\end{array}
\begin{array}{l}
        \IFN \ \EQ{\fieldnode}{\atlanta} \ \THENN \\ \quad
            \WEIGH{\failure{1.5}}{\ITER{(\pol{\atlanta})} \SEQN \NOTEQ{\fieldnode}{\atlanta}} \\
      \ELSEN \ \dots
\end{array}
\quad\quad\quad
\begin{array}{rcl}
    \abilene_{\relstr} &\triangleq& \ITER{(\SEQ{\pol{\relstr}}{\DUP})}
\end{array}
\]

Now assume we wish to check that all tunneled paths between $\bay$ and $\nyc$ have a failure rate of at most $10\%$. We use \Cref{thm:verif-safety} to decide whether the following is $0.1$-safe:
\[
    \SEQ{(\AND{\EQ{\fieldnode}{\bay}}{\EQ{\fielddst}{\nyc}})}{
        \SEQ{\abilene_{\relstr}}
        {(\AND{\EQ{\fieldnode}{\nyc}}{\NOTEQ{\fieldtid}{0}})}}
\]

The decision procedure from \Cref{thm:verif-safety} answers negatively and provides a witness: $\kansas$ is configured to always forward $\nyc$-bound traffic through tunnel 3 ($\kansas \edge \houston \edge \atlanta$). However, when a packet at $\kansas$ has just exited tunnel~$2$ ($\bay \edge \la \edge \houston \edge \kansas$), $\kansas$ will unnecessarily reroute traffic through $\houston$ and incur more probability of failure. As a result, the tunneled path $2\edge3\edge5$ between $\bay$ and $\nyc$ will have a failure rate of $10.3\%$. We remark that this is no longer a straightforward cycle detection as discussed in \Cref{sec:model-quant-prop}. \wnetkat enables the automatic identification of a specific combination of network tunnels that cause reliability issues.

A network provider can now use the generated witness to reconfigure $\kansas$ to only forward packets destined for $\nyc$ through tunnel~$3$ if they have not just exited tunnel~$2$:
\[
\begin{array}{rc}
    \pol{\kansas,\mathsf{safe}} & \triangleq
\end{array}
\begin{array}{l}
\IFN \ \EQ{\fieldtid}{0} \ \THENN \
    \dots \\
\ELSEN \ \IFN \ \AND{\EQ{\fieldtid}{2}}{\EQ{\fielddst}{\nyc}} \ \THENN \
    \ASSN{\fieldtid}{4} \quad \text{\textcolor{commentgreen}{\texttt{(* Forward directly *)}}}
\\
\ELSEN \ \IFN \ \OR{\EQ{\fieldtid}{1}}{\EQ{\fieldtid}{2}} \ \THENN \
    \ASSN{\fieldtid}{0} \\
\ELSEN \ \dots
\end{array}
\]
In particular, packets exiting tunnel~$2$ which are destined for $\nyc$ are now directly forwarded through tunnel 4. We can once again use the decision procedure from \Cref{thm:verif-safety}, which this time answers positively: all tunneled paths guarantee a failure rate of at most $10\%$.

\subsection{Finding High-Bandwidth Tunneled Paths with \Cref{thm:verif-reach}}
\noindent
We now demonstrate how \wnetkat can aid network providers as a \emph{design tool} when configuring the routing behavior: With \wnetkat, we can synthesize paths within a network satisfying specific quantitative properties; after which the network can be reconfigured appropriately.

Consider the following scenario: although there are several tunneled paths from $\bay$ to $\nyc$, we wish to refine the network to use high-bandwidth tunnels specifically for \emph{video traffic},
which should be delivered with at least $1000 \mbps$ of \textcolor{bandwidthpurple}{bandwidth}. We modify
the encoding of the network from the previous section (for which it is already guaranteed that all tunneled paths from $\bay$ to $\nyc$ have a failure rate of at most $10\%$)
 by weighting the \emph{forwarding actions} in each tunnel by the corresponding link's \textcolor{bandwidthpurple}{bandwidth}, instantiating \wnetkat with the Bottleneck semiring (cf. \Cref{tab:weighted-netkat-instances}):
\[
\begin{array}{rc}
    \pol{\kansas,\bandstr} & \triangleq
\end{array}
\begin{array}{l}
    \IFN \ \dots \\
    \ELSEN \ \IFN \ \EQ{\fieldtid}{3} \ \THENN \
        \WEIGH{\bandwidth{1250}}{\ASSN{\fieldnode}{\houston}} \\
    \ELSEN \ \IFN \ \EQ{\fieldtid}{4} \ \THENN \
        \WEIGH{\bandwidth{1750}}{\ASSN{\fieldnode}{\indiana}} \\
    \ELSEN \
        \DROP
\end{array}
\
\begin{array}{rcl}
    \pol{\bandstr} &\triangleq& \dots \\
    \\
    \abilene_{\bandstr} &\triangleq& \ITER{(\pol{\bandstr} \SEQN \DUP)}
\end{array}
\]
If it exists, we can now synthesize a tunneled path between $\bay$ and $\nyc$ with a bandwidth of at least $1000\mbps$. By \Cref{thm:verif-reach}, this can be decided by checking the $1000$-reachability of
\[
    \SEQ{(\AND{\EQ{\fieldnode}{\bay}}{\EQ{\fielddst}{\nyc}})}{
    \SEQ{\abilene_\bandstr}
        {(\AND{\EQ{\fieldnode}{\nyc}}{\NOTEQ{\fieldtid}{0}})}}~.
\]
The decision procedure from \Cref{thm:verif-reach} answers positively and provides the sought-after tunneled path:  $1\edge3\edge5$ has a bandwidth of $1250\mbps$. 
We now reconfigure the network to always choose this tunneled path for video traffic (identified via the field $\fieldvid$) destined for $\nyc$, e.g.,
\[
\begin{array}{rc}
    \pol{\kansas,\vidstr} & \triangleq
\end{array}
\begin{array}{l}
    \IFN \ \EQ{\fieldtid}{0} \ \THENN \\ \quad
        \IFN \ \EQ{\fielddst}{\nyc} \ \THENN \ (
            \IFN \ \EQ{\fieldvid}{\consttrue} \ \THENN \
                \ASSN{\fieldtid}{3} \
            \ELSEN \ (\ADD{\ASSN{\fieldtid}{3}}{\ASSN{\fieldtid}{4}}))
            \\ \quad
        \ELSEN \ \dots \\
    \ELSEN \ \dots
\end{array}
\]
In particular, video traffic ($\EQ{\fieldvid}{\consttrue}$) destined for $\nyc$ is now always tunneled through tunnel~$3$. $\bay$ would similarly be reconfigured to always tunnel $\nyc$-bound video traffic through tunnel~$1$. The forwarding behavior for regular traffic remains unchanged from our previous example.

In summary, we first used \wnetkat to verify that the tunneled paths configured in the Abilene network provide a reliability of at least $90\%$; then we used \wnetkat to refine our network configuration to use high-bandwidth tunneled paths for video traffic. We focus on reliability and bandwidth in these examples, but our framework (and decision procedures) remains parametric on a semiring and can verify several network phenomena as shown in \Cref{tab:weighted-netkat-instances}. For example, we can similarly verify that for the final network configuration above, all \emph{video traffic} in these tunneled paths is \emph{additionally} delivered within $20\millis$. Instantiating \wnetkat with the Arctic semiring and weighting by \textcolor{latencygreen}{latency}, this would correspond to checking that the following policy is $20$-safe:
\[
    \SEQ{(\AND{\AND{\EQ{\fieldnode}{\bay}}{\EQ{\fielddst}{\nyc}}}{\EQ{\fieldvid}{\consttrue}})}{
    \SEQ{\abilene_\vidstr}
        {(\AND{\EQ{\fieldnode}{\nyc}}{\NOTEQ{\fieldtid}{0}})}}
\]
As before, this property is decidable by \Cref{thm:verif-safety}; the decision procedure answers in the positive.

\section{Related Work}
\label{sec:related-work}

\wnetkat is an extension of \netkat~\cite{netkat}, which is itself an extension of
\kat~\cite{Kozen97} to reason about network behavior. \wnetkat is inspired in
large part by \probnetkat~\cite{probnetkat,cantor}, which in turn extends \netkat to
model probabilistic network behavior. \wnetkat generalizes this idea to support
modeling several different quantitative behaviors based on the choice of
semiring. Nevertheless, \wnetkat remains a conservative extension by subsuming
both \netkat and the guarded fragment of \probnetkat. Although \wnetkat
subsumes only a fragment of \probnetkat, an even smaller fragment (guarded and
$\DUP$-free) has previously been studied in a practical
setting~\cite{mcnetkat}.
\citet{costinternetkat} and \citet{wnetkat} both introduce extensions to
\netkat parametric on semirings to model latency and other quantitative
properties. However, neither extension provides a sound translation to weighted
automata nor decision procedures for verifying quantitative network
properties. Both extensions are
more expressive at the syntax-level than \wnetkat (e.g. including quantitative
tests), our extension instead is intentionally chosen so that the syntax is
expressive enough to model interesting network behavior while still being able
to use techniques based on weighted automata.

More generally, several frameworks have been proposed for reasoning about
programming systems that are parametric on
semirings (e.g., see~\cite{provenance-semi,core-coeff-calc,weighted-prog}).
In particular,
\citet{weighted-prog} propose \emph{weighted programming} as a paradigm for
specifying mathematical models beyond probability distributions. \wnetkat
follows a similar approach (and is likewise parametric on $\omega$-continuous
semirings), though our focus is specifically on extending the power of \netkat
to reasoning about the quantitative behavior of networks. Many of our example
semirings, however, are based on their applications in weighted
programming. \wgkat~\cite{wgkat} and \kawt~\cite{kawt} are both extensions to
\kat that likewise follow a similar approach to the work by
\citet{weighted-prog}. \citet{wgkat} extend Guarded Kleene Algebra with Tests
(\gkat)~\cite{gkat} to the weighted setting and show the decidability of
equivalence for weighted automata up to bisimilarity. Unfortunately, we cannot
apply their results as \wnetkat features unguarded iteration; we refer the
reader to the work of \citet{guarded-netkat} for the incompatiblity of \netkat
and \gkat. \citet{sedlar2024completeness} shows a completeness result for a
more general version of \kawt (Kleene Algebra with Weights and Tests). \wnetkat
is most similar to \kawt (albeit in the setting of \netkat which introduces
further subtleties as we discuss throughout the paper). However, \kawt does not provide a general language model and computable operational semantics, limiting this only to finite semirings.

Decision procedures for weighted automata have been studied
extensively in the literature~\cite{10.5555/646246.684713,KROB1994231}. Our decision procedures
for $\wta$-safety and $\wta$-reachability for \wnetkat automata adapt
results by \citet{decidable-weighted} for the Tropical semiring
over the natural numbers. Recent work by \citet{katch} contributed techniques for efficient implementations of decision procedures over \emph{\netkat automata} (which \wnetkat automata subsume).
In particular, they develop \emph{symbolic} versions of \netkat automata that do not explicitly enumerate their packet space and provide an ecosystem of supporting algorithms. Their approach is, however, highly tailored to deciding \emph{equivalence}, a property that is (i) undecidable in general for weighted automata (ii) often too strong when it comes to verifying quantitative network properties. Applying these techniques in our setting is fundamentally different as it would require not only the development of a novel, symbolic representation of our
\wnetkat automata but also direct symbolic decision procedures for $\wta$-safety and $\wta$-reachability rather than equivalence. 

Finally, semirings have been used as
the foundation of other frameworks in the networking domain, including network
calculus~\cite{network-calculus} and routing algebras~\cite{routing-algebras,metarouting}. Network calculus is a mathematical framework designed for modeling and reasoning about
quantitative properties. It provides primitives for modeling the arrival, buffering, and departure of traffic in a deterministic queueing system and also models interactions between multiple flows. Unlike \wnetkat, the focus is more on pencil-and-paper proofs of performance bounds rather than automated verification of safety and reachabilty properties. Routing algebras model the behavior of distributed control-plane protocols like OSPF and BGP, whereas \wnetkat focuses on behavior at the data-plane level. Deepening the connections between these frameworks is an excellent direction for future work.

\section{Conclusion}
\label{sec:conclusion}

We introduced \wnetkat, a framework for quantitative
network verification. We developed a denotational semantics of \wnetkat and an equivalent 
 language model. We then presented an operational
semantics based on \wnetkat automata to compute the exact semantics of \wnetkat.
This enabled the design of decision procedures for reasoning about $\wta$-safety and $\wta$-reachability in networks.
We then used the
framework to reason about worst-/best-case network guarantees over a range of practical network phenomena in the setting of Abilene backbone network.

As future work, we would like to implement practical versions of these decision
procedures over efficient representations of \wnetkat
automata (e.g., as in \cite{katch,netkat-learning}). Separately, network
verification with \netkat requires having accurate models of such
systems; which can be tedious and error prone. This is only more true in the
weighted setting, and so we would like to explore learning for \wnetkat
automata~\cite{netkat-learning}. Finally, we would be interested in
extensions to \wnetkat that make the language more expressive. In particular,
we would like to consider variants with quantitative tests
(e.g., as in~\cite{costinternetkat,wnetkat}), which would allow expressing
network behavior dependent on quantities (e.g., load balancing).
These extensions would however complicate our
denotational and operational model, and importantly, would not allow us to
reduce our properties to \wnetkat automata. We leave them as a possible
direction for future work.

\begin{acks}
    We are grateful to our PLDI reviewers and shepherd who helped us improve our paper significantly. We also thank Thomas Lu for helpful discussions on our case study, as well as the Cornell PLDG and UCL PPLV group for their feedback on early drafts.
This material is based upon work supported by the Defense Advanced Research Projects Agency (DARPA) under Contract No. HR001125CE018 (Approved for public release; distribution is unlimited.).
Additionally,
this material is based upon work supported by the Naval Surface Warfare Center-Corona Division (NSWC-Corona) under Contract No. N6426725C7016. Any opinions, findings and conclusions or recommendations expressed in this material are those of the author(s) and do not necessarily reflect the views of the NSWC-Corona.
Additionally,
this work was supported by
    ERC grant Autoprobe (no. 101002697),
    NSF grant DGE–2139899,
    DFF project AuRoRA,
    and a Royal Society Wolfson fellowship,
    as well as a gift the VMware University Research Fund.
Finally, a generative AI software tool (Claude, Sonnet 4.6) was used for finding typographical errors in our paper and for help producing the TikZ code used in one figure.
\end{acks}

\bibliographystyle{plainnat}
\bibliography{refs}
\ifappendix
\clearpage
\appendix
\section{Preliminaries}
\label{appendix:monoids-semirings}
In order to give an appropriate semantic treatment to \wnetkat, we
restrict the semirings that weights are drawn from to a specific class of
semirings: $\omega$-continuous semirings. In this section, we provide
definitions for these semirings and associated constructs.

\subsection{$\omega$-cpos and $\omega$-continuous Functions}
This section reviews basic concepts from domain theory that are
part of the definition of $\omega$-continuous semirings. For a more detailed
account, we refer the reader to~\citet{domain-theory}.
\begin{definition}[Partial Orders]
	A \emph{partial order} is a structure
	\(
	\cpo \eeq (\cpodom,\, \cpoord)~,
	\)
    where $\cpodom$ is a set and where $(\cpoord) \subseteq \cpodom\times\cpodom$ is a binary relation on $\cpodom$ such that $(\cpoord)$ is \emph{reflexive},  \emph{anitsymmetric}, and \emph{transitive}.	
\end{definition}

Now let $\cposubset \subseteq \cpodom$. We say that $\cpoa\in\cpodom$ is an \emph{upper bound} of $\cposubset$, if $\cpoa'\cpoord'\cpoa$ for all $\cpoa'\in'\cposubset$. We say that $\cpoa$ is the \emph{least} upper bound (or \emph{supremum}) of $\cposubset$, if $\cpoa\cpoord\cpoa'$ for all upper bounds $\cpoa'$ on $\cposubset$. Antisymmetry of $(\cpoord)$ implies that, if the least upper bounds exists, then it is unique and we denote it by $\cposup \cposubset$.
\begin{definition}[$\omega$-Complete Partial Orders]
	An \emph{$\omega$-cpo} is a partial order
	$
		(\cpodom,\, \cpoord)
	$
	such that:
	\begin{enumerate}
		\item There exists a least element $\cpobot \in \cpodom$, i.e., for all $\cpoa\in\cpodom$, we have $\cpobot \cpoord \cpoa$.
		\item For every $\omega$-chain, $\cpochain = \{\cpoa_0 \cpoord \cpoa_1 \cpoord \ldots\}$ in $\cpoord$, the {supremum} $\cposup \cpochain$ of $\cpochain$ exists in $\cpodom$.
	\end{enumerate}
\end{definition}
\begin{definition}[Monotonic and Continuous Endomaps]
	Let $(\cpodom,\, \cpoord)$ be an $\omega$-cpo and let
	$\cpoendo \colon \cpodom\to\cpodom$
	 be an endomap. We say that $\cpoendo$ is \emph{monotonic}, if for all $\cpoa,\cpoa'\in\cpodom$,
	 \[\cpoa\cpoord\cpoa'\quad\text{implies}\quad\cpoendo(\cpoa)\cpoord \cpoendo(\cpoa')~.\]
	 Moreover, we say that $\cpoendo$ is \emph{$\omega$-continuous}, if $\cpoendo$ is monotonic and preserves suprema of $\omega$-chains, i.e.,
	 \begin{align*}
	 	\forall\omega\text{-$chains$}~\cpochain~\text{in $\cpodom$}\colon\quad
	 	 \cpoendo\left(\cposup \cpochain \right) \eeq \cposup \left\{ \cpoendo(\cpoa) ~|~ \cpoa\in\cpochain \right\}~.
	 \end{align*}
\end{definition}
\subsection{Monoids}
There are several definitions of ($\omega$-)continuous monoids and semirings in
the literature (see \cite{karner1992limits,esik2008iteration}), ours are taken
from~\citet{kuich2011algebraic}.
\begin{definition}[Monoids]
	\label{def:monoids}
    A \emph{monoid} is a structure
    \(
    \mon\eeq (\mondom,\, \semimul,\, \mulid)~,
    \)
     where $\mondom$ is a set, $\semimul \colon \mondom\times\mondom \to \mondom$, and $\mulid \in \mondom$ such that \emph{multiplication} $\semimul$ is associative and $\mulid$ is an identity w.r.t.\ $\semimul$.
    We call $\mon$ \emph{commutative} if additionally $\forall \mona_1,\mona_2\colon \mona_1\semimul \mona_2 = \mona_2\semimul\mona_1$ holds.
    Moreover, we call $\mon$ \emph{idempotent} if $\forall \mona \colon \mona \semimul \mona = \mona$.
\end{definition}
\begin{definition}[$\omega$-Continuous Monoids]
    \label{def:omega-monoids}
	An \emph{$\omega$-continuous monoid} is an $\omega$-cpo
	\( (\mon,\, \semiord)~, \)
	 where $\mon = (\mondom,\, \semiadd,\, \addid)$ is a commutative monoid such that
	 \begin{enumerate}
         \item\label{def:omega-monoids1} $(\semiord)$ is \emph{positive}, i.e., $\addid$ is the least element of $(\semiord)$,
	 	\item\label{def:omega-monoids2} $\semiadd$ is \emph{monotonic} in both arguments, i.e.,
            for all $\mona_1,\mona_2,\mona_3 \in \mondom$,
	 	\begin{align*}
	 		\mona_2 \semiord \mona_3
	 		\qquad\text{implies}\qquad
	 		\mona_1 \semiadd \mona_2 \semiord \mona_1 \semiadd \mona_3
	 	\end{align*}
	 	\item\label{def:omega-monoids3} $\semiadd$ is \emph{$\omega$-continuous} in both
            arguments, i.e.
            \footnote{Notice that monotonicty of $\semiadd$ implies that the depicted suprema
             are well-defined.},
        \begin{align*}
	 		\forall \text{$\omega$-\text{chains}}~\semichain= \{\mona_0 \semiord \mona_1 \semiord \ldots\} \colon
	 		    \forall \mona \in \mondom \colon \quad
                    \mona \semiadd \semisup \semichain =
                        \semisup[\mona'\in\semichain] \mona \semiadd \mona'
        \end{align*}
	 \end{enumerate}
\end{definition}

Every $\omega$-continuous monoid induces a summation operation defined on \emph{countable} index sets as follows: Let $\indexset$ be a countable set and let
$\{\semia_\indexx\}_{\indexx\in\indexset} \colon \indexset \to \mondom$ be a family in
$\mondom$ indexed by $\indexset$. Now let $N = \{0,1,\ldots\}$ such that
$\enum \colon N \to \indexset$ is an \emph{arbitrary} enumeration of $\indexset$. One defines
\[
	\semisum{\indexx \in \indexset} \mona_\indexx \eeq
        \cposup \left\{\semisumtop{j=0}{n} \mona_{\enum(j)} ~|~ n\in N \right\}~.
\]
This supremum exists because every $\omega$-continuous monoid satisfies
$\mona \semiord \mona \semiadd \mona'$ for all $\mona,\mona' \in \mondom$. Moreover, it
can be shown that the above notion of countable summation is well-defined because the
value of the supremum is independent of the chosen enumeration $\enum$.
\begin{lemma}[Properties of $\omega$-continuous monoids~\cite{karner1992limits}]
\label{def:omega-comp-monoids}
    Let $(\mon, \semiord)$ and let $\{\mona_\indexx\}_{\indexx \in \indexset}$
    be a countable indexed family in $\mondom$. We have:
    \begin{enumerate}
		\item\label{def:omega-comp-monoids1} If $\indexset = \emptyset$, then $\semisum{\indexx \in\indexset} \semia_\indexx = \addid$.
        \item\label{def:omega-comp-monoids2} If $\indexset = \{\indexx_1,\indexx_2\}$, then $\semisum{i \in\indexset} \semia_{\indexx} = \semia_{\indexx_1} \semiadd \semia_{\indexx_2}$.
		\item\label{def:omega-comp-monoids3} If $\indexsetb$ is a countable set and $\indexset_j$
            are countable for every $j\in\indexsetb$ such that\footnote{
                Here, and elsewhere throughout the paper, $\dot{\bigcup}$
                denotes \emph{disjoint} union.
            }
            $\indexset =  \dot{\bigcup}_{j \in\indexsetb} \indexset_j$, then
		\[
				\semisum{j\in\indexsetb} \Big( \semisum{\indexx \in \indexset_j} \semia_i \Big)
				\eeq
				\semisum{\indexx \in \indexset} \semia_i~.
		\]
    \end{enumerate}
    Note that \Cref{def:omega-comp-monoids1,def:omega-comp-monoids2,def:omega-comp-monoids3}
    are exactly the properties defining an $\omega$-complete monoid, therefore every
    $\omega$-continuous monoid is also an $\omega$-complete monoid.
\end{lemma}
\subsection{Semirings}
\begin{definition}[Semirings]
	\label{def:semirings}
	A \emph{semiring} is a structure
	\(
		\semi \eeq (\semidom,\, \semiadd,\, \semimul, \, \addid,\, \mulid)~,
	\)
	where $\semidom$ is a set, with binary operations $\semiadd,\semimul\colon\semidom\times\semidom \to \semidom$, and constants $\addid,\mulid\in \semidom$ such that
	\begin{enumerate}
		\item\label{def:semirings1} $(\semidom,\, \semiadd,\, \addid)$ is a commutative monoid,
		\item\label{def:semirings2} $(\semidom,\, \semimul,\, \mulid)$ is a (not necessarily commutative) monoid,
		\item\label{def:semirings3} \emph{multiplication} $\semimul$ distributes over \emph{addition} $\semiadd$, i.e., for all $\semia_1,\semia_2,\semia_3\in\semidom$,
		\[
			\semia_1 \semimul (\semia_2 \semiadd \semia_3) \eeq \semia_1 \semimul \semia_2 \semiadd \semia_1\semimul\semia_3
			\quad\text{and}\quad
			 (\semia_2 \semiadd \semia_3) \semimul \semia_1  \eeq  \semia_2 \semimul \semia_1  \semiadd \semia_3\semimul\semia_1~,
		\]
		where we assume throughout that $\semimul$ binds stronger than $\semiadd$.
		\item\label{def:semirings4} multiplying with $\addid$ is \emph{annihilating}, i.e., $\forall \semia\in\semidom\colon \semia\semimul\addid = \addid\semimul\semia = \addid$.
	\end{enumerate}
\end{definition}
\begin{definition}[$\omega$-Continuous Semirings]
	\label{def:omega-semirings}
	An \emph{$\omega$-continuous semiring} is an $\omega$-cpo
	\[ (\semi,\, \semiord)~, \]
	 where $\semi = (\semidom,\, \semiadd,\, \semimul, \, \addid,\, \mulid)$ is a semiring such that
	 \begin{enumerate}
        \item\label{def:omega-semirings1}
            $((\semidom,\, \semiadd,\, \addid), \semiord)$ is an $\omega$-continuous monoid,
	 	\item\label{def:omega-semirings2} $\semimul$ is \emph{monotonic} in both arguments,
	 	\item\label{def:omega-semirings3}$\semimul$ is \emph{$\omega$-continuous} in both
            arguments.
	 \end{enumerate}
\end{definition}
\begin{lemma}[Properties of $\omega$-continuous semirings]
\label{def:omega-comp-semirings}
    Let $(\semi, \semiord)$ and let $\{\semia_\indexx\}_{\indexx \in \indexset}$
    be a countable indexed family in $\semidom$. 
               Multiplication distributes over countable summation, i.e.,
               for all $\semia \in \semidom$,
               \[
                   \semia \semimul \semisum{\indexx \in\indexset} \semia_\indexx \eeq \semisum{\indexx \in\indexset} \semia\semimul \semia_\indexx
                   \quad\text{and}\quad
                   \left( \semisum{\indexx \in\indexset} \semia_\indexx\right) \semimul \semia \eeq \semisum{\indexx \in\indexset} \semia_\indexx\semimul \semia~.
               \]
    Note that every $\omega$-continuous semiring is also an $\omega$-complete
    semiring.
\end{lemma}

\begin{definition}[Computable Semirings]
	\label{def:computable_semiring}
	An $\omega$-continuous semiring $(\semi,\, \semiord)$  is \emph{computable}, if:
	\begin{enumerate}
		\item The domain $\semidom$ is a computable set,
		\item The operations $\semiadd, \semimul$, and the \emph{closure operation}
        \footnote{Every $\omega$-complete semiring is a so-called \emph{starsemiring}
        admitting this operation~\cite{Droste2009}.}
        $\semistar{(-)} \colon \semidom \to \semidom$ defined as
		\[
		\semistar{\semia} \eeq \semisum{n \in \nats} \semia^n
		\qquad
		\text{where}
		\qquad
		\semia^0 = \mulid~\text{and}~\semia^{n+1} = \semia \semimul \semia^n
		\]
		are computable, and
		\item the order $\semiord$ is a decidable relation. 
	\end{enumerate}
\end{definition}

\begin{notation}
    Going forward and throughout the rest of this paper, we adpot the
    notational convention of writing $\semi$ instead of $\semidom$
    when working over a semiring $\semi$.
\end{notation}

\begin{remark}
Throughout the appendix, we make use in our proofs of several
properties that are well-known in the literature for ($\omega$-continuous)
monoids and semirings (e.g., countable summation can be shown to be
associative). We refer the reader to~\cite{golan2013semirings,Droste2009} for
more detailed accounts.
\end{remark}

\section{Denotational semantics}
\subsection{Properties of Weightings}
\label{sec:weightings-props}
\begin{theorem}[Properties of weightings]
	\label{thm:weightings-props}
    For reference, we list here several properties for weightings that follow
    from $\semi$ being an $\omega$-continuous semiring.
	\begin{enumerate}
		\item \label{thm:weightings-props1}
		$(\weightings[\semi]{X}, \, \lsemiadd, \, \laddid),  \,\semiord)$ is an
		$\omega$-continuous monoid.
		\item \label{thm:weightings-props2}
		$(\weightingsfunctor[\semi], \eta, \bind)$ forms a monad,
		\item \label{thm:weightings-props3}
		Lifted scalar multiplication (from both sides) distributes over lifted addition,
		\item \label{thm:weightings-props4}
		Lifted semiring zero is an annihilator for lifted multiplication,
		\item \label{thm:weightings-props5}
		Lifted scalar multiplication (from both sides) distributes over countable summation,
        \item \label{lmm:zero-left-annihilate-bind}
            $\laddid$ is a left-annihilator for $(\bind)$,
        \item \label{lmm:bot-right-annihilate-bind}
            $\cpobot = \mylam{x} \laddid$ is a right-annihilator for $(\bind)$,
        \item \label{lmm:bind-mono-l} \label{lmm:bind-mono-r}
            $(\bind)$ is monotonic in both arguments,
        \item  \label{lmm:bind-omega-continuous-l} \label{lmm:bind-omega-continuous-r}
            $(\bind)$ is $\omega$-continuous in both arguments.
	\end{enumerate}
\end{theorem}

\subsection{Least Fixed Point Characterization of Iteration}
\label{sec:iter-lfp}

We wish to show that $\sem{\ITER{\pola}} $ is the least solution of
\[
	\sem{\ITER{\pola}} \eeq \sem{\ADD{\SKIP}{\SEQ{\pola}{\ITER{\pola}}}}~.
\]

To formalize this, define the $\omega$-cpo $(\histories \rightarrow \weightings{\histories}, \,\cpoord)$, where $\cpoord$ is obtained from lifting the order $\semiord$ on $\weightings{\histories}$ pointwise. Now, given a policy $\pola$, define the operator
\[
	 \Phi_{\pola}\colon (\histories \rightarrow \weightings{\histories}) \to (\histories \rightarrow \weightings{\histories}), \qquad
	 \Phi_{\pola}(\cpoa) \ttriangleq
	 \mylam{\hista}
	 \eta(\hista) \semiadd (\sem{\pola}(\hista) \bind \cpoa)~.
 \]
\begin{theorem}[Least Fixed Point Characterization of Iteration]
\label{thm:iter-lfp}
	For every policy $\pola$, we have
	 \[
        \sem{\ITER{\pola}} \eeq \lfp \Phi_{\pola}
        \qquad
        \text{where}
        \qquad
         \lfp \Phi_{\pola} \eeq \cposup[n \in \nats] \Phi_{\pola}^n(\cpobot)~.
	\]
\end{theorem}
\begin{proof}
    First, we have that $\lfp \Phi_{\pola}$ is $\cposup[n \in \nats] \Phi_{\pola}^n(\cpobot)$
    by the Kleene fixed-point theorem (the $\omega$-continuity of of $\Phi_{\pola}$ follows from
    \Cref{thm:weightings-props}). Then, we prove
    \[\forall n \colon\Phi_{\pola}^{n + 1}(\bot) =
    \mylam{\hista} \semisumtop{k=0}{n} \sem{ \NFOLD[k]{\pola} }(\hista)\] by induction on $n$
    (\Cref{lmm:phi-value-n}). The claim then follows by
    \[
    \semisumtop{k=0}{\infty} \sem{ \NFOLD[k]{\pola} }(\hista) =
        \cposup[n] \semisumtop{k=0}{n} \sem{ \NFOLD[k]{\pola} }(\hista)~.\]\qedhere
\end{proof}

\begin{lemma}
\label{lmm:phi-value-n}
    For all $n\in\nats$,
    $\Phi_{\pola}^{n + 1}(\bot) =
        \mylam{\hista} \semisumtop{k=0}{n} \sem{ \NFOLD[k]{\pola} }(\hista)$
\end{lemma}
\begin{proof} By induction on $n$. \\

    \emph{Case $n = 0$:}

    \[
    \begin{array}{rclr}
    &&
        \Phi_{\pola}(\bot) \\
    &=&
        \mylam{\hista}
            \unit(\hista) \lsemiadd (\sem{\pola} (\hista) \bind \bot)
    & \text{(By definition)} \\
    &=&
        \mylam{\hista}
            \unit(\hista) \lsemiadd \laddid
    & \text{(\Cref{lmm:bot-right-annihilate-bind})} \\
    &=&
        \mylam{\hista} \unit(\hista)
    & \text{(\Cref{thm:weightings-props}.\ref{thm:weightings-props1})} \\
    &=&
        \mylam{\hista} \sem{\SKIP}(\hista)
    & \text{(By definition)} \\
    &=&
        \mylam{\hista} \sem{\NFOLD[0]{\pola}}(\hista)
    & \text{(By definition)} \\
    &=&
        \mylam{\hista} \semisumtop{k=0}{0} \sem{ \NFOLD[k]{\pola} }(\hista)
    & \text{(By definition)}
    \end{array}
    \]

    \emph{Case $n = n + 1$:}

    \[
    \begin{array}{rclr}
    &&
        \Phi_{\pola}^{n + 2}(\bot) \\
    &=&
        \mylam{\hista}
            \unit(\hista) \lsemiadd (\sem{\pola} (\hista) \bind \Phi_{\pola}^{n + 1}(\bot))
    & \text{(By definition)} \\
    &=&
        \mylam{\hista}
            \unit(\hista) \lsemiadd
                \lsemisum{\histb \in \supp(\sem{\pola})}
                    \sem{\pola}(\hista)(\histb) \lsemimul \Phi_{\pola}^{n + 1}(\bot)(\histb)
    & \text{(By definition)} \\
    &=&
        \mylam{\hista}
            \unit(\hista) \lsemiadd
                \lsemisum{\histb \in \supp(\sem{\pola})}
                    \sem{\pola}(\hista)(\histb) \lsemimul
                        \lsemisumtop{k=0}{n} \sem{ \NFOLD[k]{\pola} }(\histb)
    & \text{(IH)} \\
    &=&
        \mylam{\hista}
            \unit(\hista) \lsemiadd
                \lsemisum{\histb \in \supp(\sem{\pola})}
                    \lsemisumtop{k=0}{n}
                        \sem{\pola}(\hista)(\histb) \lsemimul
                            \sem{ \NFOLD[k]{\pola} }(\histb)
    & \text{(\Cref{thm:weightings-props}.\ref{thm:weightings-props5})} \\
    &=&
        \mylam{\hista}
            \unit(\hista) \lsemiadd
                \lsemisumtop{k=0}{n}
                    \lsemisum{\histb \in \supp(\sem{\pola})}
                        \sem{\pola}(\hista)(\histb) \lsemimul
                            \sem{ \NFOLD[k]{\pola} }(\histb)
    & \text{(\Cref{thm:weightings-props}.\ref{thm:weightings-props1},
            assoc. of countable sum)} \\
    &=&
        \mylam{\hista}
            \unit(\hista) \lsemiadd
                \lsemisumtop{k=0}{n}
                    \sem{\pola}(\hista) \bind \sem{ \NFOLD[k]{\pola} }
    & \text{(By definition)} \\
    &=&
        \mylam{\hista}
            \unit(\hista) \lsemiadd
                \lsemisumtop{k=0}{n}
                    \sem{ \NFOLD[k + 1]{\pola} }(\hista)
    & \text{(By definition)} \\
    &=&
        \mylam{\hista}
            \unit(\hista) \lsemiadd
                \lsemisumtop{k=1}{n+1}
                    \sem{ \NFOLD[k]{\pola} }(\hista)
    & \text{(Simple arithmetic)} \\
    &=&
        \mylam{\hista}
            \sem{ \NFOLD[0]{\pola} }(\hista) \lsemiadd
                \lsemisumtop{k=1}{n+1}
                    \sem{ \NFOLD[k]{\pola} }(\hista)
    & \text{(By definition)} \\
    &=&
        \mylam{\hista}
                \lsemisumtop{k=0}{n+1}
                    \sem{ \NFOLD[k]{\pola} }(\hista)
    & \text{(Simple arithmetic)} \\
    \end{array}
    \]
\end{proof}

\section{Approximation of \wnetkat policies}
\label{appendix:approximation}
\begin{figure}[h]
	\begin{tabular}{ll}
		\begin{tabular}{l  l}
			\toprule
			\toprule
			$\boldsymbol{\pola}$ & $\boldsymbol{\approxpol{\pola}{n}}$\\
			\midrule
			$\pola$~primitive \quad \quad& $\pola$ \\[1ex]
			$\SEQ{\pola_1}{\pola_2}$ & $\SEQ{\approxpol{\pola_1}{n}}{\approxpol{\pola_2}{n}}$ \\[1ex]
			$\WEIGH{\wta}{\polb}$ & $\WEIGH{\wta}{\approxpol{\polb}{n}}$ \\[1ex]
			$\ADD{\pola_1}{\pola_2}$ & $\ADD{\approxpol{\pola_1}{n}}{\approxpol{\pola_2}{n}}$ \\[1ex]
			$\ITER{\pola}$ & $\SUMTOP{i=0}{n}{\NFOLD[i]{(\approxpol{\pola}{n})}}$ \\[1ex]
			\bottomrule
			\bottomrule
		\end{tabular}
		&\qquad
		\begin{tabular}{l}
			\begin{tabular}{rcl}
				$\pola$ primitive &::=&
				$\testa$ \\ &$\mid$&
				$\DUP$ \\ &$\mid$&
				$\ASSN{\fielda}{\vala}$
			\end{tabular}
			\medskip \\ \midrule \medskip
			\begin{tabular}{rcl}
				$\SUMTOP{i=0}{0}{\pola_i}$ &$\triangleq$&
				$\DROP$ \medskip \\
				$\SUMTOP{i=0}{n+1}{\pola_i}$ &$\triangleq$&
				$\ADD{\pola_0}{\SUMTOP{i=0}{n}{\pola_{i+1}}}$
			\end{tabular}
		\end{tabular}
	\end{tabular}
	\caption{Approximants of policies in \wnetkat. Recall that $\NFOLD[i]{{\pola}}$ is defined in \Cref{fig:syntax-semantics}.}
	\label{tab:approx}
\end{figure}
We now formalize a useful notion of approximants of \swnetkat policies $\pola$, which enable us to
effectively\footnote{Under the mild condition that the semiring $\semi$ is computable
(cf.\ \Cref{def:computable_semiring}).}  approximate the semantics $\sem{\pola}$ of $\pola$ with increasing precision. Besides this, approximants will enable us to prove that \wnetkat subsumes a rich fragment of \probnetkat \cite{probnetkat,cantor}  --- a probabilistic extension of \netkat.

The idea is to generalize the notion by \citet{cantor}: Determining the semantics of policies involving iteration generally requires evaluating countable sums (cf.\ \Cref{fig:syntax-semantics}) such as 
\[
\sem{ \ITER{\pola} } (\hista)(\hista')  \eeq
\semisum{i \in \nats} \sem{ \NFOLD[i]{\pola} } (\hista)(\hista')~.
\]
However, whenever we cut such a countable sum off at some index --- thereby turning it into a \emph{finite} sum --- we soundly \emph{underapproximate} it and, in the limit, recover the entire sum, i.e.,
\[
\forall n \in \nats \colon \semisumtop{i =0}{n} \sem{ \NFOLD[i]{\pola} } (\hista)(\hista')
\semiord \semisum{i \in \nats} \sem{ \NFOLD[i]{\pola} } (\hista)(\hista')
\quad\text{and}\quad
\cposup[n\in\nats] \semisumtop{i =0}{n} \sem{ \NFOLD[i]{\pola} } (\hista)(\hista')
= 
\semisum{i \in \nats} \sem{ \NFOLD[i]{\pola} } (\hista)(\hista')~.
\]
On a syntactic level, cutting off countable sums can be realized by replacing all iterations appearing in a given policy $\pola$ by their $n$-th unrolling, which yields the notion of $n$-th approximants:

\begin{definition}[Approximants]
	Let $\pola$ be a policy and $n\in\nats$. We define the
	\emph{$n$-th approximant $\approxpol{\pola}{n}$ of $\pola$} recursively on
	the structure of $\pola$ as shown in \Cref{tab:approx}. 
\end{definition}
We often write $\approxn[n]{\sem{\pola}}$ instead  of $\sem{\approxpol{\pola}{n}}(\hista)$. It is important to notice that $\approxpol{\pola}{n}$ recursively unrolls \emph{all} (including \emph{nested}) iterations. Therefore $\approxpol{\pola}{n}$ does no longer contain iterations so that $\approxn[n]{\sem{\pola}}(\hista)$ can indeed be computed for every $n\in\nats$ by recursively applying the rules given in \Cref{fig:syntax-semantics}. 

Towards establishing soundness and completeness of approximants, we first observe that they give rise to an $\omega$-chain w.r.t.\ the order $\cpoord$ on $\histories \rightarrow \weightings{\histories}$:
\begin{rlemma}{lmm:approx-chain}
	The approximants of a policy $\pola \in \pols$ form an $\omega$-chain w.r.t.
	$(\cpoord)$, i.e.
	\[
	\approxn[0]{\sem{\pola}} \cpoord
	\approxn[1]{\sem{\pola}} \cpoord
	\approxn[2]{\sem{\pola}} \cpoord
	\dots
	\]
	More formally, we have:
	$\forall n \in \nats \colon
	\approxn[n]{\sem{\pola}} \cpoord
	\approxn[n+1]{\sem{\pola}}$.
\end{rlemma}
\begin{proof} The proof proceeds by structural induction on $\pola$,
    all cases follow from monotonicity w.r.t. $\semiord$ of operations
    at the level of weightings: 

    \emph{Case $\pola$ primitive:}     We have that $\approxn[n]{\sem{\pola}} = \sem{\pola}$ for all $n \in \nats$,
        therefore all approximants form a chain by reflexivity of $(\cpoord)$. 

     \emph{Case $\pola = \SEQ{\pola}{\polb}$:}
        \[
        \begin{array}{rclr}
            \approxn[n]{\sem{\SEQ{\pola}{\polb}}}(\hista)
        &=&
            \approxn[n]{\sem{\pola}}(\hista) \bind \approxn[n]{\sem{\polb}}
        & \text{(By definition)} \\
        &\cpoord&
            \approxn[n+1]{\sem{\pola}}(\hista) \bind \approxn[n]{\sem{\polb}}
        & \text{(\Cref{thm:weightings-props}.\ref{lmm:bind-mono-l} and IH)} \\
        &\cpoord&
            \approxn[n+1]{\sem{\pola}}(\hista) \bind \approxn[n+1]{\sem{\polb}}
        & \text{(\Cref{thm:weightings-props}.\ref{lmm:bind-mono-r} and IH)} \\
        &=&
            \approxn[n+1]{\sem{\SEQ{\pola}{\polb}}}(\hista)
        & \text{(By definition)}
        \end{array}
        \]

    \emph{Case $\pola = \WEIGH{\wta}{\pola}$:} By \Cref{def:omega-semirings}.\ref{def:omega-semirings2} and IH. 

    \emph{Case $\pola = \ADD{\pola}{\polb}$:} By \Cref{def:omega-monoids}.\ref{def:omega-monoids2} and IH. 

    \emph{Case $\pola = \ITER{\pola}$:}
        \[
        \begin{array}{rclr}
        &&
            \approxn[n]{\sem{\ITER{\pola}}} \\
        &=&
            \sem{\SUMTOP{i=0}{n}{\NFOLD[i]{(\approxpol{\pola}{n})}}}
        & \text{(By definition)} \\
        &=&
            \semisumtop{i=0}{n} \sem{\NFOLD[i]{(\approxpol{\pola}{n})}}
        & \text{(By definition)} \\
        &\cpoord&
            \sem{\NFOLD[n+1]{(\approxpol{\pola}{n})}} \semiadd
                \semisumtop{i=0}{n} \sem{\NFOLD[i]{(\approxpol{\pola}{n})}}
        & \text{($\semia \semiord \semia \semiadd \semib$ for any $\omega$-cont. monoid)} \\
        &=&
            \semisumtop{i=0}{n+1} \sem{\NFOLD[i]{(\approxpol{\pola}{n})}}
        & \text{(Simple arithmetic)} \\
        &\cpoord&
            \semisumtop{i=0}{n+1} \sem{\NFOLD[i]{(\approxpol{\pola}{n+1})}}
        & \text{(\Cref{def:omega-monoids}.\ref{def:omega-monoids2},
            \Cref{lmm:nth-fold-mono}, and IH)} \\
        &=&
            \sem{\SUMTOP{i=0}{n+1}{\NFOLD[i]{(\approxpol{\pola}{n+1})}}}
        & \text{(By definition)} \\
        &=&
            \approxn[n+1]{\sem{\ITER{\pola}}}
        & \text{(By definition)}
        \end{array}
        \]
\end{proof}
Exploiting $\omega$-continuity of the semiring operations (cf.\ \Cref{fig:semiring}) then yields the following:
\begin{rtheorem}{thm:sup-approx-equiv}[Soundness and Completeness of Approximants]
	The appoximants of a policy $\pola$ soundly underapproximate the semantics of $\pola$ and, in the limit, yield precisely the semantics of $\pola$, i.e., 
	\[
	\forall n \in \nats\colon \approxn[n]{\sem{\pola}} \cpoord \sem{\pola}
	\qquad\text{and}\qquad
	\cposup[n\in\nats] \approxn[n]{\sem{\pola}}
	= 
	\sem{\pola}~.
	\]
\end{rtheorem}
\begin{proof} By \Cref{lmm:approx-chain}, the supremum
    $\cposup[n\in\nats] \approxn[n]{\sem{\pola}}$ is well-defined.
    The rest of the proof proceeds by structural induction on $\pola$, all
    cases follow from $\omega$-continuity w.r.t. $\semiord$ of operations at
    the level of weightings: \\

    \emph{Case $\pola$ primitive:} We have that $\approxn[n]{\sem{\pola}} = \sem{\pola}$ for all $n \in \nats$,
        therefore the supremum of all approximants is $\sem{\pola}$.

     \emph{Case $\pola = \SEQ{\pola}{\polb}$:}
        \[
        \begin{array}{rclr}
            \cposup[n \in \nats] \approxn[n]{\sem{\SEQ{\pola}{\polb}}}(\hista)
        &=&
            \cposup[n \in \nats]
                \approxn[n]{\sem{\pola}}(\hista) \bind \approxn[n]{\sem{\polb}}
        & \text{(By definition)} \\
        &=&
            \cposup[m \in \nats] \cposup[n \in \nats]
                \approxn[n]{\sem{\pola}}(\hista) \bind \approxn[m]{\sem{\polb}}
        & \text{(\cite[Proposition 2.1.12]{domain-theory})} \\
        &=&
            \cposup[m \in \nats]
                \left( \cposup[n \in \nats]
                    \approxn[n]{\sem{\pola}}(\hista) \right) \bind \approxn[m]{\sem{\polb}}
        & \text{(\Cref{thm:weightings-props}.\ref{lmm:bind-omega-continuous-l})} \\
        &=&
            \left( \cposup[n \in \nats]
                \approxn[n]{\sem{\pola}}(\hista) \right) \bind
                    \cposup[m \in \nats] \approxn[m]{\sem{\polb}}
        & \text{(\Cref{thm:weightings-props}.\ref{lmm:bind-omega-continuous-r})} \\
        &=&
            \sem{\pola}(\hista) \bind
                    \cposup[m \in \nats] \approxn[m]{\sem{\polb}}
        & \text{(IH for $\pola$)} \\
        &=&
            \sem{\pola}(\hista) \bind \sem{\polb}
        & \text{(IH for $\polb$)}
        \end{array}
        \]

    \emph{Case $\pola = \WEIGH{\wta}{\pola}$:} By \Cref{def:omega-semirings}.\ref{def:omega-semirings3} and IH.

    \emph{Case $\pola = \ADD{\pola}{\polb}$:} By \Cref{def:omega-monoids}.\ref{def:omega-monoids3} and IH.

    \emph{Case $\pola = \ITER{\pola}$:}
        \[
        \begin{array}{rclr}
        &&
            \sem{\ITER{\pola}}(\hista) \\
        &=&
            \semisum{m \in \nats} \sem{\NFOLD[m]{\pola}}(\hista)
        & \text{(By definition)} \\
        &=&
            \semisum{m \in \nats}
                \cposup[n \in \nats]
                    \sem{\NFOLD[m]{(\approxpol{\pola}{n})}}(\hista)
        & \text{(\Cref{lmm:nth-fold-sup} and IH)} \\
        &=&
            \cposup[m \in \nats]
                \semisumtop{i=0}{m}
                    \cposup[n \in \nats]
                        \sem{\NFOLD[i]{(\approxpol{\pola}{n})}}(\hista)
        & \text{(By definition)} \\
        &=&
            \cposup[m \in \nats]
                \cposup[n \in \nats]
                    \semisumtop{i=0}{m}
                        \sem{\NFOLD[i]{(\approxpol{\pola}{n})}}(\hista)
        & \text{(\Cref{def:omega-monoids}.\ref{def:omega-monoids3})} \\
        &=&
            \cposup[n \in \nats]
                \semisumtop{i=0}{n}
                    \sem{\NFOLD[i]{(\approxpol{\pola}{n})}}(\hista)
        & \text{(\cite[Proposition 2.1.12]{domain-theory})} \\
        &=&
            \cposup[n \in \nats]
                \sem{\SUMTOP{i=0}{n}{\NFOLD[i]{(\approxpol{\pola}{n})}}}(\hista)
        & \text{(By definition)} \\
        &=&
            \cposup[n \in \nats]
                \approxn[n]{\sem{\ITER{\pola}}}
        & \text{(By definition)}
        \end{array}
        \]
\end{proof}

\subsection{Alternate Definiton of Approximation for While Loops}
We use the following alternate definition for the $n$-th approximant of a while
loop in \wnetkat throughout the rest of the paper.
\[
\begin{array}{rcl}
	\approxpol{\WHILE[\testa]{\pola}}{n} &\triangleq&
	\NFOLD[n,\testa]{(\approxpol{\pola}{n})}
\end{array}
\quad \text{where} \quad
\begin{array}{rcl}
	\NFOLD[0,\testa]{\pola} &\triangleq& \NOT{\testa} \\
	\NFOLD[n+1,\testa]{\pola} &\triangleq&
	\IF[\testa]{\SEQ{\pola}{\NFOLD[n,\testa]{\pola}}}{\SKIP} \\
\end{array}
\]

This definition will be convenient in the following section, where we relate
the semantics of \wnetkat to the guarded fragment of \probnetkat. Note
that although we use an alternate definition, we are able to show that it is
equivalent to the definition given in \Cref{tab:approx}.

\begin{lemma}
\label{lmm:alt-approx-equiv}
	Let $\altaltapproxpol{\pola}{n}$ denote the $n$-th approximant
	for a policy $\pola$ with the alternate definition for guarded
	iteration given above. This definition is equivalent to the one
	in \Cref{tab:approx}:
	\[
	\begin{array}{lrcl}
		\forall n \in \nats \colon &
		\approxn[n]{\sem{\pola}} &=&
		\sem{\altaltapproxpol{\pola}{n}}
	\end{array}
	\]
\end{lemma}
\begin{proof}
    By structural induction on $\pola$. All cases follow immediately except for
    guarded iteration, as the definitions coincide exactly. We consider only
    this case. \\

    \emph{Case $\pola = \WHILE[\testa]{\pola}$:}         We begin by unfolding the LHS/RHS:
        \[
        \begin{array}{rclr}
            \approxn[n]{\sem{\WHILE[\testa]{\pola}}}
        &=&
            \sem{\approxpol{\WHILE[\testa]{\pola}}{n}} \\
        &=&
            \sem{\approxpol{\SEQ{\ITER{(\SEQ{\testa}{\pola})}}
                                {\NOT{\testa}}}{n}}
        & \text{(Definition in \Cref{tab:approx})} \\
        &=&
            \sem{\SEQ{\left( \SUMTOP{i=0}{n}
                                {\NFOLD[i]{(\approxpol{\SEQ{\testa}{\pola}}{n})}\right)}}
                     {\NOT{\testa}}}
        & \text{(Definition in \Cref{tab:approx})} \\
        &=&
            \sem{\SEQ{\left( \SUMTOP{i=0}{n}
                                {\NFOLD[i]{(\SEQ{\testa}{\approxpol{\pola}{n}})}\right)}}
                     {\NOT{\testa}}}
        & \text{(Definition in \Cref{tab:approx})} \\\\

        \sem{\altaltapproxpol{\WHILE[\testa]{\pola}}{n}}
        &=&
            \sem{\NFOLD[n,\testa]{(\altaltapproxpol{\pola}{n})}}
        & \text{(Alternate definition)}
        \end{array}
        \]

        Finally, we have:
        \[
        \begin{array}{rclr}
            \sem{\SEQ{\left( \SUMTOP{i=0}{n}
                                {\NFOLD[i]{(\SEQ{\testa}{\approxpol{\pola}{n}})}\right)}}
                     {\NOT{\testa}}}
        &=&
            \sem{\NFOLD[n,\testa]{(\altaltapproxpol{\pola}{n})}}
        & \text{(\Cref{lmm:alt-defn-guarded-iteration-equiv} and IH)}
        \end{array}
        \]
\end{proof}

\subsection{Lemmas for Soundness and Completeness of Approximants}

\begin{lemma}[Monotonicity of $n$-th fold]
\label{lmm:nth-fold-mono}
Given policies $\pola_1$, $\pola_2$, we have:
    \[
    \begin{array}{rcl}
        \sem{\pola_1} \cpoord \sem{\pola_2} &\Rightarrow&
            \sem{\NFOLD[n]{\pola_1}} \cpoord \sem{\NFOLD[n]{\pola_2}}
    \end{array}
    \]
\end{lemma}
\begin{proof}
    By induction on $n$.

    \emph{Case $n = 0$:}
    $$
        \sem{\NFOLD[0]{\pola_1}}
    =\sem{\SKIP}
    =
        \sem{\NFOLD[0]{\pola_2}}
    \cpoord
        \sem{\NFOLD[0]{\pola_2}}
    $$

    \emph{Case $n = n + 1$:}
    \[
    \begin{array}{rclr}
        \sem{\NFOLD[n+1]{\pola_1}}(\hista)
    &=&
        \sem{\SEQ{\pola_1}{\NFOLD[n]{\pola_1}}}(\hista)
    & \text{(By definition)} \\
    &=&
        \sem{\pola_1}(\hista) \bind \sem{\NFOLD[n]{\pola_1}}
    & \text{(By definition)} \\
    &\cpoord&
        \sem{\pola_2}(\hista) \bind \sem{\NFOLD[n]{\pola_1}}
    & \text{(\Cref{thm:weightings-props}.\ref{lmm:bind-mono-l} and assumption)} \\
    &\cpoord&
        \sem{\pola_2}(\hista) \bind \sem{\NFOLD[n]{\pola_2}}
    & \text{(\Cref{thm:weightings-props}.\ref{lmm:bind-mono-r} and IH)} \\
    &=&
        \sem{\NFOLD[n+1]{\pola_2}}(\hista)
    & \text{(By definition)} 
    \end{array}
    \]
\end{proof}

\begin{lemma}
\label{lmm:nth-fold-sup}
    Given policy $\pola$, we have:
    \[
    \begin{array}{lrcl}
        \forall m \in \nats \colon
            \sem{\pola} = \cposup[n \in \nats] \approxn[n]{\sem{\pola}} &\Rightarrow&
                \sem{\NFOLD[m]{\pola}} =
                    \cposup[n \in \nats] \sem{\NFOLD[m]{(\approxpol{\pola}{n})}}
    \end{array}
    \]
\end{lemma}
\begin{proof}
    By induction on $m$.

    \emph{Case $m=0$:}
    $$
        \cposup[n \in \nats] \sem{\NFOLD[0]{(\approxpol{\pola}{n})}}
    =
        \cposup[n \in \nats] \sem{\SKIP}
    =
        \sem{\SKIP}
    =
        \sem{\NFOLD[0]{\pola}}
    $$

    \emph{Case $m=m+1$:}
    \[
    \begin{array}{rclr}
    &&
        \cposup[n \in \nats] \sem{\NFOLD[m+1]{(\approxpol{\pola}{n})}}(\hista) \\
    &=&
        \cposup[n \in \nats]
            \approxn[n]{\sem{\pola}}(\hista) \bind
                \sem{\NFOLD[m]{(\approxpol{\pola}{n})}}
    & \text{(By definition)} \\
    &=&
        \cposup[n_2 \in \nats]
            \cposup[n_1 \in \nats]
                \approxn[n_1]{\sem{\pola}}(\hista) \bind
                    \sem{\NFOLD[m]{(\approxpol{\pola}{n_2})}}
    & \text{(\cite[Proposition 2.1.12]{domain-theory})} \\
    &=&
        \cposup[n_2 \in \nats]
            \left( \cposup[n \in \nats]
                \approxn[n]{\sem{\pola}}(\hista) \right) \bind
                    \sem{\NFOLD[m]{(\approxpol{\pola}{n_2})}}
    & \text{(\Cref{thm:weightings-props}.\ref{lmm:bind-omega-continuous-l})} \\
    &=&
        \left( \cposup[n \in \nats]
            \approxn[n]{\sem{\pola}}(\hista) \right) \bind
                \cposup[n \in \nats]
                    \sem{\NFOLD[m]{(\approxpol{\pola}{n})}}
    & \text{(\Cref{thm:weightings-props}.\ref{lmm:bind-omega-continuous-r})} \\
    &=&
        \sem{\pola} \bind
            \cposup[n \in \nats]
                \sem{\NFOLD[m]{(\approxpol{\pola}{n})}}
    & \text{(Assumption)} \\
    &=&
        \sem{\pola} \bind \sem{\NFOLD[m]{\pola}}
    & \text{(IH)} \\
    &=&
        \sem{\NFOLD[m + 1]{\pola}}
    & \text{(By definition)}
    \end{array}
    \]
\end{proof}

\begin{lemma}[Equivalence of alternate definition of guarded iteration
    for \wnetkat]
\label{lmm:alt-defn-guarded-iteration-equiv}
Given policies $\pola_1$, $\pola_2$ and predicate $\testa$ in \wnetkat, we
have:
\[
\begin{array}{lrcl}
    \forall n \in \nats \colon &
        \sem{\pola_1} = \sem{\pola_2} &\Rightarrow&
            \sem{\SEQ{\left( \SUMTOP{i=0}{n}
                                {\NFOLD[i]{(\SEQ{\testa}{\pola_1})}\right)}}
                     {\NOT{\testa}}}
            =
                \sem{\NFOLD[n,t]{\pola_2}}

\end{array}
\]
\end{lemma}
\begin{proof}
    By induction on $n$.

\emph{Case $n=0$:}
    \[
    \begin{array}{rclr}
        \sem{\SEQ{\left( \SUMTOP{i=0}{0}
                            {\NFOLD[i]{(\SEQ{\testa}{\pola_1})}}\right)}
                 {\NOT{\testa}}}
    &=& \sem{\SEQ{\SKIP}{\NOT{\testa}}}
    & \text{(By definition)} \\
    &=& \sem{\NOT{\testa}}
    & \text{(\Cref{thm:weightings-props}.\ref{thm:weightings-props2})} \\
    &=& \sem{\NFOLD[0,t]{\pola_2}}
    \end{array}
    \]

\emph{Case $n=n+1$:}
    \[
    \begin{array}{rclr}
    &&  \sem{\SEQ{\left( \SUMTOP{i=0}{n+1}
                            {\NFOLD[i]{(\SEQ{\testa}{\pola_1})}}\right)}
                 {\NOT{\testa}}} \\
    &=&
        \sem{\SEQ{\left(\ADD{\NFOLD[0]{(\SEQ{\testa}{\pola_1})}}
                            {\SUMTOP{i=0}{n}
                                {\NFOLD[i+1]{(\SEQ{\testa}{\pola_1})}}}\right)}
                 {\NOT{\testa}}}
    & \text{(By definition)} \\
    &=&
        \sem{\SEQ{\left(\ADD{\SKIP}
                            {\SUMTOP{i=0}{n}
                                {\SEQ{\SEQ{\testa}{\pola_1}}
                                    {\NFOLD[i]{(\SEQ{\testa}{\pola_1})}}}}\right)}
                 {\NOT{\testa}}}
    & \text{(Definition of $n$-th fold)} \\
    &=&
        \sem{\SEQ{\left(\ADD{\SKIP}
                            {\SEQ{\SEQ{\testa}{\pola_1}}
                                 {\SUMTOP{i=0}{n}
                                    \NFOLD[i]{(\SEQ{\testa}{\pola_1})}}}\right)}
                 {\NOT{\testa}}}
    & \text{(\Cref{thm:weightings-props}.\ref{thm:weightings-props3})} \\
    &=&
        \sem{\ADD{\SEQ{\SKIP}{\NOT{\testa}}}
                 {\SEQ{\SEQ{\SEQ{\testa}{\pola_1}}
                           {\left(\SUMTOP{i=0}{n}
                                    \NFOLD[i]{(\SEQ{\testa}{\pola_1})}\right)}}
                      {\NOT{\testa}}}}
    & \text{(\Cref{thm:weightings-props}.\ref{thm:weightings-props3})} \\
    &=&
        \sem{\ADD{\NOT{\testa}}
                 {\SEQ{\SEQ{\SEQ{\testa}{\pola_1}}
                           {\left(\SUMTOP{i=0}{n}
                                    \NFOLD[i]{(\SEQ{\testa}{\pola_1})}\right)}}
                      {\NOT{\testa}}}}
    & \text{(\Cref{thm:weightings-props}.\ref{thm:weightings-props2})} \\
    &=&
        \sem{\ADD{\NOT{\testa}}
                 {\SEQ{\SEQ{\testa}{\pola_1}}
                      {\NFOLD[n,\testa]{\pola_2}}}}
    & \text{(IH for $n$)} \\
    &=&
        \sem{\ADD{\NOT{\testa}}
                 {\SEQ{\SEQ{\testa}{\pola_2}}
                      {\NFOLD[n,\testa]{\pola_2}}}}
    & \text{(Assumption for $\pola_1$)} \\
    &=&
        \sem{\ADD{\SEQ{\NOT{\testa}}{\SKIP}}
                 {\SEQ{\SEQ{\testa}{\pola_2}}
                      {\NFOLD[n,\testa]{\pola_2}}}}
    & \text{(\Cref{thm:weightings-props}.\ref{thm:weightings-props2})} \\
    &=&
        \sem{\ADD{\SEQ{\SEQ{\testa}{\pola_2}}
                      {\NFOLD[n,\testa]{\pola_2}}}
                 {\SEQ{\NOT{\testa}}{\SKIP}}}
    & \text{(\Cref{thm:weightings-props}.\ref{thm:weightings-props1})} \\
    &=&
        \sem{\IF[\testa]{\SEQ{\pola_2}{\NFOLD[n,\testa]{\pola_2}}}
                        {\SKIP}}
    & \text{(By definition)} \\
    &=&
        \sem{\NFOLD[n+1,\testa]{\pola_2}}
    \end{array}
    \]
\end{proof}

\section{\wnetkat Subsumes \netkat and Guarded \probnetkat}
\label{appendix:subsumption}
It is natural to ask in what sense \wnetkat subsumes existing (extensions of) \netkat.
We establish tight connections between both \netkat and its probabilistic extension \probnetkat.
Whereas it is straightforward to see that \wnetkat conservatively extends \netkat, establishing a connection to \probnetkat is much more challenging as the semantics of \probnetkat requires an involved measure and domain-theoretic treatment.
Yet, \wnetkat admits an instance subsuming a rich and practically relevant \cite{mcnetkat} fragment of \probnetkat.

\smallskip \noindent\textbf{\wnetkat Subsumes \netkat.}
The semantics of classical~\netkat \cite{netkat} maps input histories to \emph{sets} of output histories. When instantiating \wnetkat with the Boolean semiring $\boolSemiring = {(\{0,1\},\, {\lor},\, {\land},\, 0,\, 1)}$, the semantic domains of classic \netkat and \iwnetkat{\boolSemiring} hence coincide. This yields the following:

\begin{theorem}[\wnetkat subsumes \netkat]
	Under the (syntactic) identifications $\AND{\testa}{\testb} = \testa \cdot \testb, \OR{\testa}{\testb} = \testa + \testb$, $\WEIGH{0}{\pola} = 0$, $\WEIGH{1}{\pola} = \pola$, $\SEQ{\pola}{\polb} = \pola \cdot \polb$, and $\ADD{\pola}{\polb} = \pola + \polb$, the syntax of \iwnetkat{\boolSemiring} reduces to classical \netkat and
	\(
		\sem{\pola}(\hista)(\hista') =1
		\quad\text{iff}\quad
		\text{classical \netkat produces $\hista'$ when executing $\pola$ on $\hista$}~.
	\)
\end{theorem}

\smallskip \noindent\textbf{\wnetkat Subsumes Guarded \probnetkat.}
The semantics of \emph{full} \probnetkat maps \emph{sets} of input histories to
\emph{possibly continuous} probability measures over \emph{sets} of output histories,
requiring a complex measure- and domain-theoretic treatment \cite{cantor}. Loosely speaking,
this is due to the fact that \probnetkat mixes nondeterministic and probabilistic behavior,
which is well-known to be challenging. Even though \wnetkat does not subsume \emph{full}
\probnetkat, it subsumes its rich \emph{guarded} fragment~\cite{mcnetkat}. In what follows, we first link
the two extensions of \netkat syntactically. We then treat their semantic relationship.

Instantiate \wnetkat with the non-negative extended reals $(\posrealsinf, \, +, \, \cdot, \, 0, \, 1)$ and restrict the usage of weightings to choices of the form $\ADD{\WEIGH{\proba}{\pola}}{\WEIGH{(1 - \proba)}{\polb}}$, which coincides with the probabilistic choice operator $\PROBCHOICE[\proba]{\pola}{\polb}$ of \probnetkat, where $0 \leq \proba \leq 1$. Furthermore, we restrict choice and iteration to occur guarded, in the forms $\ADD{\SEQ{\testa}{\pola}}{\SEQ{\NOT{\testa}}{\polb}}$ and $\SEQ{\ITER{(\SEQ{\testa}{\pola})}}{\NOT{\testa}}$. Finally, for $\testa,\testb$, we identify \probnetkat's disjunction $\testa \& \testb$ with \wnetkat's $\OR{\testa}{\testb}$. Call the
the resulting instance \iwnetkat{$\probSemiring$}. With these syntactic identifications, the syntax of \iwnetkat{$\probSemiring$} and \emph{guarded} \probnetkat coincide.

Let us now compare the semantics of  \iwnetkat{$\probSemiring$} and guarded \probnetkat. First, \iwnetkat{$\probSemiring$} maps input policies to discrete probability \emph{sub}distributions over histories:
\[
	\forall \hista,\hista'\colon\sem{\pola}(\hista)(\hista') \in [0,1]
	\qquad\text{and}\qquad
	\underbrace{\mass{\sem{\pola}(\hista)} = \semisum{\hista' \in \supp(\sem{\pola}(\hista))} \wtinga(\hista')}_{\text{mass of $\sem{\pola}(\hista)$}} \leq 1~.
\]
The semantics of \probnetkat, denoted by $\altsem{-}$, on the other hand, maps \emph{sets} of input histories to possibly \emph{continuous} distributions over \emph{sets} of histories\footnote{See \cite{probnetkat} for the construction of the corresponding probability space.}, i.e., formally $\altsem{\pola} \colon \setsof{\histories} \rightarrow \probmonad(\setsof{\histories})$, where $\probmonad(\setsof{\histories})$ is the (continuous) Giry monad \cite{giry}. Even though $\sem{-}$ and $\altsem{-}$ are quite different in nature and \iwnetkat{$\probSemiring$} does not require involved measure theory,
we establish a tight connection by carefully relating the \emph{approximants} of both semantics. First, the mass of $\sem{\pola}(\hista)$ produced by \iwnetkat{$\probSemiring$} coincides with the probability of \probnetkat not producing the empty set of output histories. Second, the probability \iwnetkat{$\probSemiring$} assigns to $\hista'$ on input $\hista$ coincides with the probability of \probnetkat producing an output set containing $\hista'$.
Finally, guarded \probnetkat policies assign probability mass only to subsingleton
sets of histories (i.e., guarded \probnetkat policies have subsingleton support).

\begin{theorem}[\wnetkat subsumes guarded \probnetkat]
	Given a policy $\pola$ in the guarded fragment of \probnetkat and a history
	$\hista \in \histories$, the semantics of \iwnetkat{$\probSemiring$} subsume
    the semantics of guarded \probnetkat in the following sense:
    \begin{enumerate}
    \item
		$\int_{\histsa \in \setsof{\histories}} [1 < |\histsa|] \ d \altsem{\pola}(\{\hista\}) = 0$,
    \item
		$\mass{\sem{\transpol{\pola}}(\hista)} =
            \int_{\histsa \in \setsof{\histories}} \
                [\histsa \neq \varnothing] \ d \altsem{\pola}(\{\hista\})$, and
    \item
		$\forall \histb \in \supp(\sem{\transpol{\pola}}(\hista)) \colon
		    \sem{\transpol{\pola}}(\hista)(\histb) =
		        \int_{\histsa \in \setsof{\histories}} \
		            [\histb \in \histsa] \ d \altsem{\pola}(\{\hista\})$.
    \end{enumerate}
\end{theorem}
\begin{proof}
    We prove only (3), the reasoning for the other two claims is analogous.
    Our key idea is to exploit \cite[Corollary 23]{cantor}: For every Scott-continuous\footnote{w.r.t.\ the DCPO $(\setsof{\histories},\, \subseteq)$ of sets of histories} random variable $f \colon \setsof{\histories} \to \posrealsinf$, we have
	\begin{align*}
        \label{eqn:thm_cantor_approx}
		\int_{\histsa \in \setsof{\histories}} \
		f \ d \altsem{\pola}(\{\hista\})
		\eeq
		\sup_{n\in\nats}
		\int_{\histsa \in \setsof{\histories}} \
		f \ d
		\approxn[n]{\altsem{\pola}}(\hista)~,
		\tag{$\dagger$}
	\end{align*}
	where $\approxn[n]{\altsem{\pola}}$ denotes \probnetkat's $n$-th approximant.
    This enables us to reduce the reasoning to \emph{iteration-free} policies, and
    show
    \begin{align*}
		\sup_{n\in\nats}
		\int_{\histsa \in \setsof{\histories}} \
		[\histb \in \histsa] \ d
		\approxn[n]{\altsem{\pola}}(\hista)
        \eeq
		\sup_{n\in\nats} \
        \approxn[n]{\altsem{\pola}}(\hista)(\histb)~.
        \tag{\Cref{lmm:singleton-support-approx}}
    \end{align*}

    By ensuring that approximants of while-loops for \probnetkat stay within the guarded fragment
    (\Cref{lmm:probnetkat-alt-approx-pol-equiv}), we prove
	\begin{align*}
		\forall \hista\in\histories\colon\forall n \in \nats \colon
		\quad
		 \underbrace{\approxn[n]{\sem{\pola}}(\hista)}_{\text{in \iwnetkat{$\probSemiring$}}}
         \eeq
         \underbrace{\approxn[n]{\altsem{\pola}}(\{\hista\})}_{\text{in}~\probnetkat}~,
         \tag{\Cref{lmm:approx-equiv-wnketkat-probnetkat}}
	\end{align*}
    where $\approxn[n]{\sem{\pola}}$ denotes \iwnetkat{$\probSemiring$}'s $n$-th approximant. We
    additionally show that the approximants of a policy $\pola$ in \iwnetkat{$\probSemiring$}
    yield precisely the semantics of $\pola$ in the limit, i.e.,
	\begin{align*}
        \sup_{n\in\nats} \approxn[n]{\sem{\pola}}
        =
        \sem{\pola}~.
		\tag{\Cref{thm:sup-approx-equiv}}
	\end{align*}
	Finally, this implies the claim because
	\[
		\sem{{\pola}}(\hista)(\histb)
        ~\substack{\text{Thm.\ref{thm:sup-approx-equiv}} \\ {}={}}~
		 \sup_{n \in \nats} \ \approxn[n]{\sem{\pola}}(\hista)(\histb)
		 ~\substack{\text{Lmm.\ref{lmm:approx-equiv-wnketkat-probnetkat}} \\ {}={}}~
         \sup_{n \in \nats} \ \approxn[n]{\altsem{\pola}}(\{\hista\})(\{\{\histb\}\})
         ~\substack{\text{Lmm.\ref{lmm:singleton-support-approx},\ref{eqn:thm_cantor_approx}} \\ {}={}}~
           \int_{\histsa \in \setsof{\histories}} \
           [\histb \in \histsa] \ d
           \altsem{\pola}(\hista)~.
	\]
\end{proof}

\subsection{Equivalence Up to Approximation of \iwnetkat{$\probSemiring$} and Guarded \probnetkat}

In this section we derive the equivalence of \iwnetkat{$\probSemiring$} and Guarded \probnetkat
up to \emph{approximation}. We define approximation for \wnetkat in \Cref{appendix:approximation}),
and use the following definition for \probnetkat.

\begin{definition}[Semantics and approximation of \probnetkat programs]
	We take the following definitions directly from \citet{cantor}:
	\begin{enumerate}
		\item Let
		$\altsem{-} \in \setsof{\histories} \rightarrow \probmonad(\setsof{\histories})$
		denote the semantic map for \probnetkat, where $\probmonad$ is the
		(continuous) probability monad
		\item Let $\altapproxpol{\pola}{n}$ denote the \emph{$n$-th approximant} of a
		policy $\pola$ in \probnetkat
		\item Let $\approxn[n]{\altsem{\pola}}(\histsa)$ denote the (discrete)
		measure obtained from the $n$-th approximant:
		$\altsem{\altapproxpol{\pola}{n}}(\histsa)$.
		\item Let $\ALTNFOLD[n]{\pola}$ denote the \emph{$n$-th fold} of a \probnetkat policy
		$\pola$:
		\[
		\begin{array}{rcl}
			\ALTNFOLD[0]{\pola} &\triangleq& \SKIP
		\end{array}
		\quad
		\begin{array}{rcl}
			\ALTNFOLD[n+1]{\pola} &\triangleq& \PAR{\SKIP}{\SEQ{\pola}{\ALTNFOLD[n]{\pola}}}
		\end{array}
		\]
	\end{enumerate}
\end{definition}

We deviate from \citet{cantor} slightly in the definition of the $n$-th
approximant for guarded iteration (which would just be the $n$-th approximant
of its desugared policy in full \probnetkat). We use the following
alternate definition throughout:
\[
\begin{array}{rcl}
	\altapproxpol{\WHILE[\testa]{\pola}}{n} &\triangleq&
	\ALTNFOLD[n,\testa]{(\altapproxpol{\pola}{n})}
\end{array}
\quad \text{where} \quad
\begin{array}{rcl}
	\ALTNFOLD[0,\testa]{\pola} &\triangleq& \NOT{\testa} \\
	\ALTNFOLD[n+1,\testa]{\pola} &\triangleq&
	\IF[\testa]{\SEQ{\pola}{\ALTNFOLD[n,\testa]{\pola}}}{\SKIP} \\
\end{array}
\]

This alternate definition ensures that the $n$-th approximant of a policy
in the guarded fragment of \probnetkat is also in the guarded fragment of
\probnetkat (and the remainder of our proofs can proceed by induction on
approximants of guarded \probnetkat policies). Despite the deviation, we
are able to recover that this definition is equivalent
(\Cref{lmm:probnetkat-alt-approx-pol-equiv}).
With this equivalent definition of approximation in place, we show that the
guarded fragment of \probnetkat and \iwnetkat{$\probSemiring$} are equivalent
up to approximation. We define precisely what we mean by equivalent as follows.

\begin{definition}
	\label{def:weightings-measures-equiv}
	Given a weighting $\wtinga \in \weightings[\nonnegRealsSemiring]{\histories}$ and a discrete
	measure $\measurea \in \measures{\setsof{\histories}}$, we say the two are
	\emph{equivalent} and write $\wtinga \sim \measurea$ if the following three conditions hold: (1) for all \(\hista \in \supp(\wtinga)\), \(\wtinga(\hista) = \measurea(\{\{\hista\}\})\), (2) \(\sum{\{\histsa \in \setsof{\histories} \ |\ 1 < |\histsa| \}}
		\measurea(\{\histsa\}) = 0\), and (3) \(\mass{\wtinga} =
		\sum{\hista \in \histories}
		\measurea(\{\{\hista\}\})\).
		\end{definition}

\begin{lemma}[Equivalence up to approximation of \iwnetkat{$\probSemiring$} and \probnetkat]
\label{lmm:approx-equiv-wnketkat-probnetkat}
	For all $n \in \nats$: given a policy $\pola$ in the guarded fragment of \probnetkat,
	the semantics of \iwnetkat{$\probSemiring$} and the guarded fragment of \probnetkat
	are equivalent up to $n$-th approximation:
	\[
	\begin{array}{lrcl}
		\forall \hista \in \histories \colon &
		\approxn[n]{\sem{ \transpol{\pola} }}(\hista) &\sim&
		\approxn[n]{\altsem{\pola}}(\{\hista\})
	\end{array}
	\]
\end{lemma}
\begin{proof}
    By \Cref{lmm:singleton-support-approx}, we have that
    $\sum{\{\histsa \in \setsof{\histories} \ |\ 1 < |\histsa| \}}
        \approxn[n]{\altsem{\pola}}(\{\hista\})(\{\histsa\}) = 0$.
    We prove the remaining two conditions of \Cref{def:weightings-measures-equiv} by structural induction on
    the policy $\pola$.

    \emph{Case $\pola$ predicate:}\
        By \Cref{lmm:predicate-equiv-wnetkat-gprobnetkat}.

    \emph{Case $\pola = \ASSN{\fielda}{\vala}$:}\
        By \Cref{lmm:equiv-eta-dirac}.

    \emph{Case $\pola = \DUP$:}\
        By \Cref{lmm:equiv-eta-dirac}.

    \emph{Case $\pola = \SEQ{\pola}{\polb}$:}\
        By \Cref{lmm:bind-equiv-wnetkat-gprobnetkat} and IH.

    \emph{Case $\pola = \IF[\testa]{\pola}{\polb}$:}\
        By \Cref{lmm:if-then-else-equiv-wnetkat-gprobnetkat} and IH.

    \emph{Case $\pola = \WHILE[\testa]{\pola}$:}

        We first unfold the LHS/RHS.
        \[
        \begin{array}{rclr}
            \approxn[n]{\sem{\WHILE[\transpol{\testa}]{\transpol{\pola}}}}
        &=& \sem{\approxpol{\WHILE[\transpol{\testa}]{\transpol{\pola}}}{n}}
        & \text{(By definition)} \\
        &=& \sem{\NFOLD[n,\transpol{\testa}]{(\approxpol{\transpol{\pola}}{n})}}
        & \text{(\Cref{lmm:alt-approx-equiv})} \\\\

        \approxn[n]{\altsem{\WHILE[\testa]{\pola}}}
        &=& \altsem{\altapproxpol{\WHILE[\testa]{\pola}}{n}}
        & \text{(By definition)} \\
        &=& \altsem{\ALTNFOLD[n,\testa]{(\altapproxpol{\pola}{n})}}
        & \text{(\Cref{lmm:probnetkat-alt-approx-pol-equiv})}
        \end{array}
        \]

        Finally, we have:
        \[
        \begin{array}{rclrc}
            \sem{\NFOLD[n,\transpol{\testa}]{(\approxpol{\transpol{\pola}}{n})}}
        &=&
            \altsem{\ALTNFOLD[n,\testa]{(\altapproxpol{\pola}{n})}}
        & \text{(\Cref{lmm:guarded-iteration-equiv-wnetkat-gprobnetkat} and IH)}
        \end{array}
        \]

    \emph{Case $\pola = \PROBCHOICE[\proba]{\pola}{\polb}$:}

        We first unfold the LHS/RHS.
        \[
        \begin{array}{rclr}
            \approxn[n]{\sem{\transpol{\PROBCHOICE[\proba]{\pola}{\polb}}}}
        &=&
            \approxn[n]{\sem{\ADD{\WEIGH{\proba}{\transpol{\pola}}}
                                 {\WEIGH{(1 - \proba)}{\transpol{\polb}}}}}
        & \text{(By definition)} \\
        &=&
            \proba \cdot \approxn[n]{\sem{\transpol{\pola}}} +
                (1 - \proba) \cdot \approxn[n]{\sem{\transpol{\polb}}}
        & \text{(By definition)} \\\\

            \approxn[n]{\altsem{\PROBCHOICE[\proba]{\pola}{\polb}}}
        &=&
            \proba \cdot \approxn[n]{\altsem{\pola}} +
                (1 - \proba) \cdot \approxn[n]{\altsem{\polb}}
        & \text{(By definition)}
        \end{array}
        \]

        By \Cref{lmm:scalar-mul-equiv-wnetkat-gprobnetkat}, we have:
        \[
        \begin{array}{rclr}
            \proba \cdot \approxn[n]{\sem{\transpol{\pola}}} &\sim&
                \proba \cdot \approxn[n]{\altsem{\pola}}~,
            & \text{(with IH for $\pola$)} \\
            &\text{and}& \\
                (1 - \proba) \cdot \approxn[n]{\sem{\transpol{\polb}}} &\sim&
                    (1 - \proba) \cdot \approxn[n]{\altsem{\polb}}~.
            & \text{(with IH for $\polb$)}
        \end{array}
        \]

        Finally, we conclude:
        \[
        \begin{array}{rclr}
            \proba \cdot \approxn[n]{\sem{\transpol{\pola}}} +
                (1 - \proba) \cdot \approxn[n]{\sem{\transpol{\polb}}}
        &\sim&
            \proba \cdot \approxn[n]{\altsem{\pola}} +
                (1 - \proba) \cdot \approxn[n]{\altsem{\polb}}
        &\text{(\Cref{lmm:addition-equiv-wnetkat-gprobnetkat})}
        \end{array}
        \]
\end{proof}

\subsection{Lemmas for Relation Between \iwnetkat{$\probSemiring$} and \probnetkat}

\begin{lemma}[Equivalence of alternate definition of guarded iteration
    for \probnetkat]
\label{lmm:alt-defn-guarded-iteration-equiv-probnetkat}
Given policies $\pola_1$, $\pola_2$ and predicate $\testa$ in \probnetkat, we
have:
\[
\begin{array}{lrcl}
    \forall n \in \nats \colon &
        \altsem{\pola_1} = \altsem{\pola_2} &\Rightarrow&
            \altsem{\SEQ{\ALTNFOLD[n]{(\SEQ{\testa}{\pola_1})}}{\NOT{\testa}}}
                =
                    \altsem{\ALTNFOLD[n,t]{\pola_2}}

\end{array}
\]
\end{lemma}
\begin{proof}
    By induction on $n$.

    \emph{Case $n = 0$:}
        $$
            \altsem{\SEQ{\ALTNFOLD[0]{(\SEQ{\testa}{\pola_1})}}
                        {\NOT{\testa}}}
        =
            \altsem{\SEQ{\SKIP}{\NOT{\testa}}}
        = \altsem{\NOT{\testa}}
        = \altsem{\ALTNFOLD[0,\testa]{(\altapproxpol{\pola}{0})}}
        $$

    \emph{Case $n = n + 1$:}
        \[
        \begin{array}{rclr}
        &&
            \altsem{\SEQ{\ALTNFOLD[n+1]{(\SEQ{\testa}{\pola_1})}}
                        {\NOT{\testa}}} \\
        &=&
            \altsem{\SEQ{(\PAR{\SKIP}
                              {\SEQ{\SEQ{\testa}{\pola_1}}
                                   {\ALTNFOLD[n]{(\SEQ{\testa}{\pola_1})}}})}
                        {\NOT{\testa}}}
        & \text{(By definition)} \\
        &=&
            \altsem{\PAR{\SEQ{\SKIP}{\NOT{\testa}}}
                        {\SEQ{{\SEQ{\SEQ{\testa}{\pola_1}}
                                   {\ALTNFOLD[n]{(\SEQ{\testa}{\pola_1})}}}}
                             {\NOT{\testa}}}}
        & \text{(NetKAT axioms \cite[Corollary 1]{probnetkat})} \\
        &=&
            \altsem{\PAR{\NOT{\testa}}
                        {\SEQ{{\SEQ{\SEQ{\testa}{\pola_1}}
                                   {\ALTNFOLD[n]{(\SEQ{\testa}{\pola_1})}}}}
                             {\NOT{\testa}}}}
        & \text{(\cite[Lemma 4]{probnetkat})} \\
        &=&
            \altsem{\PAR{\NOT{\testa}}
                        {\SEQ{\SEQ{\testa}{\pola_1}}}
                             {\ALTNFOLD[n,\testa]{\pola_2}}}
            & \text{(IH for $n$)} \\
        &=&
            \altsem{\PAR{\NOT{\testa}}
                        {\SEQ{\SEQ{\testa}{\pola_2}}}
                             {\ALTNFOLD[n,\testa]{\pola_2}}}
            & \text{(Assumption for $\pola_1$)} \\
        &=&
            \altsem{\PAR{\SEQ{\SEQ{\testa}{\pola_2}}
                             {\ALTNFOLD[n,\testa]{\pola_2}}}
                        {\NOT{\testa}}}
        & \text{(NetKAT axioms \cite[Lemma 4]{probnetkat})} \\
        &=&
            \altsem{\PAR{\SEQ{\SEQ{\testa}{\pola_2}}
                             {\ALTNFOLD[n,\testa]{\pola_2}}}
                        {\SEQ{\NOT{\testa}}{\SKIP}}}
        & \text{(NetKAT axioms \cite[Lemma 4]{probnetkat})} \\
        &=& \altsem{\IF[\testa]{\SEQ{\pola_2}
                                    {\ALTNFOLD[n,\testa]{\pola_2}}}
                               {\SKIP}}
        & \text{(By definition)} \\
        &=&
            \altsem{\ALTNFOLD[n+1,\testa]{\pola_2}}
        & \text{(By definition)}
        \end{array}
        \]
\end{proof}

\begin{lemma}
[Equivalence of alternate definition of approximation in \probnetkat]
\label{lmm:probnetkat-alt-approx-pol-equiv}
	Let $\altaltapproxpol{\pola}{n}$ denote the $n$-th approximant
	for a policy $\pola$ as originally defined in \citet{cantor}. The alternate
	definition for approximation is equivalent in the following sense:
	\[
	\begin{array}{lrcl}
		\forall n \in \nats \colon &
		\altsem{\altaltapproxpol{\pola}{n}} &=&
		\approxn[n]{\altsem{\pola}}
	\end{array}
    \]
\end{lemma}
\begin{proof}
    By structural induction on $\pola$. All cases follow immediately except for
    guarded iteration, as the definitions coincide exactly. We consider only
    this case.

    \emph{Case $\pola = \WHILE[\testa]{\pola}$:}

    We unfold the LHS/RHS first:
    \[
    \begin{array}{rclr}
        \altsem{\altaltapproxpol{\WHILE[\testa]{\pola}}{n}}
    &=&
        \altsem{\altaltapproxpol{\SEQ{\ITER{(\SEQ{\testa}{\pola})}}{\NOT{\testa}}}{n}} \\
    &=&
        \altsem{\SEQ{\ALTNFOLD[n]{(\altaltapproxpol{\SEQ{\testa}{\pola}}{n})}}{\NOT{\testa}}}
    & \text{(By definition \cite{cantor})} \\
    &=&
        \altsem{\SEQ{\ALTNFOLD[n]{(\SEQ{\testa}{\altaltapproxpol{\pola}{n}})}}{\NOT{\testa}}}
     & \text{(By definition \cite{cantor})} \\\\

    \approxn[n]{\altsem{\WHILE[\testa]{\pola}}}
    &=&
        \altsem{\altapproxpol{\WHILE[\testa]{\pola}}{n}} \medskip \\
    &=&
        \altsem{\ALTNFOLD[n,\testa]{(\altapproxpol{\pola}{n})}}
    & \text{(Alternate definition)}
    \end{array}
    \]

    Finally, we conclude:
    \[
    \begin{array}{rclr}
        \altsem{\SEQ{\ALTNFOLD[n]{(\SEQ{\testa}{\altaltapproxpol{\pola}{n}})}}{\NOT{\testa}}}
    &=&
        \altsem{\ALTNFOLD[n,\testa]{(\altapproxpol{\pola}{n})}}
    & \text{(\Cref{lmm:alt-defn-guarded-iteration-equiv-probnetkat} and IH)}
    \end{array}
    \]
\end{proof}

\begin{definition}[Subsingleton support of a discrete measure]
\label{def:subsingleton-support}
We say a discrete measure $\measurea \in \measures{\setsof{\histories}}$ has a
\emph{subsingleton support} if \(\quad
    \forall \{\histsa\} \in \supp(\measurea) \colon  |\histsa| \leq 1\).
\end{definition}

\begin{definition}[Subsingleton support for Markov kernels]
    We say that a  (discrete) Markov kernel  $\markova \colon \setsof{\histories} \times \borelsets \rightarrow [0,1]$ has a
\emph{subsingleton support} if \(\quad        \forall \histsa \ \text{s.t.} \ |\histsa| \leq 1 \colon
            \markova(\histsa, -) \ \text{has subsingleton support}\).
\end{definition}

\begin{lemma}[Bind preserves subsingleton support]
\label{lmm:bind-preserves-subsingleton}
    For any discrete measure $\measurea \in \measures{\setsof{\histories}}$ and
    Markov kernel $\markova \colon \setsof{\histories} \times \borelsets \rightarrow [0,1]$
    with subsingleton support, we have that $\measurea \bind \markova$ has
    subsingleton support.
\end{lemma}
\begin{proof}

            \[
            \begin{array}{rclr}
                &&
                    \measurea \bind \markova \\
                &=&
                    \lambda \eventa.
                        \int_{\histsa \in \setsof{\histories}}
                            \markova(\histsa,A) \cdot
                                \measurea (d \histsa)
                & \text{(By definition \cite[Figure 2]{cantor})} \\\\
                &=&
                    \lambda \eventa.
                        \begin{cases}
                            \sum{\histsa \in \setsof{\histories}}
                                \markova(\histsa,\{\histsb\}) \cdot
                                    \measurea(\{\histsa\})
                            & \eventa = \{ \histsb \} \\
                            0 & \text{otherwise}
                        \end{cases}
                & \text{($\measurea$ is discrete)} \\
            \end{array}
            \]

            Consider any $\histsa \in \setsof{\histories}$. By assumption
            for $\measurea$, we have that $|\histsa| \leq 1$. Therefore, for any
            $\histsb \in \setsof{\histories}$ such
            that $\{\histsb\} \in \supp(\measurea \bind \markova)$,
            we have by assumption for $\markova$ that $|\histsb| \leq 1$.
\end{proof}

\begin{lemma}[Subsingleton support of conditional choice]
\label{lmm:subsingleton-conditional-choice}
    For any *-free policies $\pola$, $\polb$ and predicate $\testa$ in \probnetkat: if \(
        \altsem{\pola} \) and  \( \altsem{\polb}\) have subsingleton support then
        \(        \altsem{\IF[\testa]{\pola}{\polb}} \ \text{has subsingleton support}.
   \)
\end{lemma}
\begin{proof}
    We proceed by case analysis
    on $\altsem{\testa}(\histsa)$. There are two cases (\cite[Lemma 2]{cantor}), from both of
    which we determine that either
    $\approxn[n]{\altsem{\IF[\testa]{\pola}{\polb}}}(\histsa) =
        \approxn[n]{\altsem{\pola}}(\histsa)$ or
    $\approxn[n]{\altsem{\IF[\testa]{\pola}{\polb}}}(\histsa) =
        \approxn[n]{\altsem{\polb}}(\histsa)$.
    In either case, the property holds by assumption.
\end{proof}

\begin{lemma}[Subsingleton support of guarded iteration]
\label{lmm:subsingleton-guarded-iteration}
    For any *-free policy $\pola$ and predicate $\testa$ in \probnetkat:
    \[
    \begin{array}{lrcl}
        \forall n \in \nats \colon &
                \altsem{\pola} \ \text{has subsingleton support}
                    &\Rightarrow&
        \altsem{\ALTNFOLD[n,\testa]{\pola}} \
                \text{has subsingleton support}
    \end{array}
    \]
\end{lemma}
\begin{proof}
    By induction on $n$.

    \emph{Case $n=0$:}
        \[
        \begin{array}{rclr}
            \altsem{\ALTNFOLD[0,\testa]{\pola}}
        &=& \altsem{\NOT{\testa}}
        & \text{(By definition)}
        \end{array}
        \]

        The property holds for predicates as they behave like packet filters
        \citep[Lemma 2]{cantor}.

    \emph{Case $n=n+1$:}
        \[
        \begin{array}{rclr}
            \altsem{\ALTNFOLD[n+1,\testa]{\pola}}
        &=&
            \altsem{\IF[\testa]{\SEQ{\pola}
                                    {\ALTNFOLD[n,\testa]{\pola}}}
                               {\SKIP}}
        &\text{(By definition)}
        \end{array}
        \]

        By \Cref{lmm:subsingleton-conditional-choice}, we have to show
        that both $\altsem{\SEQ{\pola}{\ALTNFOLD[n,\testa]{\pola}}}$
        and $\altsem{\SKIP}$ have subsingleton support. Any atomic
        program only acts on individual histories \citep[Lemma 3]{cantor},
        so we have that $\altsem{\SKIP}$ has subsingleton support. We
        proceed to show that
        $\altsem{\SEQ{\pola}{\ALTNFOLD[n,\testa]{\pola}}}$
        has subsingleton support.
        We have:
        \[
        \begin{array}{rclr}
            \altsem{\SEQ{\pola}{\ALTNFOLD[n,\testa]{\pola}}}(\histsa)
        &=& \altsem{\pola}(\histsa) \bind \altsem{\ALTNFOLD[n,\testa]{\pola}}
        & \text{(By definition \cite[Figure 2]{cantor})}
        \end{array}
        \]

        By \Cref{lmm:bind-preserves-subsingleton}, we need to show that
        $\altsem{\pola}(\histsa)$ and $\altsem{\ALTNFOLD[n,\testa]{\pola}}$
        have subsingleton support. However, these conditions hold by
        assumption for $\pola$ and by the induction hypothesis for $n$.
\end{proof}

\begin{lemma}[Subsingleton support of approximants of guarded \probnetkat policies]
\label{lmm:singleton-support-approx}
    For any policy $\pola$ in the guarded fragment of \probnetkat,
    $\approxn[n]{\altsem{\pola}}$ has subsingleton support.
\end{lemma}
\begin{proof}
    By structural induction on the policy $p$. For predicates $\testa$ this follows from
    them behaving like packet filters \citep[Lemma 2]{cantor}. Likewise, the property holds
    for atomic programs as they only act on individual histories \citep[Lemma 3]{cantor}.
    We consider only the remaining cases:

        \emph{Case $\pola = \SEQ{\pola}{\polb}$:}
            \[
            \begin{array}{rclr}
                \approxn[n]{\altsem{\SEQ{\pola}{\polb}}}(\histsa)
                &=&
                    \approxn[n]{\altsem{\pola}}(\histsa) \bind
                        \approxn[n]{\altsem{\polb}}
                & \text{(By definition \cite[Figure 2]{cantor})}
            \end{array}
            \]

            Therefore, the property holds by \Cref{lmm:bind-preserves-subsingleton}
            and IH.

        \emph{Case $\pola = \IF[\testa]{\pola}{\polb}$:} \
            By \Cref{lmm:subsingleton-conditional-choice} and IH. \\

        \emph{Case $\pola = \WHILE[\testa]{\pola}$:}
            \[
            \begin{array}{rclr}
                 \approxn[n]{\altsem{\WHILE[\testa]{\pola}}}(\histsa)
                &=&
                  \altsem{\altapproxpol{\WHILE[\testa]{\pola}}{n}}(\histsa)
                & \text{(By definition \cite{cantor})} \\
                &=&
                  \altsem{\ALTNFOLD[n,t]{(\altapproxpol{\pola}{n})}}(\histsa)
                & \text{(\Cref{lmm:probnetkat-alt-approx-pol-equiv})}
            \end{array}
            \]

            Therefore, the property holds by \Cref{lmm:subsingleton-guarded-iteration}
            and IH.

        \emph{Case $\pola = \PROBCHOICE[\proba]{\pola}{\polb}$:}

            \[
            \begin{array}{rclr}
                \approxn[n]{\altsem{\PROBCHOICE[\proba]{\pola}{\polb}}}(\histsa)
            &=&
                \altsem{\PROBCHOICE[\proba]{\altapproxpol{\pola}{n}}
                                       {\altapproxpol{\polb}{n}}}(\histsa)
            & \text{(By definition)} \\
            &=&
                \proba \cdot \altsem{{\altapproxpol{\pola}{n}}}(\histsa) +
                    (1 - \proba) \cdot \altsem{{\altapproxpol{\polb}{n}}}(\histsa)
            & \text{(By definition)} \\
            &=&
                \proba \cdot \approxn[n]{\altsem{\pola}}(\histsa) +
                    (1 - \proba) \cdot \approxn[n]{\altsem{\polb}}(\histsa)
            & \text{(By definition)}
            \end{array}
            \]

            Therefore, we have that
            $\supp(\approxn[n]{\altsem{\PROBCHOICE[\proba]{\pola}{\polb}}}(\histsa)) \subseteq
                \supp(\approxn[n]{\altsem{\pola}}(\histsa)) \cup
                    \supp(\approxn[n]{\altsem{\polb}}(\histsa))$
            and the property holds by IH.
\end{proof}
    Note that approximants are *-free programs so their probability distributions
    are discrete.

\begin{lemma}[Equivalence of unit across \iwnetkat{$\probSemiring$} and
guarded \probnetkat]
\label{lmm:equiv-eta-dirac}
\[
\begin{array}{lrcl}
    \forall \hista \in \histories \colon &
        \unit(\hista) \sim \delta_{\{\hista\}}
\end{array}
\]
\end{lemma}
\begin{proof}
We prove the three conditions of \Cref{def:weightings-measures-equiv} separately:
\begin{itemize}
    \item {$\sum{\{\histsa \in \setsof{\histories} \ |\ 1 < |\histsa|\}} \delta_{\{\hista\}}(\histsa) = 0$} holds
        by definition. \medskip \\
    \item {$\mass{\unit(\hista)} = \sum{\hista \in \histories} \delta_{\{\hista\}}(\{\{\hista\}\})$}
        $$
        \mass{\unit(\hista)} =
        1 =
        \delta_{\{\hista\}}(\{\{\hista\}\}) =
        \sum{\hista \in \histories} \delta_{\{\hista\}}(\{\{\hista\}\})
        $$
    \item{$\forall \histb \in \supp(\wtinga) \colon \unit(\hista)(\histb) = \delta_{\{\hista\}}(\{\{\histb\}\})$} \\
        We proceed by case analysis:
        \begin{itemize}
            \item{$\hista = \histb$:}
                $$
                \unit(\hista)(\hista)
                = 1
                = \delta_{\{\hista\}}(\{\{\hista\}\})
                $$
            \item{$\hista \neq \histb$:}
                $$
                \unit(\hista)(\histb)
                = 0
                = \delta_{\{\hista\}}(\{\{\histb\}\})
                $$
        \end{itemize}\qedhere
\end{itemize}
\end{proof}

\begin{definition}
\label{def:equiv-fn-markov}
We lift equivalence between weightings and discrete measures to functions
and Markov kernels as follows.
\[
\begin{array}{llclrcl}
\forall f \colon \histories \rightarrow \weightings[\nonnegRealsSemiring]{\histories},
    \markova \colon \setsof{\histories} \times \borelsets \rightarrow [0,1] \colon &
        f \sim \markova &\longeq&
            \forall \hista \colon &
                f(\hista) \sim \markova(\{\hista\},-)
\end{array}
\]
\end{definition}

\begin{lemma}[Bind respects equivalence across \iwnetkat{$\probSemiring$}
    and guarded \probnetkat]
\label{lmm:bind-equiv-wnetkat-gprobnetkat}
Given a weighting $\wtinga \in \weightings[\nonnegRealsSemiring]{\histories}$, a discrete
measure $\measurea \in \measures{\setsof{\histories}}$, a function
$f \colon \histories \rightarrow \weightings[\nonnegRealsSemiring]{\histories}$, and
a Markov kernel $\markova \colon \setsof{\histories} \times \borelsets \rightarrow [0,1]$,
we have:
\[
\begin{array}{lrcl}
    \wtinga \sim \measurea \wedge f \sim \markova &\Rightarrow&
        \wtinga \bind f \sim \measurea \bind \markova
\end{array}
\]
\end{lemma}
\begin{proof}
We begin by unfolding the RHS:
\[
\begin{array}{rclr}
&& \measurea \bind \markova \\
&=&
    \lambda \eventa.\, \int_{\histsa \in \setsof{\histories}} \markova(\histsa,A) \cdot \measurea(d\histsa)
& \text{(By definition \cite{cantor})} \medskip \\
&=&
    \lambda \eventa.
        \begin{cases}
            \sum{\histsa \in \setsof{\histories}}
               \markova(\histsa,\{\histsb\}) \cdot \measurea(\{\histsa\})
               & \eventa = \{\histsb\} \\
            0 & \text{otherwise}
        \end{cases}
& \text{($\measurea$ is discrete)} \medskip \\
&=&
    \lambda \eventa.
        \begin{cases}
            \markova(\varnothing,\{\histsb\}) \cdot \measurea(\{\varnothing\}) +
                \sum{\hista \in \histories}
                    \markova(\{\hista\},\{\histsb\}) \cdot \measurea(\{\{\hista\}\})
            & \eventa = \{\histsb\} \\
        0 & \text{otherwise}
        \end{cases}
& \text{($\wtinga \sim \measurea$)}
\end{array}
\]
We now prove all three conditions of \Cref{def:weightings-measures-equiv} separately.
\begin{itemize}
    \item{$\sum{\{\histsa \in \setsof{\histories} \ |\ 1 < |\histsa|\}}
            (\measurea \bind \markova)(\{\histsa\}) = 0$} follows immediately from the property holding for $\markova$. \\

    \item{$\mass{\wtinga \bind f} =
        \sum{\hista \in \histories}
            (\measurea \bind \markova)(\{\{\hista\}\})$}
        \[
        \begin{array}{rclr}
        && \mass{\wtinga \bind f} \\
        &=& \sum{\hista \in \supp(\wtinga \bind f)} (\wtinga \bind f)(\hista)
        & \text{(By definition)} \medskip \\
        &=&
            \sum{\hista \in \supp(\wtinga \bind f)}
                \sum{\histb \in \supp(\wtinga)}
                    \wtinga(\histb) \cdot f(\histb)(\hista)
        & \text{(By definition)} \medskip \\
        &=&
            \sum{\histb \in \supp(\wtinga)}
                \sum{\hista \in \supp(\wtinga \bind f)}
                    \wtinga(\histb) \cdot f(\histb)(\hista)
        & \text{(associativity of countable summation)} \medskip \\
        &=&
            \sum{\histb \in \supp(\wtinga)}
                \sum{\hista \in \supp(f(\histb))}
                    \wtinga(\histb) \cdot f(\histb)(\hista)
        & \text{(*)} \\
        &=&
            \sum{\histb \in \supp(\wtinga)}
                \wtinga(\histb) \cdot
                    \sum{\hista \in \supp(f(\histb))}
                        f(\histb)(\hista)
        & \text{(\Cref{def:omega-comp-semirings})} \medskip \\
        &=&
            \sum{\histb \in \supp(\wtinga)}
                \wtinga(\histb) \cdot
                    \sum{\hista \in \histories}
                        \markova(\{\histb\},\{\{\hista\}\})
        & \text{($f \sim \markova$)} \\
        &=&
            \sum{\histb \in \histories}
                \measurea(\{\{\histb\}\}) \cdot
                    \sum{\hista \in \histories}
                        \markova(\{\histb\},\{\{\hista\}\})
        & \text{($\wtinga \sim \measurea$)} \\
        &=&
            \sum{\histb \in \histories}
                \measurea(\{\{\histb\}\}) \cdot
                    \sum{\hista \in \histories}
                        \markova(\{\histb\},\{\{\hista\}\})
        & \text{($f \sim \markova$)} \\
        &=&
            \sum{\histb \in \histories}
                \sum{\hista \in \histories}
                    \measurea(\{\{\histb\}\}) \cdot
                        \markova(\{\histb\},\{\{\hista\}\})
        & \text{(\Cref{def:omega-comp-semirings})} \\
        &=&
            \sum{\hista \in \histories}
                \sum{\histb \in \histories}
                    \measurea(\{\{\histb\}\}) \cdot
                        \markova(\{\histb\},\{\{\hista\}\})
        & \text{(associativity of countable summation)} \\
        &=&
            \sum{\hista \in \histories}
                \sum{\histb \in \histories}
                    \markova(\{\histb\},\{\{\hista\}\}) \cdot
                        \measurea(\{\{\histb\}\})
        & \text{(Standard arithmetic)} \\
        &=&
            \sum{\hista \in \histories} (\measurea \bind \markova)(\{\{\hista\}\})
        & \text{(By definition)}
        \end{array}
        \]

    * This equation is due to the following:
    \[
    \begin{array}{rclr}
    && \hista \in \supp(\wtinga \bind f) \\
    &\iff&
        \hista \in \supp \left(
            \mylam{\hista} \semisum{\hista' \in \supp(\wtinga)}
                \wtinga(\hista') \semimul f(\hista')(\hista) \right)
    & \text{(By definition)} \\
    &\iff&
        \supp(\wtinga) \neq \varnothing \wedge \hista \in \supp(f(\hista'))
    & \text{(Simple arithmetic)}
    \end{array}
    \]

    Furthermore, as we are considering an expression of the form:
    \[
        \semisum{\hista' \in \supp(\wtinga)}
                {\semisum{\hista \in \supp(\wtinga \bind f)}{...}}
    \]

    We can rewrite it as done above.

    \item{$\forall \hista \in \supp(\wtinga \bind f) \colon
            (\wtinga \bind f)(\hista) = (\measurea \bind \markova)(\{\{\hista\}\})$}

        \[
        \begin{array}{rclr}
        && (\measurea \bind \markova)(\{\{\hista\}\}) \\
        &=&
            \sum{\histb \in \histories}
                \markova(\{\histb\},\{\{\hista\}\}) \cdot \measurea(\{\{\histb\}\})
        & \text{(By definition)} \\
        &=&
            \sum{\histb \in \histories}
                f(\histb)(\hista) \cdot \measurea(\{\{\histb\}\})
        & \text{($f \sim \markova$)} \\
        &=&
            \sum{\histb \in \histories}
                f(\histb)(\hista) \cdot \wtinga(\histb)
        & \text{($\wtinga \sim \measurea$)} \\
        &=&
            \sum{\histb \in \supp(\wtinga)}
                f(\histb)(\hista) \cdot \wtinga(\histb)
        & \text{(Standard arithmetic)} \\
        &=&
            \sum{\histb \in \supp(\wtinga)}
                \wtinga(\histb) \cdot f(\histb)(\hista)
        & \text{(Standard arithmetic)} \\
        &=& (\wtinga \bind f)(\hista)
        \end{array}
        \]\qedhere
\end{itemize}
\end{proof}

\begin{lemma}[Scalar multiplication respects equivalence]
\label{lmm:scalar-mul-equiv-wnetkat-gprobnetkat}
\[
\begin{array}{lrcl}
    \forall \wtinga \in \weightings[\nonnegRealsSemiring]{\histories},
        \measurea \in \measures{\setsof{\histories}} \colon &
            \wtinga \sim \measurea &\Rightarrow&
                \proba \cdot \wtinga \sim \proba \cdot \measurea
\end{array}
\]
\end{lemma}
\begin{proof}
    We prove all three conditions of \Cref{def:weightings-measures-equiv} separately.
    \begin{itemize}
        \item{$\sum{\{\histsa \in \setsof{\histories} \ |\ 1 < |\histsa|\}}
            (\proba \cdot \measurea)(\{\histsa\}) = 0$} follows immediately from the property holding for $\measurea$. \\

        \item{$\mass{\proba \cdot \wtinga} =
            \sum{\hista \in \histories}
                (\proba \cdot \measurea)(\{\{\hista\}\})$}

            \[
            \begin{array}{rclr}
            &&
                \mass{\proba \cdot \wtinga} \\
            &=&
                \sum{\hista \in \supp(\proba \cdot \wtinga)}
                    (\proba \cdot \wtinga)(\hista)
            & \text{(By definition)} \\
            &=&
                \sum{\hista \in \supp(\proba \cdot \wtinga)}
                    \proba \cdot \wtinga(\hista)
            & \text{(By definition)} \\
            &=&
                \sum{\hista \in \supp(\wtinga)}
                        \proba \cdot \wtinga(\hista)
            & \text{(Standard arithmetic)} \\
            &=&
                \proba \cdot
                    \sum{\hista \in \supp(\wtinga)}
                        \wtinga(\hista)
            & \text{(Standard arithmetic)} \\
            &=&
                \proba \cdot
                    \sum{\hista \in \histories}
                        \measurea(\{\{\hista\}\})
            & \text{(Assumption)} \\
            &=&
                \sum{\hista \in \histories}
                    \proba \cdot
                        \measurea(\{\{\hista\}\})
            & \text{(Standard arithmetic)} \\
            &=&
                \sum{\hista \in \histories}
                    (\proba \cdot
                        \measurea)(\{\{\hista\}\})
            & \text{(By definition)}
            \end{array}
            \]

        \item{$\forall \hista \in \supp(\proba \cdot \wtinga) \colon
            (\proba \cdot \wtinga)(\hista) = (\proba \cdot \measurea)(\{\{\hista\}\})$}

        \[
        \begin{array}{rclr}
        &&
            (\proba \cdot \wtinga)(\hista) \\
        &=&
            \proba \cdot \wtinga(\hista)
        & \text{(By definiton)} \\
        &=&
            \proba \cdot \measurea(\{\{\hista\}\})
        & \text{(Assumption*)} \\
        &=&
            (\proba \cdot \measurea)(\{\{\hista\}\})
        & \text{(By definition)}
        \end{array}
        \]

        * Note that $\hista \in \supp(\proba \cdot \wtinga) \Rightarrow \hista \in \supp(\wtinga)$.
    \end{itemize}
\end{proof}

\begin{lemma}[Addition respects equivalence across \iwnetkat{$\probSemiring$}
    and guarded \probnetkat]
\label{lmm:addition-equiv-wnetkat-gprobnetkat}
\[
\begin{array}{lrcl}
    \forall \wtinga_1,\wtinga_2 \in \weightings[\nonnegRealsSemiring]{\histories},
        \measurea_1,\measurea_2 \in \measures{\setsof{\histories}} \colon &
            \wtinga_1 \sim \measurea_1 \wedge
                \wtinga_2 \sim \measurea_2 &\Rightarrow&
                \wtinga_1 + \wtinga_2 \sim \measurea_1 + \measurea_2
\end{array}
\]
\end{lemma}
\begin{proof}
    We prove all three conditions of \Cref{def:weightings-measures-equiv} separately.
    \begin{itemize}[leftmargin=*]
        \item{$\sum{\{\histsa \in \setsof{\histories} \ |\ 1 < |\histsa|\}}
            (\measurea_1 + \measurea_2)(\{\histsa\}) = 0$} follows immediately from the property holding for $\measurea_1$ and $\measurea_2$. \\

        \item{$\mass{\wtinga_1 + \wtinga_2} =
            \sum{\hista \in \histories}
                (\measurea_1 + \measurea_2)(\{\{\hista\}\})$}

            \[
            \begin{array}{rclr}
            &&
                \mass{\wtinga_1 + \wtinga_2} \\
            &=&
                \sum{\hista \in \supp(\wtinga_1 + \wtinga_2)}
                    (\wtinga_1 + \wtinga_2)(\hista)
            & \text{(By definition)} \\
            &=&
                \sum{\hista \in \supp(\wtinga_1 + \wtinga_2)}
                    \wtinga_1(\hista) + \wtinga_2(\hista)
            & \text{(By definition)} \\
            &=&
                \left( \sum{\hista \in \supp(\wtinga_1 + \wtinga_2)}
                    \wtinga_1(\hista) \right) +
                \sum{\hista \in \supp(\wtinga_1 + \wtinga_2)}
                    \wtinga_2(\hista)
            & \text{(Standard arithmetic)} \\
            &=&
                \left( \sum{\hista \in \supp(\wtinga_1)}
                    \wtinga_1(\hista) \right) +
                \sum{\hista \in \supp(\wtinga_2)}
                    \wtinga_2(\hista)
            & \text{($\supp(\wtinga_i(\hista)) \subseteq \supp(\wtinga_1 + \wtinga_2)$)} \\
            &=&
                \left( \sum{\hista \in \histories}
                    \measurea_1(\{\{\hista\}\}) \right) +
                \sum{\hista \in \supp(\wtinga_2)}
                    \wtinga_2(\hista)
            & \text{(Assumption)} \\
            &=&
                \left( \sum{\hista \in \histories}
                    \measurea_1(\{\{\hista\}\}) \right) +
                \sum{\hista \in \histories}
                    \measurea_2(\{\{\hista\}\})
            & \text{(Assumption)} \\
            &=&
                \sum{\hista \in \histories}
                    \measurea_1(\{\{\hista\}\}) + \measurea_2(\{\{\hista\}\})
            & \text{(Standard arithmetic)} \\
            &=&
                \sum{\hista \in \histories}
                    (\measurea_1 + \measurea_2)(\{\{\hista\}\})
            & \text{(By definition)} \\
            \end{array}
            \]

        \item{$\forall \hista \in \supp(\wtinga_1 + \wtinga_2) \colon
            (\wtinga_1 + \wtinga_2)(\hista) = (\measurea_1 + \measurea_2)(\{\{\hista\}\})$}

        \[
        \begin{array}{rclr}
        &&
            (\wtinga_1 + \wtinga_2)(\hista) \\
        &=&
            \wtinga_1(\hista) + \wtinga_2(\hista)
        & \text{(By definition)} \\
        &=&
            \measurea_1(\{\{\hista\}\}) + \wtinga_2(\hista)
        & \text{(Assumption*)} \\
        &=&
            \measurea_1(\{\{\hista\}\}) + \measurea_2(\{\{\hista\}\})
        & \text{(Assumption*)} \\
        &=&
            (\measurea_1 + \measurea_2)(\{\{\hista\}\})
        & \text{(By definition)} \\
        \end{array}
        \]

        * We have $\hista \in \supp(\wtinga_i)$ or $\wtinga_i(\hista) = 0$,
        in either case we have $\wtinga_i(\hista) = \measurea_i(\{\{\hista\}\})$
        by assumption.
    \end{itemize}
\end{proof}

\begin{lemma}[Predicate equivalence across \iwnetkat{$\probSemiring$} and
    guarded \probnetkat]
\label{lmm:predicate-equiv-wnetkat-gprobnetkat}
For any predicate $\testa$ in \probnetkat, we have
\(
    \forall \hista \in \histories \colon \sem{\transpol{\testa}} \sim \altsem{\testa}
\)
\end{lemma}
\begin{proof}
    By structural induction on the predicate $\testa$. For all cases, we
    have that
    $$
    \sum{\{\histsa \in \setsof{\histories} \ |\ 1 < |\histsa|\}}
        \altsem{\testa}(\hista)(\{\histsa\})
    = 0$$
    by \Cref{lmm:singleton-support-approx}. We prove only the remaining properties
    of \Cref{def:weightings-measures-equiv}.

    \emph{Case $\testa = \FALSE$:}

        We need only show that
        $\mass{\sem{\transpol{\FALSE}}(\hista)} =
            \sum{\histb \in \histories} \altsem{\FALSE}(\hista)(\{\{\histb\}\})$.
        We have:
        $$
        \mass{\sem{\transpol{\FALSE}}(\hista)}
        = 0
        = \sum{\histb \in \histories} \delta_{\varnothing}(\{\{\histb\}\})
        = \sum{\histb \in \histories} \altsem{\FALSE}(\hista)(\{\{\histb\}\})
        $$

    \emph{Case $\testa = \TRUE$:} By \Cref{lmm:equiv-eta-dirac}.

    \emph{Case $\testa = \EQ{\fielda}{\vala}$:}

        Let $h = \hcons{\pkt}{\hista}$. We assume $\pkt.\fielda = \vala$, otherwise
        the same reasoning as in the $\FALSE$ case applies.
        \begin{enumerate}
            \item{$\mass{\sem{\transpol{\EQ{\fielda}{\vala}}}(\hcons{\pkt}{\hista})} =
                \sum{\histb \in \histories}
                    \altsem{\EQ{\fielda}{\vala}}(\hcons{\pkt}{\hista})(\{\{\histb\}\})$:}
                \[
                \begin{array}{rclr}
                    \mass{\sem{\transpol{\EQ{\fielda}{\vala}}}(\hcons{\pkt}{\hista})}
                &=&
                    \unit(\hcons{\pkt}{\hista})(\hcons{\pkt}{\hista})
                & \text{(By definition)} \\
                &=&
                    \delta_{\{\hcons{\pkt}{\hista}\}}(\{\hcons{\pkt}{\hista}\})
                & \text{(\Cref{lmm:equiv-eta-dirac})} \\
                &=&
                    \altsem{\EQ{\fielda}{\vala}}(\{\hcons{\pkt}{\hista}\})(\{\{\hcons{\pkt}{\hista}\}\})
                & \text{(By definition)} \\
                &=&
                    \sum{\histb \in \histories}
                        \altsem{\EQ{\fielda}{\vala}}(\{\hcons{\pkt}{\hista}\})(\{\{\histb\}\})
                & \text{(Standard arithmetic)}
                \end{array}
                \]
            \item{$\forall \histb \in \supp(\sem{\transpol{\EQ{\fielda}{\vala}}}(\hista)) \colon
                \sem{\transpol{\EQ{\fielda}{\vala}}}(\hista)(\histb) =
                    \altsem{\EQ{\fielda}{\vala}}(\{\hista\})(\{\{\histb\}\})$} \\

                Let $\hista = \hcons{\pkt}{\hista}$. Because
                $\histb \in \supp(\sem{\transpol{\EQ{\fielda}{\vala}}}(\hista))$,
                we have that $\histb = \hcons{\pkt}{\hista}$ and that
                $\pkt.\fielda = \vala$. Therefore,
                $$
                \sem{\transpol{\EQ{\fielda}{\vala}}}(\hcons{\pkt}{\hista})(\hcons{\pkt}{\hista})
                = 1
                = \altsem{\EQ{\fielda}{\vala}}(\{\hcons{\pkt}{\hista}\})(\{\{\hcons{\pkt}{\hista}\}\})
                $$
        \end{enumerate}

    \emph{Case $\testa = \OR{\testa}{\testb}$:}

        We proceed by a further case analysis on $\altsem{\testa}(\{\hista\})$. We have the
        following two cases (\citet[Lemma 2]{cantor}):
        \begin{enumerate}
            \item{$\altsem{\testa}(\{\hista\}) = \delta_\varnothing$ (i.e. the case ``false'')}\\

                By IH for $\testa$, $\sem{\transpol{\testa}}(\hista) = 0$. We have:
                \[
                \begin{array}{rclr}
                    \sem{\transpol{\OR{\testa}{\testb}}}(\hista)
                &=&
                    \sem{\transpol{\testb}}(\hista)
                & \text{(By definition)} \\\\

                    \altsem{\OR{\testa}{\testb}}(\{\hista\})
                &=&
                    \altsem{\testb}(\{\hista\})
                & \text{(NetKAT axioms \cite[Corollary 1]{probnetkat})}
                \end{array}
                \]

                Finally, we have
                $\sem{\transpol{\testb}}(\hista) \sim \altsem{\testb}(\{\hista\})$
                by IH for $\testb$. \medskip

            \item{$\altsem{\testa}(\{\hista\}) = \delta_{\{\hista\}}$ (i.e. the case ``true'')}

                By IH for $\testa$, $\sem{\transpol{\testa}}(\hista) = \eta(\hista)$. We have:
                \[
                \begin{array}{rclr}
                    \sem{\transpol{\OR{\testa}{\testb}}}(\hista)
                &=&
                    \eta(\hista)
                & \text{(By definition)} \\\\

                    \altsem{\OR{\testa}{\testb}}(\{\hista\})
                &=&
                    \delta_{\{\hista\}}
                & \text{(NetKAT axioms \cite[Corollary 1]{probnetkat})}
                \end{array}
                \]

                Therefore, we can apply \Cref{lmm:equiv-eta-dirac}.
        \end{enumerate}

    \emph{Case $\testa = \AND{\testa}{\testb}$:}

    We unfold the LHS/RHS first:
    \[
    \begin{array}{rclr}
        \sem{\transpol{\AND{\testa}{\testb}}}(\hista)
        &=& \sem{\transpol{\testa}}(\hista) \bind \sem{\transpol{\testb}} \medskip \\

        \altsem{\AND{\testa}{\testb}}(\{\hista\})
    &=& \altsem{\testa}(\{\hista\}) \bind \altsem{\testb}
    \end{array}
    \]

    Finally, we have:
    \[
    \begin{array}{rclr}
        \sem{\transpol{\testa}}(\hista) \bind \sem{\transpol{\testb}}
    &\sim&
        \altsem{\testa}(\{\hista\}) \bind \altsem{\testb}
    & \text{(\Cref{lmm:bind-equiv-wnetkat-gprobnetkat} and IH)}
    \end{array}
    \]

    \emph{Case $\testa = \NOT{\testa}$:}

        We proceed by a further case analysis on $\altsem{\testa}(\{\hista\})$.
        We have the following two cases (\citet[Lemma 2]{cantor}):
        \begin{enumerate}
            \item{$\altsem{\testa}(\{\hista\}) = \delta_\varnothing$ (i.e. the case ``false''):}\\
                By IH, $\sem{\transpol{\testa}}(\hista) = 0$. We have:
                \[
                \begin{array}{rcl}
                    \sem{\transpol{\NOT{\testa}}}(\hista) &=& \eta(\hista) \medskip \\
                    \altsem{\NOT{\testa}}(\{\hista\}) &=& \delta_{\{\hista\}}
                \end{array}
                \]
                Therefore, we can apply \Cref{lmm:equiv-eta-dirac}.
            \item{$\altsem{\testa}(\{\hista\}) = \delta_{\{\hista\}}$ (i.e. the case ``true'':)}\\
                By IH, $\sem{\transpol{\testa}}(\hista) = \eta(\hista)$. We have:
                \[
                \begin{array}{rcl}
                    \sem{\transpol{\NOT{\testa}}}(\hista) = 0 \medskip \\
                    \altsem{\NOT{\testa}}(\{\hista\}) = \delta_\varnothing
                \end{array}
                \]
                The same reasoning as in the case $\testa = \FALSE$ applies.
        \end{enumerate}
\end{proof}

\begin{lemma}[Conditional branching equivalence across \iwnetkat{$\probSemiring$}
and guarded \probnetkat]
\label{lmm:if-then-else-equiv-wnetkat-gprobnetkat}
\[
\begin{array}{rrcl}
    \forall \testa,\pola,\pola^\prime,\polb,\polb^\prime \colon &
            \approxn[n]{\sem{\pola^\prime}} \sim \approxn[n]{\altsem{\pola}} \wedge
                \approxn[n]{\sem{\polb^\prime}} \sim \approxn[n]{\altsem{\polb}} &\Rightarrow&
                    \approxn[n]{\sem{\IF[\transpol{\testa}]{\pola^\prime}{\polb^\prime}}}
                        \sim
                    \approxn[n]{\altsem{\IF[\testa]{\pola}{\polb}}}
\end{array}
\]
\end{lemma}
\begin{proof}
We proceed by case analysis on $\altsem{\testa}(\{\hista\})$. We have
the following two cases (\citet[Lemma 2]{cantor}).
\begin{enumerate}
    \item{$\altsem{\testa}(\{\hista\}) = \delta_\varnothing$ (i.e. ``false'')}\\

        By \Cref{lmm:predicate-equiv-wnetkat-gprobnetkat}, we also have
        that $\sem{\testa}(\hista) = 0$. Therefore:

        \[
        \begin{array}{rclr}
            \approxn[n]{\sem{\IF[\transpol{\testa}]{\pola^\prime}{\polb^\prime}}}
        &=&
            \approxn[n]{\sem{\pola^\prime}}
        & \text{(\Cref{thm:weightings-props})} \\
        &\sim&
            \approxn[n]{\altsem{\pola}}
        & \text{(Assumption)} \\
        &=&
            \approxn[n]{\altsem{\IF[\testa]{\pola}{\polb}}}
        & \text{(NetKAT axioms \cite[Corollary 1]{probnetkat})}
        \end{array}
        \]
    \item{$\altsem{\testa}(\{\hista\}) = \delta_{\{\hista\}}$ (i.e. ``true'')}\\

        By \Cref{lmm:predicate-equiv-wnetkat-gprobnetkat}, we also have
        that $\sem{\testa}(\hista) = \eta$. Therefore:

        \[
        \begin{array}{rclr}
            \approxn[n]{\sem{\IF[\transpol{\testa}]{\pola^\prime}{\polb^\prime}}}
        &=&
            \approxn[n]{\sem{\polb^\prime}}
        & \text{(\Cref{thm:weightings-props})} \\
        &\sim&
            \approxn[n]{\altsem{\polb}}
        & \text{(Assumption)} \\
        &=&
            \approxn[n]{\altsem{\IF[\testa]{\pola}{\polb}}}
        & \text{(NetKAT axioms \cite[Corollary 1]{probnetkat})}
        \end{array}
        \]
\end{enumerate}
\end{proof}

\begin{lemma}[Guarded iteration equivalence across \iwnetkat{$\probSemiring$}
    and guarded \probnetkat]
\label{lmm:guarded-iteration-equiv-wnetkat-gprobnetkat}
Given a policy $\pola$ and predicate $\testa$ in the guarded fragment of
\probnetkat:
\[
\begin{array}{lrcl}
    \forall n \in \nats \colon &
    \approxn[n]{\sem{\transpol{\pola}}} \sim \approxn[n]{\altsem{\pola}} &\Rightarrow&
        \approxn[n]{\sem{\NFOLD[n,\transpol{\testa}]{\transpol{\pola}}}} \sim
            \approxn[n]{\altsem{\ALTNFOLD[n,\testa]{\pola}}}
\end{array}
\]
\end{lemma}
\begin{proof}
    By induction on $n$.

    \emph{Case $n=0$:}
        \[
        \begin{array}{rclr}
            \sem{\NFOLD[0,\transpol{\testa}]{\transpol{\pola}}}
        &=&
            \sem{\NOT{\transpol{\testa}}}
        & \text{(By definition)} \\
        &\sim&
            \altsem{\NOT{\testa}}
        & \text{(\Cref{lmm:predicate-equiv-wnetkat-gprobnetkat})} \\
        &=&
            \altsem{\ALTNFOLD[0,\testa]{\pola}}
        & \text{(By definition)}
        \end{array}
        \]

    \emph{Case $n=n+1$:}

        We first unfold the LHS/RHS:
        \[
        \begin{array}{rclr}
            \sem{\NFOLD[n+1,\transpol{\testa}]{\transpol{\pola}}}
        &=&
            \sem{\IF[\transpol{\testa}]{\SEQ{\transpol{\pola}}
                                            {\NFOLD[n,\transpol{\testa}]{\transpol{\pola}}}}
                                       {\SKIP}} \medskip \\

            \altsem{\ALTNFOLD[n+1,\testa]{\pola}}
        &=&
            \altsem{\IF[\testa]{\SEQ{\pola}
                                    {\ALTNFOLD[n,\testa]{\pola}}}
                               {\SKIP}}
        \end{array}
        \]

        We have that
        $\approxn[n]{\sem{\SKIP}} \sim \approxn[n]{\altsem{\SKIP}}$,
        so by \Cref{lmm:if-then-else-equiv-wnetkat-gprobnetkat} we need
        only show that:
        $$
        \sem{\SEQ{\transpol{\pola}}
                 {\NFOLD[n,\transpol{\testa}]{\transpol{\pola}}}}
                 \sim
        \altsem{\SEQ{\pola}
                    {\ALTNFOLD[n,\testa]{\pola}}}
        $$

        Unfolding the LHS/RHS:
        \[
        \begin{array}{rclr}
            \sem{\SEQ{\transpol{\pola}}
                     {\NFOLD[n,\transpol{\testa}]{\transpol{\pola}}}}(\hista)
        &=&
            \sem{\transpol{\pola}}(\hista) \bind
                     \sem{\NFOLD[n,\transpol{\testa}]{\transpol{\pola}}}
            & \text{(By definition)} \medskip \\

            \altsem{\SEQ{\pola}
                        {\ALTNFOLD[n,\testa]{\pola}}}(\{\hista\})
        &=&
            \altsem{\pola}(\{\hista\}) \bind
                     \altsem{\ALTNFOLD[n,\testa]{\pola}}
        & \text{(By definition)}
        \end{array}
        \]

        By \Cref{lmm:bind-equiv-wnetkat-gprobnetkat}, we need only show:
        \[
        \begin{array}{lrcl}
            (1) & \sem{\transpol{\pola}}(\hista) &\sim& \altsem{\pola}(\{\hista\}) \\
            (2) & \sem{\NFOLD[n,\testa]{\transpol{\pola}}} &\sim& \altsem{\ALTNFOLD[n,\testa]{\pola}}
        \end{array}
        \]

        The first condition holds by assumption for $\pola$, the second condition
        holds by IH for $n$.
\end{proof}

\section{Language model}
\subsection{Soundness of Reduction to Reduced \wnetkat Syntax}
\label{appendix:reduced}

We define the reduction map from \wnetkat policies to reduced \wnetkat policies (\Cref{fig:reduced-syntax}) as in the following figure.

\begin{figure}[h]
     \begin{tabular}{l  l}
	\toprule
	\toprule
    $\boldsymbol{\pola \in \pols}$ & $\boldsymbol{\mathsf{\red(\pola) \in \rpols}}$\\
	\midrule
    $\testa$ \quad \quad& $\SUM{\pkt \in \packets} \WEIGH{[\ctest{\pkt} \leq \testa]}{\ctest{\pkt}}$ \\[1.2ex]
    $\ASSN{\fielda}{\vala}$ \quad \quad& $\SUM{\pkt \in \packets} \SEQ{\ctest{\pkt}}{\cass{(\updatepkt{\pkt}{\fielda}{\vala})}}$ \\[1.2ex]
	$\DUP$ \quad \quad& $\DUP$ \\[1ex]
	$\SEQ{\pola_1}{\pola_2}$ \quad \quad& $\SEQ{\red(\pola_1)}{\red(\pola_2)}$ \\[1ex]
	$\WEIGH{\wta}{\pola_1}$ \quad \quad& $\WEIGH{\wta}{\red{(\pola_1)}}$ \\[1ex]
	$\ADD{\pola_1}{\pola_2}$ \quad \quad& $\ADD{\red(\pola_1)}{\red(\pola_2)}$ \\[1ex]
	$\ITER{\pola_1}$ \quad \quad& $\ITER{\red(\pola_1)}$ \\[1ex]
	\bottomrule
	\bottomrule
    \end{tabular}
\caption{Reduction map for \wnetkat policies to reduced \wnetkat policies.}\label{fig:reduced}
\end{figure}
Importantly, we have that every \wnetkat policy is semantically equivalent to its reduced
counterpart.

\begin{lemma}[Soundness of Reduction]
\label{thm:reduced}
For all $\hista \in \histories$ and $\pola \in \pols$, we have that $\sem{\pola}(\hista) = \sem{\red(\pola)}(\hista)$, where $\red(\pola) \in \rpols$ is the reduced policy as defined in \Cref{fig:reduced}.
\end{lemma}
\begin{proof}
	By induction on the structure of $\pola$, for any $\hista \in \histories$.
	Base cases:

	\begin{align*}
		\sem{\red(\testa)}(\hista) &= \sem{\testa}(\hista) \tag{\Cref{lmm:reduced-pred}}
	\end{align*}
	\begin{align*}
		\sem{\red(\DUP)}(\hista) &= \sem{\DUP}(\hista) \tag{by definition}
	\end{align*}
	\begin{align*}
		\sem{\red(\ASSN{\fielda}{\vala})}(\hista)
        &= \sem{\SUM{\ctest{\pkta} \in \Test} \ctest{\pkta} \SEQN \cass{(\updatepkt{\pkta}{\fielda}{\vala})}}(\hista) \tag{def of $\red(\ASSN{\fielda}{\vala})$} \\
        &= \semisum{\ctest{\pkta} \in \Test} \sem{\ctest{\pkta} \SEQN \cass{(\updatepkt{\pkta}{\fielda}{\vala})}}(\hista) \tag{def of $\sem{\pola \ADDN \polb}$} \\
        &= \semisum{\ctest{\pkta} \in \Test} \sem{\ctest{\pkta}}(\hista) \bind \sem{\cass{(\updatepkt{\pkta}{\fielda}{\vala})}}(\hista) \tag{def of $\sem{\pola \SEQN \polb}$} \\
        &= \semisum{\ctest{\pkta} \in \Test} \lambda \histc.\; \semisum{\histb \in \supp(\sem{\ctest{\pkta}}(\hista))} \sem{\ctest{\pkta}}(\hista)(\histb) \semimul \sem{\cass{(\updatepkt{\pkta}{\fielda}{\vala})}}(\histb)(\histc) \tag{def of $\bind$} \\
         &= \lambda \histc.\; \semisum{\ctest{\pkta} \in \Test} \semisum{\histb \in \supp(\sem{\ctest{\pkta}}(\hista))} \sem{\ctest{\pkta}}(\hista)(\histb) \semimul \sem{\cass{(\updatepkt{\pkta}{\fielda}{\vala})}}(\histb)(\histc) \tag{lambda equivalence} \\
        &= \lambda (\pkt'' \cons \hbar'').\; \semisum{\ctest{\pkta} \in \Test} \semisum{(\pkt' \cons \hbar') \in \supp(\sem{\ctest{\pkta}}(\pkt \cons \hbar))} \sem{\ctest{\pkta}}(\pkt \cons \hbar)(\pkt' \cons \hbar') \semimul \sem{\cass{(\updatepkt{\pkta}{\fielda}{\vala})}}(\pkt' \cons \hbar')(\pkt'' \cons \hbar'') \tag{expansion of $\hista, \histb, \histc$} \\
        &= \lambda (\pkt'' \cons \hbar'').\; \semisum{\ctest{\pkta} \in \Test} \semisum{(\pkt' \cons \hbar') \in \supp(\sem{\ctest{\pkta}}(\pkt \cons \hbar))} [\hbar = \hbar' = \hbar'' \wedge \pkt = \pkt' = \pkta \wedge \pkt'' = \updatepkt{\pkta}{\fielda}{\vala}] \tag{\Cref{prop:complete-test,prop:complete-assignment}} \\
		&= \lambda (\pkt'' \cons \hbar'').\; \semisum{\{(\pkt \cons \hbar)\}} [\hbar = \hbar'' \wedge \pkt = \pkta \wedge \pkt'' = \updatepkt{\pkta}{\fielda}{\vala}] \tag{singleton complete test support} \\
		&= \lambda (\pkt'' \cons \hbar'').\; [\hbar = \hbar'' \wedge \pkt = \pkta \wedge \pkt'' = \updatepkt{\pkta}{\fielda}{\vala}] \tag{singleton sum} \\
		&= \lambda (\pkt'' \cons \hbar'').\; [\hbar = \hbar'' \wedge \pkt'' = \updatepkt{\pkt}{\fielda}{\vala}]] \tag{propositional equivalence} \\
        &= \sem{\ASSN{\fielda}{\vala}}(\hista) \tag{semantics of $\ASSN{\fielda}{\vala}$}
	\end{align*}

	Inductive cases:
	\begin{align*}
		\sem{\red(\pola \ADDN \polb)}
		&= \sem{\red(\pola) \ADDN \red(\polb)} \tag{by definition} \\
		&= \sem{\pola \ADDN \polb} \tag{by IH on $\pola$ and $\polb$}
	\end{align*}
	\begin{align*}
		\sem{\red(\pola \SEQN \polb)}
		&= \sem{\red(\pola) \SEQN \red(\polb)} \tag{by definition} \\
		&= \sem{\pola \SEQN \polb} \tag{by IH on $\pola$ and $\polb$}
	\end{align*}
	\begin{align*}
		\sem{\red(\WEIGH{\wta}{\pola})}
		&= \sem{\WEIGH{\wta}{\red(\pola)}} \tag{by definition} \\
		&= \sem{\WEIGH{\wta}{\pola}} \tag{by IH on $\pola$}
	\end{align*}
	\begin{align*}
		\sem{\red(\ITER{\pola})}
		&= \sem{\ITER{\red(\pola)}} \tag{by definition} \\
		&= \sem{\ITER{\pola}} \tag{by IH on $\pola$}
	\end{align*}
\end{proof}

\subsection{Lemmas for Soundness of Reduction to Reduced \wnetkat Syntax}
\begin{proposition}
    \label{prop:complete-test}
    Every complete test $\ctest{\pkt} \triangleq \SEQ{\SEQ{\EQ{\fielda_1}{\vala_1}}{\dots}}{\EQ{\fielda_k}{\vala_k}}$ matches exactly one packet:
    \[ \sem{\ctest{\pkt}}(\pkta \cons \hbar)(\pktb \cons \hbar') = [\hbar = \hbar' \wedge \pkta = \pktb = \{\EQ{\fielda_1}{\vala_1}, \dots, \EQ{\fielda_k}{\vala_k}\}] \]
\end{proposition}

\begin{proposition}
    \label{prop:complete-assignment}
    Every complete assignment $\cass{\pkt} \triangleq \SEQ{\SEQ{\ASSN{\fielda_1}{\vala_1}}{\dots}}{\ASSN{\fielda_k}{\vala_k}}$ matches exactly one packet:
    \[ \sem{\cass{\pkt}}(\pkta \cons \hbar)(\pktb \cons \hbar') = [\hbar = \hbar' \wedge \pktb = \{\EQ{\fielda_1}{\vala_1}, \dots, \EQ{\fielda_k}{\vala_k}\}] \]
\end{proposition}

In the following lemma (and accompanying proof) we adopt the convention of denoting
complete tests by $\pkta$,$\pktb$, and $\pktc$, whereas we
denote complete assignments by $\pkt$. We write
$\pkt_\pkta$, $\pkt_\pktb$, etc., to emphasize a complete
assignment to the packet $\pkta$.

\begin{lemma}
    \label{lmm:red-props}
    For all complete tests $\oldctestA \in \Test$ and complete assignments $\oldcass{\oldctestB} \in \Assign$:
    \begin{enumerate}
	\item \label{lmm:red-props1} $\sem{\oldctestA} = \sem{\SEQ{\oldctestA}{\oldcass{\oldctestA}}}$
	\item \label{lmm:red-props2} $\sem{\SEQ{\oldctestA}{\DUP}} = \sem{\SEQ{\DUP}{\oldctestA}}$
	\item \label{lmm:red-props3} $\sem{\TRUE} = \sem{\SUM{\oldctestA \in \Test} \oldctestA}$
	\item \label{lmm:red-props4} $\sem{\oldcass{\oldctestB}} = \sem{\SUM{\oldctestA \in \Test} \oldctestA \SEQN \oldcass{\oldctestB}}$
	\item \label{lmm:red-props5} $\sem{\DUP} = \sem{\SUM{\oldctestA \in \Test} \oldctestA \SEQN \oldcass{\oldctestA} \SEQN \DUP \SEQN \oldcass{\oldctestA}}$
    \end{enumerate}
\end{lemma}
\begin{proof}
	For all histories $\hista \in \histories$:
	\begin{align*}
		\sem{\SEQ{\oldctestA}{\oldcass{\oldctestA}}}(\hista)
		&= \sem{\oldctestA}(\hista) \bind \sem{\oldcass{\oldctestA}} \tag{definition of $\sem{\SEQ{\pola}{\polb}}$} \\
		&= \lambda \histc.\; \semisum{\histb \in \supp(\sem{\oldctestA}(\hista))} \sem{\oldctestA}(\hista)(\histb) \semimul \sem{\oldcass{\oldctestA}}(\histb)(\histc) \tag{definition of $\bind$} \\
		&= \lambda (\pkt''_0 \cons \hbar'').\; \semisum{(\pkt'_0 \cons \hbar') \in \supp(\sem{\oldctestA}(\pkt_0 \cons \hbar))} \sem{\oldctestA}(\pkt_0 \cons \hbar)(\pkt'_0 \cons \hbar') \semimul \sem{\oldcass{\oldctestA}}(\pkt'_0 \cons \hbar')(\pkt''_0 \cons \hbar'') \tag{expansion of $\hista,\histb,\histc$} \\
		&= \lambda (\pkt''_0 \cons \hbar'').\; \semisum{(\pkt'_0 \cons \hbar') \in \supp(\sem{\oldctestA}(\pkt_0 \cons \hbar))} [\hbar = \hbar' \wedge \pkt_0 = \pkt'_0 = \pkt] \semimul {[\hbar' = \hbar'' \wedge \pkt''_0 = \pkt]} \tag{\Cref{prop:complete-test,prop:complete-assignment}} \\
		&= \lambda (\pkt''_0 \cons \hbar'').\; \semisum{(\pkt'_0 \cons \hbar') \in \supp(\sem{\oldctestA}(\pkt_0 \cons \hbar))} {[\hbar = \hbar' = \hbar'' \wedge \pkt_0 = \pkt'_0 = \pkt''_0 = \pkt]} \tag{Iverson bracket multiplication} \\
		&= \lambda (\pkt''_0 \cons \hbar'').\; {[\hbar = \hbar = \hbar'' \wedge \pkt_0 = \pkt_0 = \pkt''_0 = \pkt]} \tag{By \Cref{prop:complete-test}: $\supp(\sem{\oldctestA}(\pkt_0 \cons \hbar)) = \{ \pkt_0 \cons \hbar \mid \pkt_0 = \pkt_\oldctestA \}$} \\
		&= \lambda (\pkt''_0 \cons \hbar'').\; {[\hbar = \hbar'' \wedge \pkt_0 = \pkt''_0 = \pkt]} \tag{predicate simplification} \\
		&= \lambda (\pkt''_0 \cons \hbar'').\; \sem{\oldctestA}(\pkt_0 \cons \hbar)(\pkt''_0 \cons \hbar'') \tag{\Cref{prop:complete-test}} \\
		&= \sem{\oldctestA}(\pkt_0 \cons \hbar) \tag{implicit lambda} \\
		&= \sem{\oldctestA}(\hista) \tag{$\hista = \pkt_0 \cons \hbar$}
	\end{align*}
	\begin{align*}
		\sem{\SEQ{\oldctestA}{\DUP}}(\hista)
		&= \sem{\oldctestA}(\hista) \bind \sem{\DUP} \tag{definition of $\sem{\SEQ{\pola}{\polb}}$} \\
		&= \lambda \histc.\; \semisum{\histb \in \supp(\sem{\oldctestA}(\hista))} \sem{\oldctestA}(\hista)(\histb) \semimul \sem{\DUP}(\histb)(\histc) \tag{definition of $\bind$} \\
		&= \lambda (\pkt''_0 \cons \hbar'').\; \semisum{(\pkt'_0 \cons \hbar') \in \supp(\sem{\oldctestA}(\pkt_0 \cons \hbar))} \sem{\oldctestA}(\pkt_0 \cons \hbar)(\pkt'_0 \cons \hbar') \semimul \sem{\DUP}(\pkt'_0 \cons \hbar')(\pkt''_0 \cons \hbar'') \tag{expansion of $\hista,\histb,\histc$} \\
		&= \lambda (\pkt''_0 \cons \hbar'').\; \semisum{(\pkt'_0 \cons \hbar') \in \supp(\sem{\oldctestA}(\pkt_0 \cons \hbar))} [\hbar = \hbar' \wedge \pkt_0 = \pkt'_0 = \pkt] \semimul {[\pkt''_0 = \pkt'_0 \wedge \hbar'' = \pkt'_0 \cons \hbar']} \tag{\Cref{prop:complete-test,prop:complete-assignment}} \\
		&= \lambda (\pkt''_0 \cons \hbar'').\; \semisum{(\pkt'_0 \cons \hbar') \in \supp(\sem{\oldctestA}(\pkt_0 \cons \hbar))} [\hbar'' = \pkt''_0 \cons \hbar' \wedge \hbar'' = \pkt''_0 \cons \hbar  \wedge \pkt_0 = \pkt'_0 = \pkt''_0 = \pkt] \tag{Iverson bracket multiplication} \\
		&= \lambda (\pkt''_0 \cons \hbar'').\; [\hbar'' = \pkt''_0 \cons \hbar \wedge \hbar'' = \pkt''_0 \cons \hbar  \wedge \pkt_0 = \pkt_0 = \pkt''_0 = \pkt] \tag{By \Cref{prop:complete-test}: $\supp(\sem{\oldctestA}(\pkt_0 \cons \hbar)) = \{ \pkt_0 \cons \hbar \mid \pkt_0 = \pkt_\oldctestA \}$} \\
		&= \lambda (\pkt''_0 \cons \hbar'').\; [\hbar'' = \pkt''_0 \cons \hbar \wedge \pkt_0 = \pkt''_0 = \pkt] \tag{predicate simplification} \\
		&= \lambda (\pkt''_0 \cons \hbar'').\; \semisum{(\pkt'_0 \cons \hbar') \in \supp(\sem{\DUP}(\pkt_0 \cons \hbar))} [\hbar'' = \pkt''_0 \cons \hbar \wedge \pkt_0 = \pkt''_0 = \pkt] \tag{By def. of $\sem{\DUP}$: $|\supp(\sem{\DUP}(\pkt_0 \cons \hbar))| = 1$} \\
		&= \lambda (\pkt''_0 \cons \hbar'').\; \semisum{(\pkt'_0 \cons \hbar') \in \supp(\sem{\DUP}(\pkt_0 \cons \hbar))} [\hbar'' = \hbar' = \pkt_0 \cons \hbar \wedge \pkt_0 = \pkt'_0 = \pkt''_0 = \pkt] \tag{By def. of $\sem{\DUP}$: $\supp(\sem{\DUP}(\pkt_0 \cons \hbar)) = \{ \pkt_0 \cons \pkt_0 \cons \hbar\}$} \\
		&= \lambda (\pkt''_0 \cons \hbar'').\; \semisum{(\pkt'_0 \cons \hbar') \in \supp(\sem{\DUP}(\pkt_0 \cons \hbar))} [\pkt_0 = \pkt'_0 \wedge \hbar' = \pkt_0 \cons \hbar] \semimul {[\pkt'_0 = \pkt''_0 = \pkt \wedge \hbar' = \hbar'']} \tag{Iverson bracket multiplication} \\
		&= \lambda (\pkt''_0 \cons \hbar'').\; \semisum{(\pkt'_0 \cons \hbar') \in \supp(\sem{\DUP}(\pkt_0 \cons \hbar))}  \sem{\DUP}(\pkt_0 \cons \hbar)(\pkt'_0 \cons \hbar') \semimul \sem{\oldctestA}(\pkt'_0 \cons \hbar')(\pkt''_0 \cons \hbar'') \tag{By \Cref{prop:complete-test} and def. of $\sem{\DUP}$} \\
		&= \lambda \histc.\; \semisum{\histb \in \supp(\sem{\DUP}(\hista))}  \sem{\DUP}(\hista)(\histb) \semimul \sem{\oldctestA}(\histb)(\histc) \tag{definition of $\hista, \histb, \histc$} \\
		&= \sem{\DUP}(\hista) \bind \sem{\oldctestA} \tag{definition of $\bind$} \\
		&= \sem{\SEQ{\DUP}{\oldctestA}}(\hista) \tag{definition of $\sem{\SEQ{\pola}{\polb}}$}
	\end{align*}
	\begin{align*}
		\sem{\SUM{\oldctestA \in \Test} \oldctestA}(\hista)
		&= \semisum{\oldctestA \in \Test} \sem{\oldctestA}(\hista) \tag{definition of $\sem{\ADD{\pola}{\polb}}$} \\
		&= \lambda \histb.\; \big( \semisum{\oldctestA \in \Test} \sem{\oldctestA}(\hista) \big)(\histb) \tag{explicit lambda} \\
		&= \lambda \histb.\; \semisum{\oldctestA \in \Test} \sem{\oldctestA}(\hista)(\histb) \tag{definition of lifted $\semiadd$} \\
		&= \lambda (\pkt'_0 \cons \hbar').\; \semisum{\oldctestA \in \Test} \sem{\oldctestA}(\pkt_0 \cons \hbar)(\pkt'_0 \cons \hbar') \tag{expansion of $\hista, \histb$} \\
		&= \lambda (\pkt'_0 \cons \hbar').\; \semisum{\oldctestA \in \Test} [\hbar = \hbar' \wedge \pkt_0 = \pkt'_0 = \pkt] \tag{\Cref{prop:complete-test}, for packet $\pkt$ corresponding to $\oldctestA$} \\
		&= \lambda (\pkt'_0 \cons \hbar').\; [\hbar = \hbar' \wedge \pkt_0 = \pkt'_0] \tag{$\Test \cong \packets$} \\
		&= \lambda (\pkt'_0 \cons \hbar').\; [\pkt_0 \cons \hbar = \pkt'_0 \cons \hbar'] \tag{predicate simplification} \\
		&= \lambda \histb.\; [\hista = \histb] \tag{definition of $\hista, \histb$} \\
		&= \sem{\TRUE}(\hista) \tag{definition of $\sem{\TRUE}$}
	\end{align*}
		\begin{align*}
			\sem{\SUM{\oldctestA \in \Test} \oldctestA \SEQN \oldcass{\oldctestB}}(\hista)
		&= \semisum{\oldctestA \in \Test} \sem{\oldctestA \SEQN \oldcass{\oldctestB}}(\hista) \tag{definition of $\sem{\ADD{\pola}{\polb}}$} \\
		&= \semisum{\oldctestA \in \Test} \sem{\oldctestA}(\hista) \bind \sem{\oldcass{\oldctestB}} \tag{definition of $\sem{\SEQ{\pola}{\polb}}$} \\
		&= \semisum{\oldctestA \in \Test} \lambda \histc.\; \semisum{\histb \in \supp(\sem{\oldctestA}(\hista))} \sem{\oldctestA}(\hista)(\histb) \semimul \sem{\oldcass{\oldctestB}}(\histb)(\histc) \tag{definition of $\bind$} \\
		&= \lambda \histc.\; \semisum{\oldctestA \in \Test} \semisum{\histb \in \supp(\sem{\oldctestA}(\hista))} \sem{\oldctestA}(\hista)(\histb) \semimul \sem{\oldcass{\oldctestB}}(\histb)(\histc) \tag{lambda equivalence} \\
		&= \lambda (\pkt'' \cons \hbar'').\; \semisum{\oldctestA \in \Test} \semisum{(\pkt' \cons \hbar') \in \supp(\sem{\oldctestA}(\pkt \cons \hbar))} \sem{\oldctestA}(\pkt \cons \hbar)(\pkt' \cons \hbar') \semimul \sem{\oldcass{\oldctestB}}(\pkt' \cons \hbar')(\pkt'' \cons \hbar'')  \tag{expansion of $\hista, \histb, \histc$} \\
		&= \lambda (\pkt'' \cons \hbar'').\; \semisum{\oldctestA \in \Test} \semisum{(\pkt' \cons \hbar') \in \supp(\sem{\oldctestA}(\pkt \cons \hbar))} [\hbar = \hbar' \wedge \pkt = \pkt' = \pkt_\oldctestA] \semimul {[\hbar' = \hbar'' \wedge \pkt'' = \pkt_\oldctestB]} \tag{\Cref{prop:complete-test,prop:complete-assignment}} \\
		&= \lambda (\pkt'' \cons \hbar'').\; \semisum{\oldctestA \in \Test} \semisum{(\pkt' \cons \hbar') \in \supp(\sem{\oldctestA}(\pkt \cons \hbar))} [\hbar = \hbar' = \hbar'' \wedge \pkt = \pkt' = \pkt_\oldctestA \wedge \pkt'' = \pkt_\oldctestB] \tag{Iverson bracket multiplication} \\
		&= \lambda (\pkt'' \cons \hbar'').\; \semisum{(\pkt' \cons \hbar') \in \supp(\sem{\oldctestA}(\pkt \cons \hbar))} [\hbar = \hbar' = \hbar'' \wedge \pkt = \pkt' \wedge \pkt'' = \pkt_\oldctestB] \tag{for unique $\oldctestA$ s.t. $\pkt = \pkt_\oldctestA$} \\
		&= \lambda (\pkt'' \cons \hbar'').\; [\hbar = \hbar'' \wedge \pkt'' = \pkt_\oldctestB] \tag{for unique $(\pkt' \cons \hbar') = (\pkt \cons \hbar)$} \\
		&= \lambda (\pkt'' \cons \hbar'').\; \sem{\oldcass{\oldctestB}}(\pkt \cons \hbar)(\pkt'' \cons \hbar'') \tag{\Cref{prop:complete-assignment}} \\
		&= \sem{\oldcass{\oldctestB}}(\pkt \cons \hbar) \tag{implicit lambda} \\
		&= \sem{\oldcass{\oldctestB}}(\hista) \tag{$\hista = (\pkt \cons \hbar)$} \\
	\end{align*}
	\begin{align*}
		\sem{\DUP}
		&= \sem{\DUP \SEQN \TRUE} \tag{by monadic structure of $\sem{\pola \SEQN \polb}$} \\
		&= \sem{\DUP \SEQN (\SUM{\oldctestA \in \Test} \oldctestA)} \tag{\Cref{lmm:red-props}.\ref{lmm:red-props3}} \\
		&= \sem{\DUP \SEQN (\SUM{\oldctestA \in \Test} \oldctestA \SEQN \oldcass{\oldctestA})} \tag{\Cref{lmm:red-props}.\ref{lmm:red-props1}} \\
		&= \sem{\SUM{\oldctestA \in \Test} \DUP \SEQN \oldctestA \SEQN \oldcass{\oldctestA}} \tag{distributivity into the infinitary sum from $\bind$, of $\sem{\pola \SEQN \polb}$} \\
		&= \sem{\SUM{\oldctestA \in \Test} \oldctestA \SEQN \DUP \SEQN \oldcass{\oldctestA}} \tag{\Cref{lmm:red-props}.\ref{lmm:red-props2}} \\
		&= \sem{\SUM{\oldctestA \in \Test} \oldctestA \SEQN \oldcass{\oldctestA} \SEQN \DUP \SEQN \oldcass{\oldctestA}} \tag{\Cref{lmm:red-props}.\ref{lmm:red-props1}}
	\end{align*}

\end{proof}

\begin{lemma}
	\label{lmm:reduced-pred}
	For all $\testa \in \tests$: $\sem{\testa} = \sem{\red(\testa)}$
\end{lemma}
\begin{proof}
    Fix two given histories $\hista, \histb \in \histories$. Then, as $(\Test, \OR{}{}, \AND{}{}, \NOT{}, \FALSE, \TRUE)$ forms a Boolean algebra we have that $\sem{\testa}(\hista)(\histb) = \sem{\bigvee_{\ctest{\pkta} \leq t} \ctest{\pkta}}(\hista)(\histb)$.

    We now show that $\sem{\bigvee_{\ctest{\pkta} \leq t} \ctest{\pkta}}(\hista)(\histb) = \sem{\SUM{\ctest{\pkta} \in \Test} \WEIGH{[\ctest{\pkta} \leq \testa]}{\ctest{\pkta}}}(\hista)(\histb)$. We proceed by case analysis.

    Suppose that $\sem{\bigvee_{\ctest{\pkta} \leq t} \ctest{\pkta}}(\hista)(\histb) = \addid$. Then:
	\begin{align*}
		\sem{\bigvee_{\ctest{\pkta} \leq t} \ctest{\pkta}}(\hista)(\histb) = \addid
		&\implies \forall \ctest{\pkta} \in \Test.\; \ctest{\pkta} \leq \testa \Rightarrow \sem{\ctest{\pkta}}(\hista)(\histb) = \addid \tag{$A \vee B = \bot \implies A = \bot \wedge B = \bot$} \\
		&\implies \forall \ctest{\pkta} \in \Test.\; \ctest{\pkta} \leq \testa \wedge \sem{\ctest{\pkta}}(\hista)(\histb) = \addid \tag{$A \Rightarrow B = \bot \implies A \wedge B = \bot$} \\
		&\implies \forall \ctest{\pkta} \in \Test.\; [\ctest{\pkta} \leq \testa] \semimul \sem{\ctest{\pkta}}(\hista)(\histb) = \addid \tag{Iverson bracket multiplication} \\
		&\implies \forall \ctest{\pkta} \in \Test.\; \sem{\WEIGH{[\ctest{\pkta} \leq \testa]}{\ctest{\pkta}}}(\hista)(\histb) = \addid \tag{def of $\sem{\WEIGH{\wta}{\pola}}$} \\
		&\implies \semisum{\ctest{\pkta} \in \Test} \sem{\WEIGH{[\ctest{\pkta} \leq \testa]}{\ctest{\pkta}}}(\hista)(\histb) = \addid \tag{$\semisum{} \addid = \addid$} \\
		&\implies \sem{\SUM{\ctest{\pkta} \in \Test} \WEIGH{[\ctest{\pkta} \leq \testa]}{\ctest{\pkta}}}(\hista)(\histb) = \addid \tag{def of $\sem{\pola \ADDN \polb}$}
	\end{align*}

	Suppose that $\sem{\bigvee_{\ctest{\pkta} \leq t} \ctest{\pkta}}(\hista)(\histb) = \mulid$. Then:
	\begin{align*}
		\sem{\bigvee_{\ctest{\pkta} \leq t} \ctest{\pkta}}(\hista)(\histb) = \mulid
		&\implies \exists \ctest{\pkta} \in \Test.\; \ctest{\pkta} \leq \testa \wedge \sem{\ctest{\pkta}}(\hista)(\histb) = \mulid \tag{$A \vee B = \top \implies A = \top \vee B = \top$} \\
		&\implies \exists! \ctest{\pkta} \in \Test.\; \ctest{\pkta} \leq \testa \wedge \sem{\ctest{\pkta}}(\hista)(\histb) = \mulid \tag{$\Test \cong \packets$, for $\hista = \pkt \cons \hbar$, only $\sem{\ctest{\pkta}_\pkt}(\hista)(\hista) = \mulid$} \\
		&\implies \exists! \ctest{\pkta} \in \Test.\; [\ctest{\pkta} \leq \testa] \semimul \sem{\ctest{\pkta}}(\hista)(\histb) = \mulid \tag{Iverson bracket multiplication} \\
		&\implies \exists! \ctest{\pkta} \in \Test.\; \sem{\WEIGH{[\ctest{\pkta} \leq \testa]}{\ctest{\pkta}}}(\hista)(\histb) = \mulid \tag{def of $\sem{\WEIGH{\wta}{\pola}}$} \\
		&\implies \semisum{\ctest{\pkta} \in \Test} \sem{\WEIGH{[\ctest{\pkta} \leq \testa]}{\ctest{\pkta}}}(\hista)(\histb) = \mulid \tag{singleton supported index set} \\
        &\implies \sem{\SUM{\ctest{\pkta} \in \Test} \WEIGH{[\ctest{\pkta} \leq \testa]}{\ctest{\pkta}}}(\hista)(\histb) = \mulid \tag{def of $\sem{\pola \ADDN \polb}$}
	\end{align*}

    As $\testa \in \tests$, $\sem{\testa}(\hista)(\histb) \in \{\addid, \mulid\}$, and so $\sem{\testa} = \sem{\bigvee_{\ctest{\pkta} \leq t} \ctest{\pkta}} = \sem{\SUM{\ctest{\pkta} \in \Test} \WEIGH{[\ctest{\pkta} \leq \testa]}{\ctest{\pkta}}} = \sem{\red(\testa)}$.

\end{proof}

\subsection{Properties of the Language Model}

\begin{lemma}
	\label{lmm:supp-weigh}
	For all $\hista, \histb \in \histories$, $\pola \in \pols$, $\wta \in \semi$:
	\[\semisum{x \in \supp(G(\pola))} G(\WEIGH{\wta}{\pola})(x) \semimul \sem{x}(\hista)(\histb)
	\eeq
	\semisum{x \in \supp(G(\WEIGH{\wta}{\pola}))} G(\WEIGH{\wta}{\pola})(x) \semimul \sem{x}(\hista)(\histb)\]
\end{lemma}
\begin{proof}
	For convenience, let $J, K \subseteq \GS$ and $I \colon \GS \rightarrow \semi$ be given by:
	\[J \triangleq \supp(G(\pola)) \qquad K \triangleq \supp(G(\WEIGH{\wta}{\pola})) \qquad I(x) = G(\WEIGH{\wta}{\pola})(x) \semimul \sem{x}(\hista)(\histb)\]
	We will first consider the sums over sets $J' \subseteq J$ and $K' \subseteq K$ of indices leading to non-zero terms:
	\[J' \triangleq \{x \mid I(x) \neq \addid \wedge x \in J\} \qquad K' \triangleq \{x \mid I(x) \neq \addid \wedge x \in K\}\]
	We now show that $J' = K'$. Note that $J' \subseteq K'$ as:
	\begin{align*}
		x \in J'
		&\implies I(x) \neq \addid \wedge x \in J \tag{definition of $J'$} \\
		&\implies I(x) \neq \addid \tag{$A \wedge B \implies A$} \\
		&\implies I(x) \neq \addid \wedge I(x) \neq \addid \tag{$A \implies A \wedge A$} \\
		&\implies I(x) \neq \addid \wedge G(\WEIGH{\wta}{\pola})(x) \semimul \sem{x}(\hista)(\histb) \neq \addid \tag{definition of $I(x)$} \\
		&\implies I(x) \neq \addid \wedge G(\WEIGH{\wta}{\pola})(x) \neq \addid \wedge \sem{x}(\hista)(\histb) = \mulid \tag{$\sem{x}(\hista)(\histb) \in \{\addid, \mulid\}$} \\
		&\implies I(x) \neq \addid \wedge G(\WEIGH{\wta}{\pola})(x) \neq \addid \tag{$A \wedge B \implies A$} \\
		&\implies I(x) \neq \addid \wedge x \in \supp(G(\WEIGH{\wta}{\pola})) \tag{definition of $\supp(G(\WEIGH{\wta}{\pola}))$} \\
		&\implies I(x) \neq \addid \wedge x \in K \tag{definition of $K$} \\
		&\implies x \in K' \tag{definition of $K'$}
	\end{align*}
	And similarly, $K' \subseteq J'$ as:
	\begin{align*}
		x \in K'
		&\implies I(x) \neq \addid \wedge x \in K \tag{definition of $K'$} \\
		&\implies I(x) \neq \addid \wedge x \in \supp(G(\WEIGH{\wta}{\pola})) \tag{definition of $K$} \\
		&\implies I(x) \neq \addid \wedge G(\WEIGH{\wta}{\pola})(x) \neq \addid \tag{definition of $\supp(G(\WEIGH{\wta}{\pola}))$} \\
		&\implies I(x) \neq \addid \wedge \wta \semimul G(\pola)(x) \neq \addid \tag{definition of $G(\WEIGH{\wta}{\pola})$} \\
		&\implies I(x) \neq \addid \wedge \wta \neq \addid \wedge G(\pola)(x) \neq \addid \tag{\Cref{def:semirings}.\ref{def:semirings4}} \\
		&\implies I(x) \neq \addid \wedge G(\pola)(x) \neq \addid \tag{$A \wedge B \implies A$} \\
		&\implies I(x) \neq \addid \wedge x \in \supp(G(\pola)) \tag{definition of $\supp(G(\pola))$} \\
		&\implies I(x) \neq \addid \wedge x \in J \tag{definition of $J$} \\
		&\implies x \in J' \tag{definition of $J'$}
	\end{align*}
	Therefore $J' = K'$ and so $\semisum{x \in J'} I(x) \eeq \semisum{x \in K'} I(x)$. Finally, by the properties of $\omega$-continuous monoids:
	\[\semisum{x \in J} I(x) \eeq \semisum{x \in J'} I(x) \eeq \semisum{x \in K'} I(x) \eeq \semisum{x \in K} I(x)\]
\end{proof}

\begin{lemma}
\label{lmm:supp-sum}
    For all $\hista, \histb \in \histories$ and $\pola, \polb \in \pols$, the following sums are equal:
\[
    \big(\semisum{x \in \supp(G(\pola))} G(\pola)(x) \semimul \llb x \rrb (\hista)(\histb)\big) \semiadd \big(\semisum{y \in \supp(G(\polb))} G(\polb)(y) \semimul \llb y \rrb (\hista)(\histb)\big)
\]\[
    \semisum{z \in \supp(G(\pola)) \cup \supp(G(\polb))} (G(\pola)(z) \semimul \llb z \rrb (\hista)(\histb)) \semiadd (G(\polb)(z) \semimul \llb z \rrb (\hista)(\histb))
\]
\end{lemma}
\begin{proof}
	Let $I_1, I_2 \colon \GS \rightarrow \semi$ defined by $I_1(g) = G(\pola)(g) \semimul \llb x \rrb (\hista)(\histb)$ and $I_2(g) = G(\polb)(g) \semimul \llb y \rrb (\hista)(\histb)$. Additionally, let $X_1 = \supp(G(\pola))$ and $X_2 = \supp(G(\polb))$. Then:
	\begin{align*}
		g \notin X_1
		&\implies g \notin \supp(G(\pola)) \tag{Definition of $X_1$} \\
		&\implies G(\pola)(g) = \addid \tag{Definition of $\supp$} \\
		&\implies G(\pola)(g) \semimul \llb x \rrb (\hista)(\histb) = \addid \tag{\Cref{def:semirings}.\ref{def:semirings4}} \\
		&\implies I_1(g) = \addid \tag{Definition of $I_1(g)$}
	\end{align*}
	The reasoning to show that $g \notin X_2 \implies I_2(g) = \addid$ is analogous. By properties of $\omega$-continuous monoids, we have:
	\[\semisum{x_1 \in X_1} I_1(x_1) \semiadd \semisum{x_2 \in X_2} I_2(x_2) = \semisum{x \in X_1 \cup X_2} I_1(x) \semiadd I_2(x)\]
	And so, by definition of $X_1, X_2$ and $I_1, I_2$, the claim is proved.
\end{proof}

\begin{lemma}
\label{lmm:supp-union}
    For all $\pola, \polb \in \pols$, $\supp(G(\ADD{\pola}{\polb})) = \supp(G(\pola)) \cup \supp(G(\polb))$.
\end{lemma}
\begin{proof}
    An element is in $\supp(G(\ADD{\pola}{\polb}))$ if and only if it is in $\supp(G(\pola)) \cup \supp(G(\polb))$:
    \begin{align*}
        z \in \supp(G(\ADD{\pola}{\polb}))
        &\iff G(\ADD{\pola}{\polb})(z) \neq \addid \tag{definition of $\supp$} \\
        &\iff G(\pola)(z) \semiadd G(\polb)(z) \neq \addid \tag{definition of $G(\ADD{\pola}{\polb}$)}  \\
        &\iff G(\pola)(z) \neq \addid \vee G(\polb)(z) \neq \addid \tag{$\omega$-cont. monoid is zerosumfree} \\
        &\iff z \in \supp(G(\pola)) \vee z \in \supp(G(\polb)) \tag{definition of $\supp$} \\
        &\iff z \in \supp(G(\pola)) \cup \supp(G(\polb)) \tag{property of set union}
    \end{align*}
    Therefore both sets have exactly the same elements, therefore they are equal.
\end{proof}

\begin{lemma}
	\label{lmm:bind-dist-1}
        For countable sets $X$ and $Y$, if $f \colon X \rightarrow \weightings{\histories}$ and $g \colon Y \rightarrow (\histories \rightarrow \weightings{\histories})$ then:
	\[\big(\semisum{x \in X} f(x)\big) \bind \big(\semisum{y \in Y} g(y)\big) \eeq \semisum{x \in X} \semisum{y \in Y} \big(f(x) \bind g(y)\big)\]
\end{lemma}
\begin{proof}
	\begin{align*}
		\big(\semisum{x \in X} f(x)\big) \bind \big(\semisum{y \in Y} g(y)\big)
		&= \lambda \histb .\; \semisum{\hista \in \supp(\semisum{x \in X} f(x))} \big( \semisum{x \in X} f(x) \big)(\hista) \semimul \big( \semisum{y \in Y} g(y) \big)(\hista)(\histb) \tag{definition of $\bind$} \\
		&= \lambda \histb .\; \semisum{\hista \in \supp(\semisum{x \in X} f(x))} \big( \semisum{x \in X} f(x)(\hista) \big) \semimul \big( \semisum{y \in Y} g(y)(\hista)(\histb) \big) \tag{definition of lifted $\semisum{}$} \\
		&= \lambda \histb .\; \semisum{\hista \in \histories} \big( \semisum{x \in X} f(x)(\hista) \big)  \semimul \big( \semisum{y \in Y} g(y)(\hista)(\histb) \big) \tag{prop. of $\omega$-cont. monoids} \\
		&= \lambda \histb .\; \semisum{\hista \in \histories} \big( \semisum{x \in X} f(x)(\hista) \semimul \big( \semisum{y \in Y} g(y)(\hista)(\histb) \big) \big) \tag{\Cref{def:omega-comp-semirings}} \\
		&= \lambda \histb .\; \semisum{\hista \in \histories} \big( \semisum{x \in X} \big( \semisum{y \in Y} f(x)(\hista) \semimul g(y)(\hista)(\histb) \big) \big) \tag{\Cref{def:omega-comp-semirings}} \\
		&= \lambda \histb .\; \semisum{x \in X} \big( \semisum{\hista \in \histories} \big( \semisum{y \in Y} f(x)(\hista) \semimul g(y)(\hista)(\histb) \big) \big) \tag{assoc. of countable sum} \\
		&= \lambda \histb .\; \semisum{x \in X} \big( \semisum{y \in Y} \big( \semisum{\hista \in \histories} f(x)(\hista) \semimul g(y)(\hista)(\histb) \big) \big) \tag{assoc. of countable sum} \\
		&= \lambda \histb .\; \semisum{x \in X} \big( \semisum{y \in Y} \big( \semisum{\hista \in \supp(f(x))} f(x)(\hista) \semimul g(y)(\hista)(\histb) \big) \big) \tag{prop. of $\omega$-cont. monoids} \\
		&= \semisum{x \in X} \big(\lambda \histb .\; \semisum{y \in Y} \big( \semisum{\hista \in \supp(f(x))} f(x)(\hista) \semimul g(y)(\hista)(\histb) \big) \big) \tag{definition of lifted $\semisum{}$} \\
		&= \semisum{x \in X} \big( \semisum{y \in Y} \big(\lambda \histb .\; \semisum{\hista \in \supp(f(x))} f(x)(\hista) \semimul g(y)(\hista)(\histb) \big) \big) \tag{definition of lifted $\semisum{}$} \\
		&= \semisum{x \in X} \big( \semisum{y \in Y} \big(f(x) \bind g(y) \big) \big) \tag{definition of $\bind$} \\
		&= \semisum{x \in X} \semisum{y \in Y} \big(f(x) \bind g(y) \big) \tag{notational equivalence}
	\end{align*}
\end{proof}

\begin{lemma}
	\label{lmm:bind-dist-2}
        If $f \colon \weightings{\histories}$, $g \colon \histories \rightarrow \weightings{\histories}$, and $\wta \in \semi$ then:
	\[f \bind \big(\lambda \hista. \wta \semimul g(\hista)\big) \eeq (f \semimul \wta) \bind g\]
\end{lemma}
\begin{proof}
	\begin{align*}
		f \bind \big(\lambda \hista. \wta \semimul g(\hista)\big)
		&= \lambda \histb.\; \semisum{\hista \in \supp(f)} f(\hista) \semimul \big(\lambda \hista.\; \wta \semimul g(\hista) \big)(\hista)(\histb) \tag{definition of $\bind$} \\
		&= \lambda \histb.\; \semisum{\hista \in \supp(f)} f(\hista) \semimul \big(\wta \semimul g(\hista) \big)(\histb) \tag{$\lambda$ application} \\
		&= \lambda \histb.\; \semisum{\hista \in \supp(f)} f(\hista) \semimul \big(\wta \semimul g(\hista)(\histb) \big) \tag{definition of lifted $\semimul$} \\
		&= \lambda \histb.\; \semisum{\hista \in \histories} f(\hista) \semimul \big(\wta \semimul g(\hista)(\histb) \big) \tag{prop. of $\omega$-cont. monoids} \\
		&= \lambda \histb.\; \semisum{\hista \in \histories} \big(f(\hista) \semimul \wta \big) \semimul g(\hista)(\histb) \tag{associativity of $\semimul$} \\
		&= \lambda \histb.\; \semisum{\hista \in \supp(f \semimul \wta)} \big(f(\hista) \semimul \wta \big) \semimul g(\hista)(\histb) \tag{prop. of $\omega$-cont. monoids} \\
		&= (f \semimul \wta) \bind g \tag{definition of $\bind$}
	\end{align*}
\end{proof}

\begin{lemma}
	\label{lmm:supp-seq-sub}
	\[\supp(G(\SEQ{\pola}{\polb})) = \{x \diamond y \mid G(\pola)(x) \semimul G(\polb)(y) \neq \addid \wedge x \in \supp(G(\pola)) \wedge y \in \supp(G(\polb))\}\]
\end{lemma}
\begin{proof}
We show that the sets have the same elements, and must then be equal.
\begin{align*}
	z \in \supp(G(\SEQ{\pola}{\polb}))
	&\iff G(\SEQ{\pola}{\polb})(z) \neq \addid \tag{definition of $\supp(G(\SEQ{\pola}{\polb}))$} \\
	&\iff (G(\pola) \diamond G(\polb))(z) \neq \addid \tag{definition of $G(\SEQ{\pola}{\polb})$} \\
	&\iff \semisum{x \in \supp(G(\pola)), y \in \supp(G(\polb)), z = x \diamond y} G(\pola)(x) \semimul G(\polb)(y) \neq \addid \tag{definition of $(G(\pola) \diamond G(\polb))$} \\
	&\iff z = x \diamond y \wedge G(\pola)(x) \semimul G(\polb)(y) \neq \addid \tag{for some $x \in \supp(G(\pola)), y \in \supp(G(\polb))$, by $\omega$-cont. monoids being zerosumfree} \\
	&\iff z \in \{x \diamond y \mid G(\pola)(x) \semimul G(\polb)(y) \neq \addid \wedge x \in \supp(G(\pola)) \wedge y \in \supp(G(\polb))\} \tag{by set condition}
\end{align*}
\end{proof}

\begin{lemma}
\label{lmm:supp-seq}
    If $\pola, \polb \in \pols$, then:
\[
    \semisum{(\gsA, \gsB) \in I} S(\gsA, \gsB) \eeq \semisum{j \in J} \big(\semisum{(\gsA, \gsB) \in K_j} S(\gsA, \gsB)\big)
\]
where:
\begin{align*}
	J &\triangleq \supp(G(\SEQ{\pola}{\polb})) &K_j &\triangleq \{ (x, y) \mid (x, y) \in I \wedge x \diamond y = j\} \\
	I & \triangleq \supp(G(\pola)) \times \supp(G(\polb)) &S(x, y) &\triangleq G(\pola)(x) \semimul G(\polb)(y) \semimul \llb x \diamond y \rrb (\hista)(\histb)
\end{align*}
\end{lemma}
\begin{proof}
	We first consider sums over sets $I' \subseteq I$ and $K' \subseteq K$ of indices leading to non-zero terms:
	\[I' = \{ (x, y) \mid S(x, y) \neq \addid \wedge (x, y) \in I\} \qquad K_j' = \{ (x, y) \mid S(x, y) \neq \addid \wedge (x, y) \in K_j\}\]

	We now show that $I' = \bigcup_{j \in J} K'_j$. Note that $I' \subseteq \bigcup_{j \in J} K'_j$ as:
    \begin{align*}
	    (x, y) \in I'
	&\implies S(x, y) \neq \addid \wedge (x, y) \in I \tag{definition of $I'$} \\
	&\implies S(x, y) \neq \addid \wedge S(x, y) \neq \addid \wedge (x, y) \in I \tag{$A \implies A \wedge A$} \\
	&\implies S(x, y) \neq \addid \wedge G(\pola)(x) \semimul G(\polb)(y) \semimul \llb x \diamond y \rrb (\hista)(\histb) \neq \addid \wedge (x, y) \in I \tag{definition of $S(x, y)$} \\
	&\implies S(x, y) \neq \addid \wedge G(\pola)(x) \semimul G(\polb)(y) \neq \addid \wedge \llb x \diamond y \rrb (\hista)(\histb) = \mulid \wedge (x, y) \in I \tag{$\llb x \diamond y \rrb (\hista)(\histb) \in \{\addid, \mulid \}$} \\
	&\implies S(x, y) \neq \addid \wedge G(\pola)(x) \semimul G(\polb)(y) \neq \addid \wedge (x, y) \in K_{(x \diamond y)} \tag{by definition of $K_{(x \diamond y)}$, as $x \diamond y$ is defined} \\
	&\implies G(\pola)(x) \semimul G(\polb)(y) \neq \addid \wedge (x, y) \in K'_{(x \diamond y)} \tag{definition of $K'_{(x \diamond y)}$} \\
	&\implies (x \diamond y) \in \supp(G(\SEQ{\pola}{\polb})) \wedge (x, y) \in K'_{(x \diamond y)} \tag{\Cref{lmm:supp-seq-sub}} \\
	&\implies (x \diamond y) \in J \wedge (x, y) \in K'_{(x \diamond y)} \tag{definition of $J$} \\
	&\implies (x, y) \in \bigcup_{j \in J} K'_j \tag{when $j = x \diamond y$}
    \end{align*}

    Additionally, $\bigcup_{j \in J} K'_j \subseteq I'$ as:
    \begin{align*}
	    (x, y) \in \bigcup_{j \in J} K'_j
	    &\implies (x, y) \in K'_j \wedge j \in J \tag{for some $j \in J$} \\
	    &\implies (x, y) \in K'_j \tag{$A \wedge B \implies A$} \\
	    &\implies S(x, y) \neq \addid \wedge (x, y) \in K_j \tag{definition of $K'_j$} \\
	    &\implies S(x, y) \neq \addid \wedge (x, y) \in I \tag{definition of $K_j$} \\
	    &\implies (x, y) \in I' \tag{definition of $I'$}
    \end{align*}

    Therefore, $I' = \bigcup_{j \in J} K'_j$. Furthermore, note that for any $i, j \in J$, if $i \neq j$, then:
    \begin{align*}
	K_i \cap K_j
	&= \{ (x, y) \mid (x, y) \in I \wedge x \diamond y = i\} \cap \{ (x, y) \mid (x, y) \in I \wedge x \diamond y = j\} \tag{definition of $K_i$, $K_j$} \\
	&= \{ (x, y) \mid (x, y) \in I \wedge x \diamond y = i \wedge x \diamond y = j\} \tag{$\wedge$ of set conditions} \\
	&= \{ (x, y) \mid (x, y) \in I \wedge i = x \diamond y = j\} \tag{transitivity of equality} \\
	&= \varnothing \tag{as $i \neq j$}
    \end{align*}

    And as $K'_i \subseteq K_i$, then $I' = \bigcup_{j \in J} K_j' = \biguplus_{j \in J} K_j'$. Thus, by properties of $\omega$-continuous monoids, we have:
    \[\semisum{(x, y) \in I} S(x, y) \eeq \semisum{(x, y) \in I'} S(x, y) \eeq \semisum{j \in J} \big(\semisum{(x, y) \in K_j'} S(x, y) \big) \eeq \semisum{j \in J} \big(\semisum{(x, y) \in K_j} S(x, y)\big) \]
\end{proof}

\subsection{Proof of Denotational--Language Correspondence}
\label{pf:deno-lang}
\repeatlemma{thm:deno-lang}
\begin{proof}
By induction on the structure of $e \in \pols$.

Base cases:
\begin{align*}
    \sem{\ctest{\pkt}} (\hista)
    &= \sem{ \SEQ{\ctest{\pkt}}{\cass{\pkt}} } (\hista) \tag{\Cref{lmm:red-props}.\ref{lmm:red-props1}} \\
    &= \mulid \semimul \sem{\SEQ{\ctest{\pkt}}{\cass{\pkt}}} (\hista) \tag{identity of $\semimul$} \\
    &= G(\ctest{\pkt})(\SEQ{\ctest{\pkt}}{\cass{\pkt}}) \semimul \sem{ \SEQ{\ctest{\pkt}}{\cass{\pkt}} } (\hista) \tag{$G(\ctest{\pkt})(\SEQ{\ctest{\pkt}}{\cass{\pkt}}) = \mulid$} \\
    &= \semisum{\gsA \in \supp(G(\ctest{\pkt}))} G(\ctest{\pkt})(\gsA) \semimul \sem{ \gsA } (\hista) \tag{$\supp(G(\ctest{\pkt})) = \{\SEQ{\ctest{\pkt}}{\cass{\pkt}}\}$}
\end{align*}
\begin{align*}
    \sem{ \cass{\pkt} } (\hista)
    &= \sem{ \SUM{\ctest{\pkta} \in \Test} \SEQ{\ctest{\pkta}}{\cass{\pkt}} } (\hista) \tag{\Cref{lmm:red-props}.\ref{lmm:red-props4}} \\
    &= \semisum{\ctest{\pkta} \in \Test} \sem{ \SEQ{\ctest{\pkta}}{\cass{\pkt}} } (\hista) \tag{definition of $\sem{ \ADD{\pola}{\polb} }$} \\
    &= \semisum{\ctest{\pkta} \in \Test} \mulid \semimul \sem{ \SEQ{\ctest{\pkta}}{\cass{\pkt}} } (\hista) \tag{identity of $\semimul$} \\
    &= \semisum{\ctest{\pkta} \in \Test} G(\cass{\pkt})(\SEQ{\ctest{\pkta}}{\cass{\pkt}}) \semimul \sem{ \SEQ{\ctest{\pkta}}{\cass{\pkt}} } (\hista) \tag{$G(\cass{\pkt})(\SEQ{\ctest{\pkta}}{\cass{\pkt}}) \triangleq \mulid$} \\
    &= \semisum{\gsA \in \supp(G(\cass{\pkt}))} G(\cass{\pkt})(\gsA) \semimul \sem{ \gsA } (\hista) \tag{$\supp(G(\cass{\pkt})) = \{\SEQ{\ctest{\pkta}}{\cass{\pkt}} \mid \ctest{\pkta} \in \Test\}$}
\end{align*}
\begin{align*}
    \sem{ \DUP } (\hista)
    &= \sem{ \SUM{\ctest{\pkt} \in \Test} \SEQ{\SEQ{\SEQ{\ctest{\pkt}}{\cass{\pkt}}}{\DUP}}{\cass{\pkt}} } (\hista) \tag{\Cref{lmm:red-props}.\ref{lmm:red-props5}} \\
    &= \semisum{\ctest{\pkt} \in \Test} \sem{ \SEQ{\SEQ{\SEQ{\ctest{\pkt}}{\cass{\pkt}}}{\DUP}}{\cass{\pkt}} } (\hista) \tag{definition of $\sem{ \ADD{\pola}{\polb} }$} \\
    &= \semisum{\ctest{\pkt} \in \Test} \mulid \semimul \sem{  \SEQ{\SEQ{\SEQ{\ctest{\pkt}}{\cass{\pkt}}}{\DUP}}{\cass{\pkt}} } (\hista) \tag{identity of $\semimul$} \\
    &= \semisum{\ctest{\pkt} \in \Test} G(\DUP)(\SEQ{\SEQ{\SEQ{\ctest{\pkt}}{\cass{\pkt}}}{\DUP}}{\cass{\pkt}}) \semimul \sem{ \SEQ{\SEQ{\SEQ{\ctest{\pkt}}{\cass{\pkt}}}{\DUP}}{\cass{\pkt}} } (\hista) \tag{$G(\DUP)(\SEQ{\SEQ{\SEQ{\ctest{\pkt}}{\cass{\pkt}}}{\DUP}}{\cass{\pkt}}) \triangleq \mulid$}\\
    &= \semisum{\gsA \in \supp(G(\DUP))} G(\DUP)(\gsA) \semimul \sem{ \gsA } (\hista) \tag{$\supp(G(\DUP)) = \{\SEQ{\SEQ{\SEQ{\ctest{\pkt}}{\cass{\pkt}}}{\DUP}}{\cass{\pkt}} \mid \ctest{\pkt} \in \Test\}$}
\end{align*}

Inductive steps:
\begin{align*}
    \llb \WEIGH{\wta}{\pola} \rrb (\hista)
    &= \wta \semimul \llb \pola \rrb (\hista) \tag{definition of $\llb \WEIGH{\wta}{\pola} \rrb$} \\
    &= \wta \semimul \semisum{\gsA \in \supp(G(\pola))} G(\pola)(\gsA) \semimul \llb \gsA \rrb (\hista) \tag{IH on $\pola$} \\
    &= \semisum{\gsA \in \supp(G(\pola))} \wta \semimul \big(G(\pola)(\gsA) \semimul \llb \gsA \rrb (\hista)\big) \tag{\Cref{def:omega-comp-semirings}} \\
    &= \semisum{\gsA \in \supp(G(\pola))} \big(\wta \semimul G(\pola)(\gsA)\big) \semimul \llb x \rrb (\hista) \tag{associativity of $\semimul$} \\
    &= \semisum{\gsA \in \supp(G(\pola))} G(\WEIGH{\wta}{\pola})(\gsA) \semimul \llb \gsA \rrb (\hista) \tag{definition of $G(\WEIGH{\wta}{\pola})$} \\
    &= \semisum{\gsA \in \supp(G(\WEIGH{\wta}{\pola}))} G(\WEIGH{\wta}{\pola})(\gsA) \semimul \llb \gsA \rrb (\hista) \tag{\Cref{lmm:supp-weigh}}
\end{align*}
\begin{align*}
    \llb \ADD{\pola}{\polb} \rrb (\hista)
    &= \llb \pola \rrb (\hista) \semiadd \llb \polb \rrb (\hista) \tag{definition of $\llb \ADD{\pola}{\polb} \rrb$} \\
    &= \big(\semisum{gsA \in \supp(G(\pola))} G(\pola)(\gsA) \semimul \llb \gsA \rrb (\hista)\big) \semiadd \big(\semisum{\gsB \in \supp(G(\polb))} G(\polb)(\gsB) \semimul \llb \gsB \rrb (\hista)\big) \tag{IH on $\pola$ and $\polb$} \\
    &= \semisum{\gsC \in \supp(G(\pola)) \cup \supp(G(\polb))} \big(G(\pola)(\gsC) \semimul \llb \gsC \rrb (\hista)\big) \semiadd \big(G(\polb)(\gsC) \semimul \llb \gsC \rrb (\hista)\big) \tag{\Cref{lmm:supp-sum}} \\
    &= \semisum{\gsC \in \supp(G(\pola)) \cup \supp(G(\polb))} \big(G(\pola)(\gsC) \semiadd G(\polb)(\gsC)\big) \semimul \llb \gsC \rrb (\hista) \tag{dist. of $\semimul$ over $\semiadd$} \\
    &= \semisum{\gsC \in \supp(G(\ADD{\pola}{\polb}))} \big(G(\pola)(\gsC) \semiadd G(\polb)(\gsC)\big) \semimul \llb \gsC \rrb (\hista) \tag{\Cref{lmm:supp-union}} \\
    &= \semisum{\gsC \in \supp(G(\ADD{\pola}{\polb}))} G(\ADD{\pola}{\polb})(\gsC) \semimul \llb \gsC \rrb (\hista) \tag{definition of $G(\ADD{\pola}{\polb})$}
\end{align*}
\begin{align*}
    \llb \SEQ{\pola}{\polb} \rrb (\hista)
    &= \llb \pola \rrb (\hista) \bind \llb \polb \rrb \tag{definition of $\llb\SEQ{\pola}{\polb}\rrb$} \\
    &= \big(\semisum{\gsA \in \supp(G(\pola))} G(\pola)(\gsA) \semimul \llb \gsA \rrb (\hista) \big) \bind \big(\lambda \histb. \semisum{\gsB \in \supp(G(\polb))} G(\polb)(\gsB) \semimul \llb \gsB \rrb (\histb) \big) \tag{IH on $\pola$ and $\polb$} \\
    &= \semisum{\gsA \in \supp(G(\pola))} \semisum{\gsB \in \supp(G(\polb))} G(\pola)(\gsA) \semimul \llb \gsA \rrb (\hista) \bind \lambda \histb. G(\polb)(\gsB) \semimul \llb \gsB \rrb (\histb) \tag{\Cref{lmm:bind-dist-1}}\\
    &= \semisum{\gsA \in \supp(G(\pola))} \semisum{\gsB \in \supp(G(\polb))} G(\pola)(\gsA) \semimul \llb \gsA \rrb (\hista) \semimul G(\polb)(\gsB) \bind \lambda \histb. \llb \gsB \rrb (\histb) \tag{\Cref{lmm:bind-dist-2}}\\
    &= \semisum{\gsA \in \supp(G(\pola))} \semisum{\gsB \in \supp(G(\polb))} G(\pola)(\gsA) \semimul G(\polb)(\gsB) \semimul \llb \gsA \rrb (\hista) \bind \lambda \histb. \llb \gsB \rrb (\histb) \tag{$\llb \gsA \rrb (\hista)(\histb) \in \{ \addid, \mulid\}$}\\
    &= \semisum{\gsA \in \supp(G(\pola))} \semisum{\gsB \in \supp(G(\polb))} G(\pola)(\gsA) \semimul G(\polb)(\gsB) \semimul \llb \SEQ{\gsA}{\gsB} \rrb (\hista) \tag{definition of $\llb \SEQ{\gsA}{\gsB} \rrb$}\\
    &= \semisum{\gsA \in \supp(G(\pola))} \semisum{\gsB \in \supp(G(\polb))} G(\pola)(\gsA) \semimul G(\polb)(\gsB) \semimul \llb \gsA \diamond \gsB \rrb (\hista) \tag{sequecing of guarded strings}\\
    &= \semisum{\gsC \in \supp(G(\SEQ{\pola}{\polb}))} \big(\semisum{\gsA \in \supp(G(\pola)), \gsB \in \supp(G(\polb)), \gsC = \gsA \diamond \gsB} G(\pola)(\gsA) \semimul G(\polb)(\gsB) \semimul \llb \gsA \diamond \gsB \rrb (\hista)\big) \tag{\Cref{lmm:supp-seq}} \\
    &= \semisum{\gsC \in \supp(G(\SEQ{\pola}{\polb}))} \big(\semisum{\gsA \in \supp(G(\pola)), \gsB \in \supp(G(\polb)), \gsC = \gsA \diamond \gsB} G(\pola)(\gsA) \semimul G(\polb)(\gsB) \semimul \llb \gsC \rrb (\hista)\big) \tag{$\gsC = \gsA \diamond \gsB$} \\
    &= \semisum{\gsC \in \supp(G(\SEQ{\pola}{\polb}))} \big(\semisum{\gsA \in \supp(G(\pola)), \gsB \in \supp(G(\polb)), \gsC = \gsA \diamond \gsB} G(\pola)(\gsA) \semimul G(\polb)(\gsB)\big) \semimul \llb \gsC \rrb (\hista) \tag{\Cref{def:omega-comp-semirings}} \\
    &= \semisum{\gsC \in \supp(G(\SEQ{\pola}{\polb}))} (G(\pola) \diamond G(\polb))(\gsC) \semimul \llb \gsC \rrb (\hista) \tag{Definition of $G(\pola) \diamond G(\polb)$} \\
    &= \semisum{\gsC \in \supp(G(\SEQ{\pola}{\polb}))} G(\SEQ{\pola}{\polb})(\gsC) \semimul \llb \gsC \rrb (\hista) \tag{definition of $G(\SEQ{\pola}{\polb})$}
\end{align*}
\begin{align*}
    \llb \ITER{\pola} \rrb (\hista)
    &= \semisum{n \in \nats} \big( \sem{ \NFOLD{\pola} } (\hista) \big) \tag{definition of $\sem{\ITER{\pola}}$} \\
    &= \semisum{n \in \nats} \big( \semisum{\gsA \in \supp(G(\NFOLD{\pola}))} G(\NFOLD{\pola})(\gsA) \semimul \sem{\gsA}(\hista) \big) \tag{IH on $\NFOLD{\pola}$} \\
    &= \semisum{n \in \nats} \big( \semisum{\gsA \in \GS} G(\NFOLD{\pola})(\gsA) \semimul \sem{\gsA}(\hista) \big) \tag{prop. of $\omega$-cont. monoids}\\
    &= \semisum{\gsA \in \GS}  \big( \semisum{n \in \nats} G(\NFOLD{\pola})(\gsA) \semimul \sem{\gsA}(\hista) \big) \tag{assoc. of countable sum} \\
    &= \semisum{\gsA \in \GS}  \big( \semisum{n \in \nats} G(\NFOLD{\pola})(\gsA) \big) \semimul \sem{\gsA}(\hista) \tag{\Cref{def:omega-comp-semirings}} \\
    &= \semisum{\gsA \in \GS} G(\ITER{\pola})(\gsA) \semimul \sem{\gsA}(\hista) \tag{definition of $G(\ITER{\pola})$} \\
    &= \semisum{\gsA \in \supp(G(\ITER{\pola}))} G(\ITER{\pola})(\gsA) \semimul \sem{\gsA}(\hista) \tag{prop. of $\omega$-cont. monoid}
\end{align*}
\end{proof}

\section{\wnetkat Automata}
\subsection{Proof of Soundness of Thompson}
\label{pf:thompson-soundness}
\repeatlemma{thm:thompson-soundness}
\begin{proof}
    By induction on $\pola$. A full mechanization of our proof is available at \url{https://github.com/cornell-pl/wnetkat-lean/blob/pldi2026/WeightedNetKAT/Papers/PLDI2026.lean}.

\end{proof}

\section{Decidability Results for \wnetkat}
\label{appendix:decidability}

This appendix section includes the proofs of the main decidability results in
\Cref{sec:decidability}. We include a few additional definitions omitted in the
main body of the paper, beginning with the decision problems at the level of
\wnetkat automata that we show to be decidable.

\begin{definition}[Decision problems for \wnetkat automata]
\label{def:dec-problems}
    Let $\wnka$ be a \wnetkat automaton over a computable semiring. We define
    the following decision problems:
    \begin{enumerate}
        \item\label{def:dec-problems1} We say that $\wnka$ is \emph{$\wta$-safe}
            (denoted $\wnka \in \autsafe{\wta}$) iff \
            $\forall \gsA \in \GS \colon \sem{\wnka}(\gsA) \semiord \wta$
        \item\label{def:dec-problems2} We say that $\wnka$ is \emph{$\wta$-reachable}
            (denoted $\wnka \in \autreach{\wta}$) iff \
            $\exists \gsA \in \GS \colon \sem{\wnka}(\gsA) \semiordrev \wta$
    \end{enumerate}
\end{definition}

\begin{definition}[Extended family of transition functions and output weights]
\label{def:extended-fam}
    Given a \wnetkat automaton $\wnka$:
    \begin{enumerate}
        \item\label{def:extended-trans}
            We extend the family of transition functions $\wnkatrans{\pkta}{\pktb}$ to
            define an \emph{extended family of transition functions}
            $\{\wnkatransext{\gsA}\}_{\gsA \in \GS} \colon \GS \to \weightings{\wnkastates}$:
            \[
                \wnkatransext{\oldctestA\oldctestB} = \idmatrix_{\wnkastates\times\wnkastates}
                \qquad\text{and}\qquad
                \wnkatransext{\oldctestA \SEQN \oldctestB \SEQN \gsB} =
                    \wnkatrans{\oldctestA}{\oldctestB} \times \wnkatransext{\oldctestB \SEQN \gsB}~.
            \]
        \item\label{def:extended-out}
            We extend the family of output weights $\wnkaout{\oldctestA}{\oldctestB}$
            to an \emph{extended family of output weightings}
            $\{\wnkaoutext{\gsA}\}_{\gsA \in \GS} \colon \GS \to \weightings{\GS}$:
            \[
            \wnkaoutext{\oldctestA\oldctestB} \triangleq \wnkaout{\oldctestA}{\oldctestB} \qquad\qquad
            \wnkaoutext{\oldctestA\oldctestB\,\DUP\,\gsA} \triangleq \wnkatrans{\oldctestA}{\oldctestB} \times \wnkaoutext{\oldctestB\gsA}~.
            \]
    \end{enumerate}
\end{definition}

Notice that, $\wnkatransext{}$ and $\wnkaoutext{}$ are merely inductive definitions
corresponding to the weighted language of guarded strings recognized by $\wnka$, i.e.\ :
\begin{align*}
    \sem{\wnka}(\gsA = \oldctestA_0 \SEQN \oldctestA_1 \SEQN \DUP \SEQN \cdots \SEQN \oldctestA_n) &
        \triangleq \wnkainit \times \wnkatrans{\oldctestA_0}{\oldctestA_1} \times \cdots \times \wnkatrans{\oldctestA_{n-2}}{\oldctestA_{n-1}} \times \wnkaout{\oldctestA_{n-1}}{\oldctestA_n}
            &=& \ \wnkainit \times \wnkatransext{\gsA} \times \wnkaout{\oldctestA_{n-1}}{\oldctestA_{n}} \\
            &&=& \ \wnkainit \times \wnkaoutext{\gsA}
\end{align*}

We will use the extended versions of these functions to show that we can define alternate
(but equivalent) definitions for the weighted language of guarded strings recognized by a
\wnetkat automaton (\Cref{appendix:unf-sem,appendix:runs-wnka}). These alternate
definitions will be useful for proving the decidability of both $\semithresh$-safety
and $\semithresh$-reachability.

\subsection{Packet Configuration Semantics of \wnetkat Automata}
\label{appendix:unf-sem}
In this section, we give an alternate formulation of the weighted language of
guarded strings recognized by \wnetkat automaton. This alternate definition will
be useful for proving the decidability of $\semithresh$-safety. In particular, we
make the carry-on packet in the transitions of the \wnetkat automaton explicit by
internalizing it in the states of the automaton.

\begin{definition}[Packet Configuration Semantics]
    Given a \wnetkat automaton $\wnka$, we define its \emph{packet configuration semantics} $\altsem{\wnka} \colon \GS \rightarrow \semi$ by:
    \[
	\altsem{\wnka}(\oldctestA_0 \SEQN \oldctestA_1 \SEQN \DUP \SEQN \cdots \SEQN \oldctestA_n) \triangleq \wfainit \times \wfatrans{\oldctestA_0} \times \wfatrans{\oldctestA_1} \times \cdots \times \wfatrans{\oldctestA_n} \times \wfaout~,
    \]
    	where $\wfainit \colon \weightings{\wfastates}, \wfatrans{\oldctestA} \colon \weightings{\wfastates \times \wfastates}, \wfaout \colon \weightings{\wfastates}$ for $\wfastates \triangleq (\wnkastates \times \packets) + \{\wfastartstate, \wfaendstate\}$ are defined by:
	\begin{align*}
        \wfainit &\triangleq \unit(\wfastartstate)	&\quad \wfaout &\triangleq \unit(\wfaendstate)	&\quad    \wfatrans{\oldctestB}(\wfastatea, \oldctestA)(\wfastateb, \oldctestB) &\triangleq \wnkatrans{\oldctestA}{\oldctestB}(\wfastatea)(\wfastateb) \\
	    \wfatrans{\oldctestB}(\wfastartstate)(\wfastatea, \oldctestB) &\triangleq \wnkainit(\wfastatea) &\quad \wfatrans{\oldctestB}(\wfastatea, \oldctestA)(\wfaendstate) &\triangleq \wnkaout{\oldctestA}{\oldctestB}(\wfastatea) &\quad \wfatrans{\oldctestB}(-)(-) &\triangleq \addid
	\end{align*}
\end{definition}
As for the standard definition of $\wnetkat$ automata, we additionally define
an extended version of the family of alternate transition functions
$\{\wfatransext{\gsA}\}_{\gsA \in (\packets \cup \{\DUP\})^*} \colon (\packets \cup \{\DUP\})^* \to \weightings{\wfastates \times \wfastates}$:
\[
\wfatransext{\varepsilon} \triangleq \idmatrix_{\wfastates \times \wfastates} \qquad\qquad \wfatransext{\DUP\,\gsA} \triangleq \wfatransext{\gsA} \qquad\qquad \wfatransext{\oldctestA\gsA} \triangleq \wfatrans{\oldctestA} \times \wfatransext{\gsA}~.
\]

Again, $\wfatransext{\gsA}$ is merely an inductive definition corresponding to the
packet configuration semantics of a \wnetkat automaton, i.e.\ :
\begin{align*}
\altsem{\wnka}(\gsA = \oldctestA_0 \SEQN \oldctestA_1 \SEQN \DUP \SEQN \cdots \SEQN \oldctestA_n) &\triangleq \wfainit \times \wfatrans{\oldctestA_0} \times \wfatrans{\oldctestA_1} \times \cdots \times \wfatrans{\oldctestA_n} \times \wfaout &&= \wfainit \times \wfatransext{\gsA} \times \wfaout
\end{align*}

\begin{lemma}[Equivalence of packet configuration semantics of \wnetkat automata]
    \label{lmm:unf-sem}
    For all $\gsA \in \GS$:
    \[
        \sem{\wnka}(\gsA) = \altsem{\wnka}(\gsA)~.
    \]
\end{lemma}
\begin{proof}
    We now show by induction that for any $\wnkastatea \in \wnkastates$, and $\oldctestA\gsA\oldctestB \in \GS$, $\wnkaoutext{\oldctestA\gsA\oldctestB}(\wnkastatea) = \wfatransext{\gsA\oldctestB}(\wnkastatea, \oldctestA)(\wfaendstate)$.

    For the base case ($\gsA = \eword$) we have that:
    \begin{align*}
	\wfatransext{\oldctestB}(\wnkastatea, \oldctestA)(\wfaendstate)
	&= \wfatrans{\oldctestB}(\wnkastatea, \oldctestA)(\wfaendstate) \tag{by definition of $\wfatransext{}$} \\
	&= \wnkaout{\oldctestA}{\oldctestB}(\wnkastatea) \tag{by definition of $\wfatrans{}$} \\
	&= \wnkaoutext{\oldctestA\oldctestB}(\wfastatea) \tag{by definition of $\wnkaoutext{}$}
    \end{align*}

    For the inductive case ($\gsA = \oldctestB\,\DUP\,\gsB$) we have that:
    \begin{align*}
	\wfatransext{\oldctestB\,\DUP\,\gsB\oldctestC}(\wnkastatea, \oldctestA)(\wfaendstate)
	&= (\wfatrans{\oldctestB} \times \wfatransext{\gsB\oldctestC})(\wnkastatea, \oldctestA)(\wfaendstate) \tag{by definition of $\wfatransext{}$} \\
	&= \semisum{\wfastateb \in \wfastates} \wfatrans{\oldctestB}(\wnkastatea, \oldctestA)(\wfastateb) \times \wfatransext{\gsB\oldctestC}(\wfastateb)(\wfaendstate) \tag{applied matrix multiplication}
    \end{align*}

    Now observe that although this sum is over $\wfastateb \in \wfastates$, only a very restrictive subset of $\wfastates$ gives rise to non-zero terms. In particular, by definition of $\wfatrans{\oldctestB}$, $\wfastateb$ can be at most of the form $(\wnkastatec, \oldctestB)$ (for the fixed $\oldctestB$), or $\wfaendstate$. Furhermore, by definition of $\wfatransext{\gsB\oldctestC}$, $\wfastateb$ cannot be the state $\wfaendstate$, as otherwise $\wfatransext{\gsB\oldctestC}$ would unavoidably multiply $\wfatrans{\zeta}(\wfaendstate)(-)$, for some packet $\zeta$, which by definition would also lead to a zero term. As such this sum can be rewritten to:
    \begin{align*}
	\semisum{\wfastateb \in \wfastates} \wfatrans{\oldctestB}(\wnkastatea, \oldctestA)(\wfastateb) \times \wfatransext{\gsB\oldctestC}(\wfastateb)(\wfaendstate)
	&= \semisum{\wnkastateb \in \wnkastates} \wfatrans{\oldctestB}(\wnkastatea, \oldctestA)(\wnkastateb, \oldctestB) \times \wfatransext{\gsB\oldctestC}(\wnkastateb, \oldctestB)(\wfaendstate) \tag{by above} \\
	&= \semisum{\wnkastateb \in \wnkastates} \wnkatrans{\oldctestA}{\oldctestB}(\wnkastatea)(\wnkastateb) \times \wfatransext{\gsB\oldctestC}(\wnkastateb, \oldctestB)(\wfaendstate) \tag{by definition of $\wfatrans{}$} \\
	&= \semisum{\wnkastateb \in \wnkastates} \wnkatrans{\oldctestA}{\oldctestB}(\wnkastatea)(\wnkastateb) \times \wnkaoutext{\oldctestB\gsB\oldctestC}(\wnkastateb) \tag{by the inductive hypothesis} \\
	&= (\wnkatrans{\oldctestA}{\oldctestB} \times \wnkaoutext{\oldctestB\gsB\oldctestC})(\wnkastatea) \tag{applied matrix multiplication} \\
	&= \wnkaoutext{\oldctestA\oldctestB\,\DUP\,\gsB\oldctestC}(\wnkastatea) \tag{by definition of $\wnkaoutext{}$}
    \end{align*}

    We can then finally conclude:
    \begin{align*}
	\altsem{\wnka}(\oldctestA \SEQN \gsA \SEQN \oldctestB)
	&= \wfainit \times \wfatransext{\oldctestA\gsA\oldctestB} \times \wfaout \tag{by above} \\
	&= \wfainit \times \wfatrans{\oldctestA} \times \wfatransext{\gsA\oldctestB} \times \wfaout \tag{by definition of $\wfatransext{}$} \\
	&= (\wfainit \times \wfatrans{\oldctestA}) \times \wfatransext{\gsA\oldctestB} \times \wfaout \tag{associativity of matrix multiplication} \\
	&= \semisum{\wfastatea \in \wfastates} (\wfainit(\wfastatea) \semimul \wfatrans{\oldctestA}(\wfastatea)) \times \wfatransext{\gsA\oldctestB} \times \wfaout \tag{applied matrix multiplication} \\
	&= \wfainit(\wfastartstate) \semimul \wfatrans{\oldctestA}(\wfastartstate) \times \wfatransext{\gsA\oldctestB} \times \wfaout \tag{by definition of $\wfainit$} \\
	&= \wfatrans{\oldctestA}(\wfastartstate) \times \wfatransext{\gsA\oldctestB} \times \wfaout \tag{multiplicative identity} \\
	&= \big(\semisum{\wfastatea \in \wfastates} \wfatrans{\oldctestA}(\wfastartstate)(\wfastatea) \semimul \wfatransext{\gsA\oldctestB}(\wfastatea)\big) \times \wfaout \tag{applied matrix multiplication} \\
	&= \big(\semisum{\wnkastatea \in \wnkastates} \wfatrans{\oldctestA}(\wfastartstate)(\wnkastatea, \oldctestA) \semimul \wfatransext{\gsA\oldctestB}(\wnkastatea, \oldctestA)\big) \times \wfaout \tag{by definition of $\wfatrans{}$} \\
	&= \big(\semisum{\wnkastatea \in \wnkastates} \wnkainit(\wnkastatea) \semimul \wfatransext{\gsA\oldctestB}(\wnkastatea, \oldctestA) \big) \times \wfaout \tag{by definition of $\wfatrans{}$} \\
	&= \semisum{\wnkastatea \in \wnkastates} \wnkainit(\wnkastatea) \semimul \big( \wfatransext{\gsA\oldctestB}(\wnkastatea, \oldctestA) \times \wfaout \big) \tag{distributivity and associativity} \\
	&= \semisum{\wnkastatea \in \wnkastates} \wnkainit(\wnkastatea) \semimul \big( \semisum{\wfastateb \in \wfastates} \wfatransext{\gsA\oldctestB}(\wnkastatea, \oldctestA)(\wfastateb) \times \wfaout(\wfastateb) \big) \tag{applied matrix multiplication} \\
	&= \semisum{\wnkastatea \in \wnkastates} \wnkainit(\wnkastatea) \semimul \big(\wfatransext{\gsA\oldctestB}(\wnkastatea, \oldctestA)(\wfaendstate) \times \wfaout(\wfaendstate) \big) \tag{by definition of $\wfaout$} \\
	&= \semisum{\wnkastatea \in \wnkastates} \wnkainit(\wnkastatea) \semimul \wfatransext{\gsA\oldctestB}(\wnkastatea, \oldctestA)(\wfaendstate) \tag{multiplicative identity} \\
	&= \semisum{\wnkastatea \in \wnkastates} \wnkainit(\wnkastatea) \semimul \wnkaoutext{\oldctestA\gsA\oldctestB}(\wnkastatea) \tag{by $\wfatransext{} / \wnkaoutext{}$ conversion} \\
	&= \wnkainit \times \wnkaoutext{\oldctestA\gsA\oldctestB} \tag{applied matrix multiplication} \\
	&= \sem{\wnka}(\oldctestA \SEQN \gsA \SEQN \oldctestB) \tag{by above}
    \end{align*}
\end{proof}

\begin{lemma}
    \label{lmm:unf-gs-pk}
    \[\semisum{\gsA \in \GS} \wfainit \times \wfatransext{\gsA} \times \wfaout \eeq \semisum{\gsA \in \packets^*} \wfainit \times \wfatransext{\gsA} \times \wfaout\]
\end{lemma}
\begin{proof}
    Firstly, for convenience we show that $\wfainit \times \wfatransext{\gsA} \times \wfaout = \wfatransext{\gsA}(\wfastartstate)(\wfaendstate)$:
    \begin{align*}
	\wfainit \times \wfatransext{\gsA} \times \wfaout
	&= (\wfainit \times \wfatransext{\gsA}) \times \wfaout \tag{matrix multiplication associativity} \\
	&= \big( \semisum{\wfastatea \in \wfastates} \wfainit(\wfastatea) \semimul \wfatransext{\gsA}(\wfastatea) \big) \times \wfaout \tag{applied matrix multiplication} \\
	&= \big( \wfainit(\wfastartstate) \semimul \wfatransext{\gsA}(\wfastartstate) \big) \times \wfaout \tag{by definition of $\wfainit$} \\
	&= \wfatransext{\gsA}(\wfastartstate) \times \wfaout \tag{multiplicative identity} \\
	&= \semisum{\wfastatea \in \wfastates} \wfatransext{\gsA}(\wfastartstate)(\wfastatea) \semimul \wfaout(\wfastatea) \tag{applied matrix multiplication} \\
	&= \wfatransext{\gsA}(\wfastartstate)(\wfaendstate) \semimul \wfaout(\wfaendstate) \tag{by definition of $\wfaout$} \\
	&= \wfatransext{\gsA}(\wfastartstate)(\wfaendstate) \tag{multiplicative identity}
    \end{align*}

    Now we can show that:
    \begin{align*}
	\semisum{\gsA \in \packets^*} \wfainit \times \wfatransext{\gsA} \times \wfaout
	&= \semisum{\gsA \in \packets^*} \wfatransext{\gsA}(\wfastartstate)(\wfaendstate) \tag{by above} \\
	&= \semisum{\gsA \in \packets^0 \, \dot{\cup} \, \packets^1 \, \dot{\cup} \, \packets^{>=2}} \wfatransext{\gsA}(\wfastartstate)(\wfaendstate) \tag{$\packets^* \eeq \packets^0 \, \dot{\cup} \, \packets^1 \, \dot{\cup} \, \packets^{>=2}$} \\
	&= \semisum{\gsA \in \packets^0} \wfatransext{\gsA}(\wfastartstate)(\wfaendstate) \semiadd \semisum{\gsA \in \packets^1} \wfatransext{\gsA}(\wfastartstate)(\wfaendstate) \semiadd \semisum{\gsA \in \packets^{>=2}} \wfatransext{\gsA}(\wfastartstate)(\wfaendstate) \tag{disjoint union sum} \\
	&= \wfatransext{\eword}(\wfastartstate)(\wfaendstate) \semiadd \semisum{\oldctestA \in \packets} \wfatransext{\oldctestA}(\wfastartstate)(\wfaendstate) \semiadd \semisum{\gsA \in \packets^{>=2}} \wfatransext{\gsA}(\wfastartstate)(\wfaendstate) \tag{$\packets^0 = \{\varepsilon\}, \packets^1 = \packets$} \\
	&= \addid \semiadd \addid \semiadd \semisum{\gsA \in \packets^{>=2}} \wfatransext{\gsA}(\wfastartstate)(\wfaendstate) \tag{$\wfatransext{\eword}(\wfastartstate)(\wfaendstate) = \wfatransext{\oldctestA}(\wfastartstate)(\wfaendstate) = \addid$} \\
	&= \semisum{\gsA \in \packets^{>=2}} \wfatransext{\gsA}(\wfastartstate)(\wfaendstate) \tag{additive identity}
    \end{align*}

    Finally, note that $\GS \cong \packets^{>=2}$ with regards to $\wfatransext{\gsA}(\wfastartstate)(\wfaendstate)$. This is witnessed by the isomorphism $\phi \colon \GS \rightarrow \packets^{>=2}$ defined by:
    \[
	\phi(\oldctestA\, \gsA) \triangleq \oldctestA \cdot \phi(\gsA) \qquad\qquad \phi(\DUP\, \gsA) \triangleq \phi(\gsA)
    \]
    We can then see that for all $\gsA \in \GS$, $\wfatransext{\gsA}(\wfastartstate)(\wfaendstate) = \wfatransext{\phi(\gsA)}(\wfastartstate)(\wfaendstate)$. The base case is trivial:
    \[\wfatransext{\oldctestA\oldctestB}(\wfastartstate)(\wfaendstate) = \wfatransext{\oldctestA \cdot \phi(\oldctestB)}(\wfastartstate)(\wfaendstate) = \wfatransext{\phi(\oldctestA\oldctestB)}(\wfastartstate)(\wfaendstate)\]
    For the inductive case, we have that:
    \begin{align*}
	\wfatransext{\oldctestA \oldctestB \, \DUP \, \gsA}(\wfastartstate)(\wfaendstate)
	&= \big( \wfatrans{\oldctestA} \times \wfatrans{\oldctestB} \times \wfatransext{\gsA} \big)(\wfastartstate)(\wfaendstate) \tag{by definition of $\wfatransext{}$} \\
	&= \wfatransext{\oldctestA\oldctestB\gsA}(\wfastartstate)(\wfaendstate) \tag{by definition of $\wfatransext{}$} \\
	&= \wfatransext{\phi(\oldctestA\oldctestB \, \DUP \, \gsA)}(\wfastartstate)(\wfaendstate) \tag{by definition of $\phi$}
    \end{align*}

    An equivalent argument works to show the inverse direction. As such we can form a bijection between the terms of sums over these sets, preserving $\wfatransext{\gsA}(\wfastartstate)(\wfaendstate) = \wfainit \times \wfatransext{\gsA} \times \wfaout$, and conclude:
    \[
	\semisum{\gsA \in \GS} \wfainit \times \wfatransext{\gsA} \times \wfaout  \eeq \semisum{\gsA \in \packets^{>=2}} \wfainit \times \wfatransext{\gsA} \times \wfaout \eeq \semisum{\gsA \in \packets^*} \wfainit \times \wfatransext{\gsA} \times \wfaout
    \]
\end{proof}

\subsection{Proof of Decidability of $\wta$-safety for \wnetkat Policies}
\label{pf:verif-safety}

\repeattheorem{thm:verif-safety}
\begin{proof}
    By \Cref{cor:deno-aut}, we have that $\pola$ is $\wta$-safe
    if and only if $\forall \gsA \in \GS \colon \sem{\wnka_\pola}(\gsA) \semiord \wta$, i.e.\ :
    \[
    \begin{array}{rl}
        \wnka_\pola \in \autsafe{\wta} & (\text{\Cref{def:dec-problems}.\ref{def:dec-problems1}})~.
    \end{array}
    \]
    Further, this property is decidable at the level of \wnetkat automata
    (and we can provide a witness if it does not hold) by \Cref{thm:r-safe-dec},
    which we go over next.
\end{proof}

\begin{lemma}[Decidability of $\wta$-safety for \wnetkat automata]
\label{thm:r-safe-dec}
	Given a computable $\omega$-continuous semiring $(\semi,\, \semiord)$,
		if
		for all $\semia_1,\semia_2,\semia_3\in\semi$, we have
        $
		\semia_1 \semiadd \semia_2 \semiord \semia_3
        $
		iff
        $
		(\semia_1 \semiord \semia_3 ~\text{and}~\semia_2\semiord\semia_3),
        $
		then
        \begin{center}
            \enquote{Given $\semithresh \in \semi$ and \WNKA $\wnka$, does $\wnka \in \autsafe{\semithresh}$ hold?}
        \end{center}
        is decidable.
        Moreover, if $\semiord$ is total and $\wnka \not\in \autsafe{\semithresh}$, there is an effectively
        constructible witness, i.e., some $\gsA \in \GS$ such that $\wlang{\wnka}(\gsA) \semiordrevstrict \semithresh$.
\end{lemma}
\begin{proof}
	We proceed in two steps. First, we show that
	\[
		\wnka \in \autsafe{\semithresh}\qquad\text{iff}\qquad \semisum{\gsA\in\GS} \wlang{\wnka}(\gsA) \semiord \semithresh~.
	\]
	Second, we show that $\semisum{\gsA\in\GS} \wlang{\wnka}(\gsA)$ is computable, which implies the
    decidability claim by decidability of $\semiord$. Lastly, we will prove the claim on the
    effective constructibility of witnesses.

	For the first claim, consider the following:
	\begin{align*}
		& \wnka \in \autsafe{\semithresh} \\
        \text{iff}\quad& \forall \gsA \in \GS \colon \wlang{\wnka}(\gsA) \semiord \semithresh
        \tag{\Cref{def:dec-problems}}\\
        \text{iff}\quad &
            \forall n \in \nats \colon \forall \gsA \in \GS^{\leq n} \colon
                \wlang{\wnka}(\gsA) \semiord \semithresh
		\tag{$\GS^{\leq n}$ is the set of guarded strings with at most $n$ $\DUP$s}\\
        \text{iff}\quad & \forall n \in \nats \colon
            \semisum{\gsA \in \GS^{\leq n}} \wlang{\wnka}(\gsA) \semiord \semithresh
		\tag{assumption} \\
        \text{iff}\quad&
            \cposup[n \in \nats]
                \semisum{\gsA \in \GS^{\leq n}} \wlang{\wnka}(\gsA) \semiord \semithresh
		\tag{definition of suprema} \\
        \text{iff}\quad & \semisum{\gsA \in \GS} \wlang{\wnka}(\gsA) \semiord \semithresh~.
		\tag{definition of countably infinite sums}
	\end{align*}

		For the second claim, let $M = \semisum{\oldctestA \in \packets} \wfatrans{\oldctestA}$ and define the
    matrix closure of $M$ by  \[M^* = \semisum{i \in \nats} M^i~,\] which is computable \cite{bloom1993}. We show that
	\[
		\semisum{\gsA\in\GS} \sem{\wnka}(\gsA) = \wfainit \times M^* \times \wfaout~,
	\]
	which implies the claim. To see this, consider the following:
	\begin{align*}
		 & \semisum{\gsA\in\GS} \sem{\wnka}(\gsA)  \\
		 {}={}~ &\semisum{\gsA\in\GS}   \altsem{\wnka}(\gsA)
		 \tag{\Cref{lmm:unf-sem}}\\
		 {}={}~ &\semisum{\gsA\in\GS}   \wfainit \times \wfatransext{\gsA} \times \wfaout
		 \tag{definition}\\
		 {}={}~ &\semisum{\gsA\in\packets^*}   \wfainit \times \wfatransext{\gsA} \times \wfaout
		 \tag{\Cref{lmm:unf-gs-pk}}\\
		 {}={}~ &   \wfainit \times \big( \semisum{\gsA\in\packets^*}\wfatransext{\gsA} \big) \times \wfaout
		 \tag{distributivity of $\times$} \\
		 {}={}~ &   \wfainit \times \big( M^* \big) \times \wfaout
		 \tag{see below} ~.
	\end{align*}
	For the latter equality, consider the following:
	\begin{align*}
		 & \semisum{\gsA\in\packets^*}\wfatransext{\gsA}
		 ~{}={}~
		 \semisum{i\in\nats}  \semisum{\gsA\in\packets^i}\wfatransext{\gsA}
		 ~{}={}~  \semisum{i\in\nats}  M^i~,
	\end{align*}
	where $\semisum{\gsA\in\packets^i}\wfatransext{\gsA} =  M^i$ follows by induction on $i$. For $i=0$, we have
	\[
		\semisum{\gsA\in\packets^i}\wfatransext{\gsA}
		~{}={}~
		\wfatransext{\eword}
		~{}={}~
		\idmatrix_{\wfastates \times \wfastates}
		~{}={}~
		M^0~.
	\]
	For the induction step, we have
	\begin{align*}
		&\semisum{\gsA \in\packets^{i+1}}\wfatransext{\gsA}\\
		~{}={}~ &
	    \semisum{\oldctestA \in \packets} \semisum{\gsA \in \packets^i}\wfatransext{\oldctestA\gsA}
		\tag{decompose words of length $i+1$}\\
		~{}={}~ &
		\semisum{\oldctestA \in \packets} \semisum{\gsA \in \packets^i} \wfatrans{\oldctestA} \times \wfatransext{\gsA}
		\tag{definition}\\
		~{}={}~ &
		\semisum{\oldctestA \in \packets} \wfatrans{\oldctestA} \times \semisum{\gsA \in \packets^{i}} \wfatransext{\gsA}
		\tag{distributivity of $\times$}\\
		~{}={}~ &
	    \semisum{\oldctestA \in \packets} \wfatrans{\oldctestA} \times M^i
		\tag{I.H.}\\
		~{}={}~ &
		\big( \semisum{\oldctestA \in \packets} \wfatrans{\oldctestA} \big) \times M^i
		\tag{distributivity of $\times$} \\
		~{}={}~ &
		M \times M^i ~{}={} M^{i+1}~.
		\tag{definition}
	\end{align*}

	Let us now prove the effective costructibility of witnesses.
    Assume $\neg (\forall \gsA \in \GS \colon \wlang{\wnka}(\gsA) \semiord \semithresh)$. We have
	\begin{align*}
		& \neg (\forall \gsA \in \GS \wlang{\wnka}(\gsA) \semiord \semithresh) \\
		\text{iff} \quad &  \neg ( \semisum{\gsA\in\GS} \wlang{\wnka}(\gsA) \semiord \semithresh)
		\tag{see above }\\
		\text{iff} \quad &   \semisum{\gsA\in\GS} \wlang{\wnka}(\gsA) \semiordrevstrict \semithresh
		\tag{$\semiord$ is total}\\
		\text{iff} \quad &
            \cposup[n \in \nats]
                \semisum{\gsA \in \GS{\leq n}}
                    \wlang{\wnka}(\gsA) \semiordrevstrict \semithresh
		\tag{definition of countable sums}\\
		\text{iff} \quad &
            \exists n \in \nats\colon
                \semisum{\gsA \in \GS^{\leq n}}
                    \wlang{\wnka}(\gsA) \semiordrevstrict \semithresh
		\tag{definition of suprema}\\
		\text{iff} \quad &
            \exists n \in \nats\colon
                \exists \gsA \in \GS^{\leq n} \colon
                    \wlang{\wnka}(\gsA) \semiordrevstrict \semithresh
		\tag{assumption}\\
		\text{iff} \quad &
            \exists \gsA \in \GS \colon
                \wlang{\wnka}(\gsA) \semiordrevstrict \semithresh~.
		\tag{assumption}
	\end{align*}
	We can thus effectively construct a witness $\gsA$ with
    $\wlang{\wnka}(\gsA) \semiordrevstrict \semithresh$ by enumerating
    all guarded strings $\gsA$, computing $\wlang{\wnka}(\gsA)$, and deciding
    $\wlang{\wnka}(\gsA) \semiordrevstrict \semithresh$. Since such an $\gsA$
    exists, this procedure terminates.
\end{proof}

\subsection{Runs of a \wnetkat Automaton}
\label{appendix:runs-wnka}
In this section, we give another alternate formulation of the weighted
language of guarded strings recognized by a \wnetkat automaton. This
alternate definition will be useful for proving the decidability of
$\semithresh$-reachability: we describe the underlying ``graph'' structure of
the \wnetkat automaton and show that reachability can be decided by considering
only finitely many paths through the graph.

\begin{definition}[Runs of a \wnetkat automaton]
\label{def:wnka-runs}
Given a \wnetkat automaton $\wnka$, a \emph{run} $\wrun$ is a string
describing a path in the automaton. We denote by $\wnkapath$ the set
    $(\wnkastates \times (\oldTest \times \oldTest) \times \wnkastates)^* \cdot ((\oldTest \times \oldTest) \times \wnkastates)$
of all potential runs through the automaton.

Given states $\wnkastatea$ and $\wnkastateb$, the runs
$\wruns{\wnkastatea}{\wnkastateb} \colon \GS \to 2^{\wnkapath}$ from
$\wnkastatea$ to $\wnkastateb$ are defined recursively as
\begin{align*}
    \wruns{\wnkastatea}{\wnkastateb}(\oldctestA \SEQN \oldctestB) =&
	\begin{cases}
        \{ ((\oldctestA, \oldctestB), \wnkastatea) \} & \text{if $\wnkastatea=\wnkastateb$} \\
		\emptyset & \text{otherwise}
	\end{cases} \\
    \wruns{\wnkastatea}{\wnkastateb}(\oldctestA \SEQN \oldctestB \SEQN \gsB)
    =&
    \bigcup_{\{ \wnkastatec | \wnkatrans{\oldctestA}{\oldctestB}(\wnkastatea)(\wnkastatec) \neq \addid \}}
    \concatword{(\wnkastatea,(\oldctestA,\oldctestB),\wnkastatec)}
               {\wruns{\wnkastatec}{\wnkastateb}({\oldctestB \SEQN \gsB})}~.
\end{align*}
\end{definition}

\begin{definition}[Weight of a run]
\label{def:weight-run}
The \emph{weight} of a given run $\wrun$ is defined recursively on
$\wrun$ as:
\[
\begin{array}{rcl}
    \runweight(((\oldctestA,\oldctestB),\wnkastatea)) &=& \mulid \\
    \runweight(\concatword{(\wnkastatea, (\oldctestA,\oldctestB), \wnkastateb)}{\wrun}) &=&
        \wnkatrans{\oldctestA}{\oldctestB}(\wnkastatea)(\wnkastateb) \semimul \runweight(\wrun)
    ~,
\end{array}
\]
so that, crucially, we have for every $\gsA \in \GS$:\footnote{
    Recall that $\{\wnkatransext{\gsA}\} \colon \GS \to \weightings{\GS}$
    is the extended family of transition functions defined in
    \Cref{def:extended-fam}.\ref{def:extended-trans}.
}
\[
    \wnkatransext{\gsA}(\wnkastatea)(\wnkastateb) =
	\semisum{\wrun \in \wruns{\wnkastatea}{\wnkastateb}(\gsA)}\runweight(\wrun)~.
\]
\end{definition}
We denote by $\wrunscfree{\wnkastatea}{\wnkastateb}(\gsA)$ the set of \emph{cycle-free}
runs from $\wnkastatea$ to $\wnkastateb$ on $\gsA$. For the purposes of a \wnetkat
automaton, a \emph{cycle} is defined as a run where not only is a target state
$\wnkastateb$ repeated but the carry-on packet at the first occurrence of the target
state $\wnkastateb$ in the run corresponds with the carry-on packet of the next occurrence
of the target state $\wnkastateb$ in the run, i.e. a step
$(\wnkastatea,(\oldctestA,\oldctestB),\wnkastateb)$ in the run followed by another step
$(\wnkastatec,(\oldctestC,\oldctestB),\wnkastateb)$ or $((\oldctestC,\oldctestB),\wnkastateb)$.
Notice that, for any such run, we can consider a shorter run without this cycle
(which is relevant in the case where
$\forall \semia,\semia' \in \semi \colon \semia \semiordrev \semia \semimul \semia'$).
Additionally, the set of cycle-free runs between any two states is finite.
Finally, $\wruns{\wnkastatea}{\wnkastateb}$
(resp.\ $\wrunscfree{\wnkastatea}{\wnkastateb}$), denotes the set of \emph{all}
(cycle-free) runs from $\wnkastatea$ to $\wnkastateb$.

\subsection{Proof of Decidability of $\wta$-reachability for \wnetkat Policies}
\label{pf:verif-reach}

\repeattheorem{thm:verif-reach}
\begin{proof}
    By \Cref{cor:deno-aut}, $\pola$ is $\wta$-reachable if and only if
    $\exists \gsA \in \GS \colon \sem{\wnka}(\pola) \semiordrev \wta$, i.e.\ :
    \[
    \begin{array}{rl}
        \wnka_\pola \in \autreach{\wta} & \text{(\Cref{def:dec-problems}.\ref{def:dec-problems2})}~.
    \end{array}
    \]
    Further, this property is decidable at the level of \wnetkat automata
    (and we can provide a witness if it is satisfied) by \Cref{thm:r-reachable-dec},
    which we go over next.
\end{proof}

\begin{lemma}[Decidability of $\semithresh$-reachability for \wnetkat automata]
\label{thm:r-reachable-dec}
	Given a computable $\omega$-continuous semiring $(\semi,\, \semiord)$,
	if for all $\semia_1,\semia_2,\semia_3$, we have (i)
    $
	\semia_1 \semiadd \semia_2 \semiordrev \semia_3
	$
    iff
    ($\semia_1 \semiordrev \semia_3$
    or
    $\semia_2 \semiordrev \semia_3$), and (ii) $\semia_1 \semiordrev \semia_1 \semimul \semia_2$,
	then
    \begin{center}
	\enquote{Given $\semithresh \in \semi$ and a \WNKA $\wnka$,
        does $\wnka \in \autreach{\semithresh}$ hold?}
    \end{center}
	is decidable. Moreover, if $\wnka \in \autreach{\semithresh}$,
    there is an effectively constructible witness, i.e., some
    $\gsA \in \GS$ such that
    $\wlang{\wnka}(\gsA) \semiordrev \semithresh$.
\end{lemma}
\begin{proof}
    We prove that $\semithresh$-reachability is decidable by showing that we
    can consider only the cycle-free \emph{runs} (see
    \Cref{def:wnka-runs}) of the \wnetkat automaton.
	We have:
       \begin{align*}
               & \wnka \in \autreach{\semithresh} \\
               \text{iff}\quad&
            \exists \gsA \in \GS \colon
                \wlang{\wnka}(\gsA) \semiordrev \semithresh
        \tag{\Cref{def:dec-problems}} \\
               \text{iff}\quad&
            \exists \gsA \in \GS \colon
                \wnkainit \times \wnkatransext{\gsA} \times \wnkaout{}{} \semiordrev \semithresh
        \tag{By definition, ($\star$)}  \\
               \text{iff} \quad &
            \exists \gsA \in \GS \colon
                \semisum{\wnkastatea,\wnkastateb \in \wnkastates}
                    \semisum{\wrun \in \wruns{\wnkastatea}{\wnkastateb}(\gsA)}
                        \wnkainit(\wnkastatea) \semimul
                            \runweight(\wrun) \semimul
                                \wnkaout{}{}(\wnkastateb)
                                    \semiordrev \semithresh
               \tag{\Cref{def:weight-run}}
               \\
               \text{iff} \quad &
            \exists \gsA \in \GS \colon
                \exists \wnkastatea,\wnkastateb \in \wnkastates \colon
                    \exists \wrun \in \wruns{\wnkastatea}{\wnkastateb}(\gsA) \colon
                        \wnkainit(\wnkastatea) \semimul
                            \runweight(\wrun) \semimul
                                \wnkaout{}{}(\wnkastateb) \semiordrev \semithresh
               \tag{Assumption} \\
               \text{iff} \quad &
            \exists \gsA \in \GS \colon
                \exists \wnkastatea,\wnkastateb \in \wnkastates \colon
                    \exists \wrun \in \wrunscfree{\wnkastatea}{\wnkastateb}(\gsA) \colon
                        \wnkainit(\wnkastatea) \semimul
                            \runweight(\wrun) \semimul
                                \wnkaout{}{}(\wnkastateb) \semiordrev \semithresh
               \tag{Cycle-free runs suffice since $\semia_1 \semiordrev \semia_1 \semimul \semia_2$} \\
               \text{iff}\quad &
                \exists \wnkastatea,\wnkastateb \in \wnkastates \colon
                    \exists \wrun \in \wrunscfree{\wnkastatea}{\wnkastateb} \colon
                        \wnkainit(\wnkastatea) \semimul
                            \runweight(\wrun) \semimul
                                \wnkaout{}{}(\wnkastateb) \semiordrev \semithresh
               \tag{Every run corresponds to a guarded string, ($\star\star$)}
       \end{align*}
	which can be decided by considering all of the (finitely many) cycle-free runs from some
    $\wnkastatea\in\wnkastates$ to some $\wnkastateb\in\wnkastates$. If a corresponding run exists,
    the sought-after $\gsA \in \GS$ can be read off that run.
    Additionally, for the above proof we make the following remarks:\\
    $(\star)$: We write $\wnkaout{}{}$ for $\wnkaout{\pkta}{\pktb}$, where the
    guarded string $\gsA \in \GS$ ends in $\pkta\pktb$.\\
    $(\star\star)$: Note that for all $\wrun \in \wruns{\wnkastatea}{\wnkastateb}(\gsA)$,
    we have that $\wrun$ ends in $((\pkta,\pktb),\wnkastateb)$, so $\wnkaout{}{}$ remains
    well-defined.
\end{proof}

\section{Case Studies}
\label{appendix:case-studies}

We provide here the full details of the example \wnetkat policies given in \Cref{sec:case-studies} for the network in \Cref{fig:abilene}. We assume only the following fields: $\fieldnode$ (packet's current location), $\fielddst$ (packet's destination), $\fieldtid$ (packet's current tunnel), and $\fieldvid$ (packet's type: video or not).

\subsection{\wnetkat Policy for Tunneled Paths}

As discussed in \Cref{sec:case-studies}, certain source-destination pairs are configured in the network to use tunnels instead of following their usual forwarding logic. We leave the default forwarding logic for each node unspecified---assuming we have, e.g., $\pol{\atlanta,\defaultstr}$ for $\atlanta$---and specify below only the complete \emph{tunneling logic} of every node (as this is what is relevant for our verification questions).

\begin{longtable}{rcl}
    $\pol{\atlanta}$ &$\triangleq$&$
        \IFN \ \EQ{\fieldtid}{0} \ \THENN $\\&&$ \quad
            \IFN \ \EQ{\fielddst}{\nyc} \ \THENN \
                \ASSN{\fieldtid}{5} \
            \ELSEN \ \pol{\atlanta,\defaultstr} $\\&&$
        \ELSEN \ \IFN \ \OR{\EQ{\fieldtid}{3}}{\EQ{\fieldtid}{4}} \ \THENN \ \ASSN{\fieldtid}{0} $\\&&$
        \ELSEN \ \IFN \ \EQ{\fieldtid}{5} \ \THENN \
            \ASSN{\fieldnode}{\dc} $\\&&$
        \ELSEN \ \DROP
        $ \medskip \\
    
    $\pol{\bay}$ &$\triangleq$&$
        \IFN \ \EQ{\fieldtid}{0} \ \THENN $\\&&$ \quad
            \IFN \ \EQ{\fielddst}{\nyc} \ \THENN \
                (\ADD{\ASSN{\fieldtid}{1}}{\ASSN{\fieldtid}{2}}) \
            \ELSEN \ \pol{\bay,\defaultstr} $\\&&$
        \ELSEN \ \IFN \ \EQ{\fieldtid}{1} \ \THENN \
            \ASSN{\fieldnode}{\denver} $\\&&$
        \ELSEN \ \IFN \ \EQ{\fieldtid}{2} \ \THENN \
            \ASSN{\fieldnode}{\la} $\\&&$
        \ELSEN \ \DROP
        $ \medskip \\

    $\pol{\chicago}$ &$\triangleq$& $\pol{\chicago,\defaultstr}$

    \medskip \\

    $\pol{\dc}$ &$\triangleq$&$
        \IFN \ \EQ{\fieldtid}{0} \ \THENN \
            \pol{\dc,\defaultstr} $\\&&$
        \ELSEN \ \IFN \ \EQ{\fieldtid}{5} \ \THENN \
            \ASSN{\fieldnode}{\nyc} $\\&&$
        \ELSEN \ \DROP
        $ \medskip \\

    $\pol{\denver}$ &$\triangleq$&$
        \IFN \ \EQ{\fieldtid}{0} \ \THENN \
            \pol{\denver,\defaultstr} $\\&&$
        \ELSEN \ \IFN \ \EQ{\fieldtid}{1} \ \THENN \
            \ASSN{\fieldnode}{\kansas} $\\&&$
        \ELSEN \ \DROP
        $ \medskip \\

    $\pol{\houston}$ &$\triangleq$&$
        \IFN \ \EQ{\fieldtid}{0} \ \THENN \
            \pol{\houston,\defaultstr} $\\&&$
        \ELSEN \ \IFN \ \EQ{\fieldtid}{2} \ \THENN \
            \ASSN{\fieldnode}{\kansas} $\\&&$
        \ELSEN \ \IFN \ \EQ{\fieldtid}{3} \ \THENN \
            \ASSN{\fieldnode}{\atlanta} $\\&&$
        \ELSEN \ \DROP
        $ \medskip \\

    $\pol{\indiana}$ &$\triangleq$&$
        \IFN \ \EQ{\fieldtid}{0} \ \THENN \
            \pol{\indiana,\defaultstr} $\\&&$
        \ELSEN \ \IFN \ \EQ{\fieldtid}{4} \ \THENN \
            \ASSN{\fieldnode}{\atlanta} $\\&&$
        \ELSEN \ \DROP
        $ \medskip \\

    $\pol{\kansas}$ &$\triangleq$&$
        \IFN \ \EQ{\fieldtid}{0} \ \THENN $\\&&$ \quad
            \IFN \ \EQ{\fielddst}{\nyc} \ \THENN \
                (\ADD{\ASSN{\fieldtid}{3}}{\ASSN{\fieldtid}{4}}) \
            \ELSEN \ \pol{\kansas,\defaultstr} $\\&&$
        \ELSEN \ \IFN \ \OR{\EQ{\fieldtid}{1}}{\EQ{\fieldtid}{2}} \ \THENN \ \ASSN{\fieldtid}{0} $\\&&$
        \ELSEN \ \IFN \ \EQ{\fieldtid}{3} \ \THENN \
            \ASSN{\fieldnode}{\houston} $\\&&$
        \ELSEN \ \IFN \ \EQ{\fieldtid}{4} \ \THENN \
            \ASSN{\fieldnode}{\indiana} $\\&&$
        \ELSEN \ \DROP
        $ \medskip \\

    $\pol{\la}$ &$\triangleq$&$
        \IFN \ \EQ{\fieldtid}{0} \ \THENN \
            \pol{\la,\defaultstr} $\\&&$
        \ELSEN \ \IFN \ \EQ{\fieldtid}{2} \ \THENN \
            \ASSN{\fieldnode}{\houston} $\\&&$
        \ELSEN \ \DROP
        $ \medskip \\

    $\pol{\nyc}$ &$\triangleq$& $\pol{\nyc,\defaultstr}$

    \medskip \\

    $\pol{\seattle}$ &$\triangleq$& $\pol{\seattle,\defaultstr}$
\end{longtable}

All of these policies are put together to represent the entire Abilene network. We weight each (iterated) policy by its \textcolor{failureorange}{failure rate} to check that its tunneled paths have a failure rate of at most 10\%.
\begin{longtable}{rcl}
    $\pol{\relstr}$ &$\triangleq$&$
        \IFN \ \EQ{\fieldnode}{\atlanta} \ \THENN
    $ \\ && \quad $
            \WEIGH{\failure{1.5}}{\ITER{(\pol{\atlanta})} \SEQN \NOTEQ{\fieldnode}{\atlanta}}
    $ \\ && $
        \ELSEN \ \IFN \ \EQ{\fieldnode}{\bay} \ \THENN 
    $ \\ && \quad $
            \WEIGH{\failure{1}}{\ITER{(\pol{\bay})} \SEQN \NOTEQ{\fieldnode}{\bay}}
    $ \\ && $
      \ELSEN \ \IFN \ \EQ{\fieldnode}{\chicago} \ \THENN
    $ \\ && \quad $
            \WEIGH{\failure{1}}{\ITER{(\pol{\chicago})} \SEQN \NOTEQ{\fieldnode}{\chicago}}
    $ \\ && $
      \ELSEN \ \IFN \ \EQ{\fieldnode}{\dc} \ \THENN
    $ \\ && \quad $
            \WEIGH{\failure{1.75}}{\ITER{(\pol{\dc})} \SEQN \NOTEQ{\fieldnode}{\dc}}
    $ \\ && $
      \ELSEN \ \IFN \ \EQ{\fieldnode}{\denver} \ \THENN
    $ \\ && \quad $
            \WEIGH{\failure{1}}{\ITER{(\pol{\denver})} \SEQN \NOTEQ{\fieldnode}{\denver}}
    $ \\ && $
      \ELSEN \ \IFN \ \EQ{\fieldnode}{\houston} \ \THENN
    $ \\ && \quad $
            \WEIGH{\failure{1.75}}{\ITER{(\pol{\houston})}\SEQN \NOTEQ{\fieldnode}{\houston}}
    $ \\ && $
      \ELSEN \ \IFN \ \EQ{\fieldnode}{\indiana} \ \THENN
    $ \\ && \quad $
            \WEIGH{\failure{1.25}}{\ITER{(\pol{\indiana})} \SEQN \NOTEQ{\fieldnode}{\indiana}}
    $ \\ && $
      \ELSEN \ \IFN \ \EQ{\fieldnode}{\kansas} \ \THENN
    $ \\ && \quad $
            \WEIGH{\failure{1.5}}{\ITER{(\pol{\kansas})} \SEQN \NOTEQ{\fieldnode}{\kansas}}
    $ \\ && $
      \ELSEN \ \IFN \ \EQ{\fieldnode}{\la} \ \THENN
    $ \\ && \quad $
            \WEIGH{\failure{1.5}}{\ITER{(\pol{\la})} \SEQN \NOTEQ{\fieldnode}{\la}}
    $ \\ && $
      \ELSEN \ \IFN \ \EQ{\fieldnode}{\nyc} \ \THENN
    $ \\ && \quad $
            \WEIGH{\failure{0.5}}{\ITER{(\pol{\nyc})} \SEQN \NOTEQ{\fieldnode}{\nyc}}
    $ \\ && $
      \ELSEN \ \IFN \ \EQ{\fieldnode}{\seattle} \ \THENN
    $ \\ && \quad $
            \WEIGH{\failure{0.75}}{\ITER{(\pol{\seattle})} \SEQN \NOTEQ{\fieldnode}{\seattle}}
    $ \medskip \\ $
    \abilene_{\relstr} $&$\triangleq$&$ \ITER{(\SEQ{\pol{\relstr}}{\DUP})}
    $
\end{longtable}

We verify that all tunneled paths between $\bay$ and $\nyc$ have a failure rate of at most $10\%$ by checking that:
\[
    \SEQ{(\AND{\EQ{\fieldnode}{\bay}}{\EQ{\fielddst}{\nyc}})}{
        \SEQ{\abilene_{\relstr}}
        {(\AND{\EQ{\fieldnode}{\nyc}}{\NOTEQ{\fieldtid}{0}})}}
\]
is 0.1-safe. As discussed in \Cref{sec:case-studies}, this check would fail. The complete fixed policy for $\kansas$ so that all tunneled paths have a failure rate of at most $10\%$ is as follows:
\[
\begin{array}{rcl}
    \pol{\kansas,\mathsf{safe}} &\triangleq&
        \IFN \ \EQ{\fieldtid}{0} \ \THENN \\&& \quad
            \IFN \ \EQ{\fielddst}{\nyc} \ \THENN \
                (\ADD{\ASSN{\fieldtid}{3}}{\ASSN{\fieldtid}{4}}) \
            \ELSEN \ \pol{\kansas,\defaultstr} \\&&
        \ELSEN \ \IFN \ \AND{\EQ{\fieldtid}{2}}{\EQ{\fielddst}{\nyc}} \ \THENN \
            \ASSN{\fieldtid}{4} \\&&
        \ELSEN \ \IFN \ \OR{\EQ{\fieldtid}{1}}{\EQ{\fieldtid}{2}} \ \THENN \
        \ASSN{\fieldtid}{0} \\&&
        \ELSEN \ \IFN \ \EQ{\fieldtid}{3} \ \THENN \ \ASSN{\fieldnode}{\houston} \\&&
        \ELSEN \ \IFN \ \EQ{\fieldtid}{4} \ \THENN \ \ASSN{\fieldnode}{\indiana} \\&&
        \ELSEN \ \DROP
\end{array}
\]

\subsection{\wnetkat Policy for High-Bandwidth Tunneled Paths}

After reconfiguring the network so tunneled paths guarantee a failure rate of at most $10\%$, we show how to leverage \Cref{thm:verif-reach} to find a tunneled path with a bandwidth of at least $1000\mbps$. In particular, the policies in the previous section are modified so that the forwarding actions \emph{within tunnels} are weighted by bandwidth. We show only the policies that are modified:
\begin{longtable}{rcl}
    $\pol{\atlanta,\bandstr}$ &$\triangleq$&$
        \IFN \ \EQ{\fieldtid}{0} \ \THENN $\\&&$ \quad
            \IFN \ \EQ{\fielddst}{\nyc} \ \THENN \
                \ASSN{\fieldtid}{5} \
            \ELSEN \ \pol{\atlanta,\defaultstr} $\\&&$
        \ELSEN \ \IFN \ \OR{\EQ{\fieldtid}{3}}{\EQ{\fieldtid}{4}} \ \THENN \ \ASSN{\fieldtid}{0} $\\&&$
        \ELSEN \ \IFN \ \EQ{\fieldtid}{5} \ \THENN \
            \WEIGH{\bandwidth{1750}}{\ASSN{\fieldnode}{\dc}} $\\&&$
        \ELSEN \ \DROP
        $ \medskip \\
    
    $\pol{\bay,\bandstr}$ &$\triangleq$&$
        \IFN \ \EQ{\fieldtid}{0} \ \THENN $\\&&$ \quad
            \IFN \ \EQ{\fielddst}{\nyc} \ \THENN \
                (\ADD{\ASSN{\fieldtid}{1}}{\ASSN{\fieldtid}{2}}) \
            \ELSEN \ \pol{\bay,\defaultstr} $\\&&$
        \ELSEN \ \IFN \ \EQ{\fieldtid}{1} \ \THENN \
            \WEIGH{\bandwidth{1500}}{\ASSN{\fieldnode}{\denver}} $\\&&$
        \ELSEN \ \IFN \ \EQ{\fieldtid}{2} \ \THENN \
            \WEIGH{\bandwidth{1000}}{\ASSN{\fieldnode}{\la}} $\\&&$
        \ELSEN \ \DROP
        $ \medskip \\

    $\pol{\dc,\bandstr}$ &$\triangleq$&$
        \IFN \ \EQ{\fieldtid}{0} \ \THENN \
            \pol{\dc,\defaultstr} $\\&&$
        \ELSEN \ \IFN \ \EQ{\fieldtid}{5} \ \THENN \
            \WEIGH{\bandwidth{1500}}{\ASSN{\fieldnode}{\nyc}} $\\&&$
        \ELSEN \ \DROP
        $ \medskip \\

    $\pol{\denver,\bandstr}$ &$\triangleq$&$
        \IFN \ \EQ{\fieldtid}{0} \ \THENN \
            \pol{\denver,\defaultstr} $\\&&$
        \ELSEN \ \IFN \ \EQ{\fieldtid}{1} \ \THENN \
            \WEIGH{\bandwidth{1250}}{\ASSN{\fieldnode}{\kansas}} $\\&&$
        \ELSEN \ \DROP
        $ \medskip \\

    $\pol{\houston,\bandstr}$ &$\triangleq$&$
        \IFN \ \EQ{\fieldtid}{0} \ \THENN \
            \pol{\houston,\defaultstr} $\\&&$
        \ELSEN \ \IFN \ \EQ{\fieldtid}{2} \ \THENN \
            \WEIGH{\bandwidth{1250}}{\ASSN{\fieldnode}{\kansas}} $\\&&$
        \ELSEN \ \IFN \ \EQ{\fieldtid}{3} \ \THENN \
            \WEIGH{\bandwidth{1750}}{\ASSN{\fieldnode}{\atlanta}} $\\&&$
        \ELSEN \ \DROP
        $ \medskip \\

    $\pol{\indiana,\bandstr}$ &$\triangleq$&$
        \IFN \ \EQ{\fieldtid}{0} \ \THENN \
            \pol{\indiana,\defaultstr} $\\&&$
        \ELSEN \ \IFN \ \EQ{\fieldtid}{4} \ \THENN \
            \WEIGH{\bandwidth{950}}{\ASSN{\fieldnode}{\atlanta}} $\\&&$
        \ELSEN \ \DROP
        $ \medskip \\

    $\pol{\kansas,\bandstr}$ &$\triangleq$&$
        \IFN \ \EQ{\fieldtid}{0} \ \THENN $\\&&$ \quad
            \IFN \ \EQ{\fielddst}{\nyc} \ \THENN \
                (\ADD{\ASSN{\fieldtid}{3}}{\ASSN{\fieldtid}{4}}) \
            \ELSEN \ \pol{\kansas,\defaultstr} $\\&&$
        \ELSEN \ \IFN \ \AND{\EQ{\fieldtid}{2}}{\EQ{\fielddst}{\nyc}} \ \THENN \
            \ASSN{\fieldtid}{4} $\\&&$
        \ELSEN \ \IFN \ \OR{\EQ{\fieldtid}{1}}{\EQ{\fieldtid}{2}} \ \THENN \
        \ASSN{\fieldtid}{0} $\\&&$
        \ELSEN \ \IFN \ \EQ{\fieldtid}{3} \ \THENN \ \WEIGH{\bandwidth{1250}}{\ASSN{\fieldnode}{\houston}} $\\&&$
        \ELSEN \ \IFN \ \EQ{\fieldtid}{4} \ \THENN \ \WEIGH{\bandwidth{1750}}{\ASSN{\fieldnode}{\indiana}} $\\&&$
        \ELSEN \ \DROP
        $ \medskip \\

    $\pol{\la,\bandstr}$ &$\triangleq$&$
        \IFN \ \EQ{\fieldtid}{0} \ \THENN \
            \pol{\la,\defaultstr} $\\&&$
        \ELSEN \ \IFN \ \EQ{\fieldtid}{2} \ \THENN \
            \WEIGH{\bandwidth{1500}}{\ASSN{\fieldnode}{\houston}} $\\&&$
        \ELSEN \ \DROP
        $
\end{longtable}

With these policies, we verify in \Cref{sec:case-studies} that the tunneled path $1 \edge 3 \edge 5$ satisfies the desired bandwidth. We then modify the \wnetkat encoding to model the network being reconfigured to always forward $\nyc$-bound video traffic through this tunneled path. The only policies that change are those for $\kansas$ and $\bay$, their complete (reconfigured) policies are as follows:
\begin{longtable}{rcl}
    $\pol{\bay,\vidstr}$ &$\triangleq$&$
        \IFN \ \EQ{\fieldtid}{0} \ \THENN $\\&&$ \quad
            \IFN \ \EQ{\fielddst}{\nyc} \ \THENN $\\&&$ \qquad
                \IFN \ \EQ{\fieldvid}{\consttrue} \ \THENN \
                    \ASSN{\fieldtid}{1} \
                \ELSEN \
                    (\ADD{\ASSN{\fieldtid}{1}}{\ASSN{\fieldtid}{2}}) $\\&&$ \quad
            \ELSEN \ \pol{\bay,\defaultstr} $\\&&$
        \ELSEN \ \IFN \ \EQ{\fieldtid}{1} \ \THENN \
            \ASSN{\fieldnode}{\denver} $\\&&$
        \ELSEN \ \IFN \ \EQ{\fieldtid}{2} \ \THENN \
           \ASSN{\fieldnode}{\la} $\\&&$
        \ELSEN \ \DROP
        $ \medskip \\

    $\pol{\kansas,\vidstr}$ &$\triangleq$&$
        \IFN \ \EQ{\fieldtid}{0} \ \THENN $\\&&$ \quad
            \IFN \ \EQ{\fielddst}{\nyc} \ \THENN $\\&&$ \qquad
                \IFN \ \EQ{\fieldvid}{\consttrue} \ \THENN \
                    \ASSN{\fieldtid}{3} \
                \ELSEN \
                    (\ADD{\ASSN{\fieldtid}{3}}{\ASSN{\fieldtid}{4}}) $\\&&$ \quad
            \ELSEN \ \pol{\kansas,\defaultstr} $\\&&$
        \ELSEN \ \IFN \ \AND{\EQ{\fieldtid}{2}}{\EQ{\fielddst}{\nyc}} \ \THENN \
            \ASSN{\fieldtid}{4} $\\&&$
        \ELSEN \ \IFN \ \OR{\EQ{\fieldtid}{1}}{\EQ{\fieldtid}{2}} \ \THENN \
        \ASSN{\fieldtid}{0} $\\&&$
        \ELSEN \ \IFN \ \EQ{\fieldtid}{3} \ \THENN \ \ASSN{\fieldnode}{\houston} $\\&&$
        \ELSEN \ \IFN \ \EQ{\fieldtid}{4} \ \THENN \ \ASSN{\fieldnode}{\indiana} $\\&&$
        \ELSEN \ \DROP
    $
\end{longtable}

\fi
\end{document}